\documentclass[11pt]{article}
\usepackage{fullpage}
\usepackage{authblk}
\usepackage{enumitem}
\usepackage{afterpage}
\usepackage{caption}
\usepackage{subcaption}
\usepackage{cancel}
\usepackage{amsmath}
\usepackage{nicefrac}
\usepackage[round]{natbib}
\usepackage{hyperref}
\usepackage{thmtools}
\usepackage{thm-restate}
\usepackage[capitalise]{cleveref}
\usepackage{graphics}
\usepackage{psfrag}
\usepackage{placeins}

\usepackage{macros}

\usepackage[suppress]{color-edits}  
\addauthor{ka}{magenta}   
\addauthor{as}{blue} 
\addauthor{rev}{red} 

\hypersetup{colorlinks, citecolor=blue, filecolor=blue, linkcolor=blue, urlcolor=blue}
\usepackage{enumitem}

\newcommand{\dd}{\mathrm{d}}

\def \Dev {\term{DEV}}
\def \DevM {\ensuremath{\term{DEV}_{\max}}} 
\def \OBJ {\term{OBJ}}
\def \guessold {\guess_{\term{old}}}

\newcommand{\taualg}{\tau_{\text{alg}}}

\title{Adversarial Bandits with Knapsacks%
\footnote{An extended abstract is published in \emph{FOCS 2019}: 60th Annual IEEE Symposium on Foundations of Computer Science.
\newline\indent
The definitive version of this paper will be published in \emph{JACM}: Journal of the ACM.
\vspace{1mm}\newline\indent
\emph{Version history.} The conference version corresponds, as an extended abstract, to the March'19 version of this manuscript. A major revision in Oct'19 improved the competitive ratios in Sections~\ref{sec:adversarial} and~\ref{sec:HP}, reducing the dependence on $d$ and shaving off some constant factors. In particular, we streamlined some looseness in the algorithm in Section~\ref{sec:adversarial}, and made the final computation somewhat more efficient. Also, we made the lower bound statements more explicit, and expanded the discussion of open questions. The subsequent revisions fixed various inaccuracies in the proofs and added some follow-up work.
\vspace{1mm}\newline\indent
\emph{Acknowledgements.} We are grateful to Sahil Singla and Thomas Kesselheim for pointing out the reduction in Remark~\ref{rem:reduction}, and an inefficiency in our original analysis in Section~\ref{sec:adv-analysis}. We are also grateful to Omid Sadeghi and the JACM reviewers for pointing out several typos and inaccuracies. We thank Robert Kleinberg, Akshay Krishnamurthy, Steven Wu, and Chicheng Zhang for many insightful conversations on online machine learning.
\vspace{1mm}
}}
\author{Nicole Immorlica \thanks{Microsoft Research, Cambridge, MA. Email: nicimm@microsoft.com.}}
\author{Karthik A. Sankararaman \thanks{Facebook, Menlo Park, CA.  Email: karthikabinavs@gmail.com. The research was done while the author was a graduate student at University of Maryland (College Park, MD) and an intern at Microsoft Research (New York, NY). Supported in part by NSF Awards CNS 1010789, CCF 1422569, CCF-1749864 and research awards from Adobe, Amazon, and Google.}}
\author{Robert Schapire\thanks{Microsoft Research, New York, NY. Email: schapire@microsoft.com.}}
\author{Aleksandrs Slivkins\thanks{Microsoft Research, New York, NY. Email: slivkins@microsoft.com.}}
\affil{}

\date{First version: November 2018\\This version: July 2022}

\begin{document}

\maketitle

\begin{abstract}
We consider \emph{Bandits with Knapsacks} (henceforth, \emph{BwK}), a general model for multi-armed bandits under supply/budget constraints. In particular, a bandit algorithm needs to solve a well-known \emph{knapsack problem}: find an optimal packing of items into a limited-size knapsack. The \BwK problem is a common generalization of numerous motivating examples, which range from dynamic pricing to repeated auctions to dynamic ad allocation to network routing and scheduling. While the prior work on \BwK focused on the stochastic version, we pioneer the other extreme in which the outcomes can be chosen adversarially. This is a considerably harder problem, compared to both the stochastic version and the ``classic" adversarial bandits, in that regret minimization is no longer feasible. Instead, the objective is to minimize the \emph{competitive ratio}: the ratio of the benchmark reward to algorithm's reward.

We design an algorithm with competitive ratio $O(\log T)$ relative to the best fixed distribution over actions, where $T$ is the time horizon; we also prove a matching lower bound. The key conceptual contribution is a new perspective on the stochastic version of the problem. We suggest a new algorithm for the stochastic version, which builds on the framework of regret minimization in repeated games and admits a substantially simpler analysis compared to prior work. We then analyze this algorithm for the adversarial version, and use it as a subroutine to solve the latter.

Our algorithm is the first ``black-box reduction" from bandits to \BwK: it takes an arbitrary bandit algorithm and uses it as a subroutine. We use this reduction to derive several extensions.
\end{abstract}

\newpage
\tableofcontents
\newpage
\section{Introduction}
\label{sec:intro}
Multi-armed bandits is a simple abstraction for the tradeoff between \emph{exploration} and \emph{exploitation}, \ie between making potentially suboptimal decisions for the sake of acquiring new information and using this information for making better decisions. Studied over many decades, multi-armed bandits is a very active research area spanning computer science, operations research, and economics
\citep{CesaBL-book,Bergemann-survey06,Gittins-book11,Bubeck-survey12,slivkins-MABbook,LS19bandit-book}.

In this paper, we focus on bandit problems which feature supply or budget constraints, as is the case in many realistic applications.
For example, a seller who experiments with prices may have a limited inventory, and a website optimizing ad placement may be constrained by the advertisers' budgets.
This general problem is called
\emph{Bandits with Knapsacks} (\emph{BwK}) since, in this model, a bandit algorithm needs effectively to solve a \emph{knapsack problem}
(find an optimal packing of items into a limited-size knapsack) or generalization thereof.
The \BwK model was introduced in \cite{BwK-focs13} as a common generalization of numerous motivating examples, ranging from dynamic pricing to ad allocation to repeated auctions  to network routing/scheduling. Various special cases with budget/supply constraints were studied previously, \citep[\eg][]{BZ09,DynPricing-ec12,DynProcurement-ec12,Krause-www13,combes2015bandits}.

In \BwK, the algorithm is endowed with $d\geq 1$ limited resources that are consumed by the algorithm. In each round, the algorithm chooses an action (\emph{arm}) from a fixed set of $K$ actions, and the outcome consists of a reward and consumption of each resource; all are assumed to lie in $[0,1]$. The algorithm observes \emph{bandit feedback}, \ie only the outcome of the chosen arm. The algorithm stops at time horizon $T$, or when the total consumption of some resource exceeds its budget. The goal is to maximize the total reward, denoted $\REW$.

For a concrete example, consider \emph{dynamic pricing}.%
\footnote{See Section 8 in \cite{BwK-focs13} for a detailed discussion of this and many other examples.}
The algorithm is a seller with limited supply of some product. In each round, a new customer arrives, the algorithm chooses a price, and the customer either buys one item at this price or leaves. A sale at price $p$ implies reward of $p$ and consumption of $1$. This example easily extends to $d>1$ products/resources. Now in each round the algorithm chooses the per-unit price for each resource, and the customer decides how much of each resource to buy at this price.

Prior work on \BwK focused on the stochastic version of the problem, called \emph{\StochasticBwK}, where the outcome of each action is drawn from a fixed distribution. This problem has been solved optimally using three different techniques \citep{BwK-focs13,AgrawalDevanur-ec14}, and extended in various directions in subsequent work \citep{AgrawalDevanur-ec14,cBwK-colt14,CBwK-colt16,CBwK-nips16}.

We go beyond the stochastic version, and instead study the most ``pessimistic", adversarial version where the rewards and resource consumptions can be arbitrary. We call it \emph{adversarial bandits with knapsacks} (\emph{\AdvBwK}), as it extends the classic model of ``adversarial bandits" \citep{bandits-exp3}.
Bandits aside, this problem subsumes online packing problems \citep{OnlineMatchingBook,buchbinder2009design}, where algorithm observes \emph{full feedback} (the outcomes of all possible actions) in each round, and observes it \emph{before} choosing an action.

\xhdr{Hardness of the problem.}
\AdvBwK is a much harder problem compared to \StochasticBwK. The new challenge is that the algorithm needs to decide how much budget to save for the future, without being able to predict it. (It is also the essential challenge in online packing problems, and it drives our lower bounds.) This challenge compounds the ones already present in \StochasticBwK: that exploitation may be severely limited by the resource consumption during exploration, that optimal per-round reward no longer guarantees optimal total reward, and that the best fixed distribution over arms may perform much better than the best fixed arm. Jointly, these challenges amount to the following. An algorithm for \AdvBwK must compete, during any given time segment $[1,\tau]$, with a distribution over arms that maximizes the total reward on this time segment. However, this distribution may behave very differently, in terms of expected per-round outcomes, compared to the optimal distribution for some other time segment $[1,\tau']$.

	In more concrete terms, let $\OPTFD$ be the total expected reward of the \emph{best fixed distribution} over arms. In \StochasticBwK (as well as in adversarial bandits) an algorithm can achieve sublinear regret:
    $\OPTFD - \E[\REW] = o(T)$.%
\footnote{More specifically, one can achieve regret $\tilde{O}(\sqrt{KT})$ for adversarial bandits \citep{bandits-exp3}, as well as for \StochasticBwK if all budgets are $\Omega(T)$ \citep{BwK-focs13}. One can achieve sublinear regret for \StochasticBwK if all budgets are $\Omega(T^\alpha)$, $\alpha \in(0,1)$ \citep{BwK-focs13}.}
	In contrast, in \AdvBwK regret minimization is no longer possible, and we therefore are primarily interested in the \emph{competitive ratio} $\OPTFD/\E[\REW]$.

It is instructive to consider a simple example in which the competitive ratio is at least $\tfrac54-o(1)$ for any algorithm.
There are two arms and one resource with budget $\tfrac{T}{2}$. Arm $1$ has zero rewards and zero consumption. Arm $2$  has consumption $1$ in each round, and offers reward $\tfrac12$ in each round of the first half-time ($\tfrac{T}{2}$ rounds). In the second half-time, it offers either reward $1$ in all rounds, or reward $0$ in all rounds. Thus, there are two problem instances that coincide for the first half-time and differ in the second half-time. The algorithm needs to choose how much budget to invest in the first half-time, without knowing what comes in the second. Any choice leads to competitive ratio at least $\tfrac54$ on one of the problem instances.%


Extending this idea, we prove an even stronger lower bound on the competitive ratio:
\begin{align}\label{eq:intro-LB}
    \OPTFD/ \E[\REW] \geq \Omega(\log T).
\end{align}
Like the simple example above, the lower-bounding construction involves only two arms and only one resource, and forces the algorithm to make a huge commitment without knowing the future.

\xhdr{Algorithmic contributions.}
Our main result is an algorithm which nearly matches  \eqref{eq:intro-LB}, achieving
\begin{align}\label{eq:intro-ratio}
    \E[\REW] \geq \tfrac{1}{O(\log T)} \left( \OPTFD - \reg \right),
\end{align}
where $\reg$ is a low-order regret term.

	We put forward a new algorithm for \BwK, called \MainALG, that unifies the stochastic and adversarial versions. It has a natural game-theoretic interpretation for \StochasticBwK, and admits a simpler analysis compared to the prior work. For \AdvBwK, we use \MainALG as a subroutine, though with a different parameter and a different analysis, to derive two algorithms: a simple one that achieves \eqref{eq:intro-ratio}, and a more involved one that achieves the same competitive ratio with high probability. Absent resource consumption, we recover the optimal $\tilde{O}(\sqrt{KT})$ regret for adversarial bandits.


	\MainALG is based on a new perspective on \StochasticBwK. We reframe a standard linear relaxation for \StochasticBwK in a way that gives rise to a repeated zero-sum game, where the two players choose among arms and resources, respectively, and the payoffs are given by the Lagrange function of the linear relaxation. Our algorithm consists of  two online learning algorithms playing this repeated game. We analyze \MainALG for \StochasticBwK, building on the tools from regret minimization in stochastic games, and achieve a near-optimal regret bound when the optimal value and the budgets are $\Omega(T)$.%
\footnote{This regime is of primary importance in prior work \citep[\eg][]{BZ09,Wang-OR14}.}

	We obtain several extensions, where we derive improved performance guarantees for some scenarios. These extensions showcase the \emph{modularity} of \MainALG, in the sense that the two players can be implemented as arbitrary algorithms for adversarial online learning that admit a given regret bound. Each extension follows from the main results, with a different choice of the players' algorithms.

\OMIT{We obtain several extensions, where we derive improved performance guarantees as easy corollaries of the main results. The extensions concern three prolific lines of work: respectively, full feedback, combinatorial semi-bandits, and contextual bandits (see Section~\ref{sec:ext} for mode details).  In particular, for the contextual extension of \StochasticBwK we match the best regret rate from prior work when the optimal value is $\Omega(T)$.}

\OMIT{ 
Next, we analyze \MainALG for the adversarial setting. This analysis follows a different path, as we no longer have a stochastic game. We give a simple algorithm for \AdvBwK, which randomly guesses a parameter (essentially, the value of $\OPTFD$), and then invokes \MainALG with this parameter. This algorithm suffices to guarantee \eqref{eq:intro-ratio}, \asedit{which is in terms of the expected reward $\E[\REW]$}. We provide another algorithm for \AdvBwK to achieve the same competitive ratio with high probability. This algorithm is considerably  more involved, both in design and in analysis. It uses \MainALG as a subroutine, and invokes its adversarial analysis as a lemma.} 

\OMIT{We work out an application to dynamic pricing, where the algorithm is a seller with limited supply $B$ of several products. (Here, price vectors are the arms, and products are the resources.) We obtain regret for the stochastic version, and $O(\log B)$ competitive ratio for the adversarial version. Both performance guarantees are relative to unrestricted price vectors. This is significant because prior work for $d>1$ products only provided performance guarantees relative to the discretized version in which the price vectors are restricted to some fixed and finite mesh. As a by-product, we upper-bound the \emph{discretization error}: the difference in $\OPTFD$ between the original problem and the discretized version.}

\xhdr{Discussion.}
\MainALG has numerous favorable properties.  As just discussed, it is simple, unifying, modular, and yields strong performance guarantees in multiple settings.  It is the first ``black-box reduction" from bandits to \BwK: we take a bandit algorithm and use it as a subroutine for \BwK. This is a very natural algorithm for the stochastic version once the single-shot game is set up; indeed, it is immediate from prior work that the repeated game converges to the optimal distribution over arms. Its regret analysis for \StochasticBwK is extremely clean. Compared to prior work \citep{BwK-focs13,AgrawalDevanur-ec14}, we side-step the intricate analysis of sensitivity of the linear program to non-uniform stochastic deviations that arise from adaptive exploration.

\MainALG has a primal-dual interpretation, as arms and resources correspond respectively to primal and dual variables in the linear relaxation. Two players in the repeated game can be seen as the respective \emph{primal algorithm} and \emph{dual algorithm}. Compared to the rich literature on \emph{primal-dual algorithms}~\citep{williamson2011design,buchbinder2009design,OnlineMatchingBook}
\citep[including the more recent literature on stochastic online packing problems][]{DevanurH-ec09,AgrawalWY-OR14,DevanurJSW-ec11,FeldmanHKMS-esa10,MolinaroR-icalp12}
\MainALG has a very specific and modular structure dictated by the repeated game.

Logarithmic competitive ratios are fairly common and well-accepted in the area of approximation algorithms, and particularly in online algorithms (see Related Work for citations).

\OMIT{The modularity of \MainALG is crucial for the extensions mentioned above: each extension follows from the main results, with a different choice of the primal algorithm.}

\OMIT{
The modularity of \MainALG is crucial for our result on dynamic pricing. Indeed, for this result we analyze the primal algorithm in the repeated game, derive a new regret bound, and plug this regret bound into the general theorem for \MainALG. In particular, it suffices to analyze the discretization error for the primal algorithm, \ie for a standard bandit problem, whereas prior work had to deal with a much more complicated problem of analyzing the discretization error in \BwK directly.
}

\xhdr{Benchmarks.}
We argue that the best fixed distribution over arms is an appropriate benchmark for \AdversarialBwK. First, consider the total expected reward of the \emph{best dynamic policy}, denote it $\OPTDP$. (The best dynamic policy is the best algorithm, in hindsight, that is allowed to switch arms arbitrarily across time-steps.) This is the strongest possible benchmark, but it is \emph{too} strong for \AdvBwK. Indeed, we show a simple example with just one resource (with budget $B$), where competitive ratio against this benchmark is at least $\tfrac{T}{B^2}$ for any algorithm. Second, consider the total expected reward of the \emph{best fixed arm}, denote it $\OPTFA$. It is a traditional benchmark in multi-armed bandits, but is uninteresting for \AdvBwK. We show that the competitive ratio is at least $\Omega(K)$ in the worst case, and this is matched, in expectation and up to a constant factor, by a trivial algorithm that samples one arm at random and sticks with it forever.

For \StochasticBwK, these three benchmarks are related as follows. The best fixed distribution is still the main object of interest in the analysis. However, all results -- both ours and prior work -- are almost automatically extended to compete against the best dynamic policy. The best fixed arm is a much weaker benchmark than the best fixed distribution: there are simple examples when their expected reward differs by a factor of two, in multiple special cases of interest \citep{BwK-focs13}.

\OMIT{
Both our result and prior work on \StochasticBwK can be extended to a stronger benchmark, namely, the total expected reward of the \emph{best dynamic policy} (the best algorithm, in hindsight, that is allowed to switch arms arbitrarily across time-steps), denoted $\OPTDP$. Likewise, prior work on online packing problems achieves a similar $\log(T)$ competitive ratio against this stronger benchmark. However, we argue that this benchmark is \emph{too} strong for \AdvBwK: we show a simple example with just one resource (with budget $B$), where competitive ratio is at least $\tfrac{T}{B^2}$ for any algorithm.

A traditional benchmark in multi-armed bandits is the expected reward of the \emph{best fixed arm}, denoted $\OPTFA$. This is a much weaker benchmark for \StochasticBwK: there are simple examples with $\OPTFD\geq 2\,\OPTFA$ in multiple special cases of interest \cite{BwK-focs13}.%
\footnote{\Eg we have two arms and two resources. In each round, each arm $i\in \{1,2\}$ collects reward $1$, consumes $1$ unit of resource $i$, and $0$ units of the other resource. Both budgets are $B$, and $T=2B$. Then $\OPTFA=B$, whereas $\OPTFD$ is close to $2B$.}
For the adversarial version, $\OPTFA$ is both weaker and less interesting: we show that the competitive ratio is at least $\Omega(K)$ in the worst case, and this is matched, in expectation, by a trivial algorithm that samples one arm at random and sticks with it forever.
} 

\section{Related work}
\label{sec:related}

The literature on regret-minimizing online learning  is vast; see \cite{CesaBL-book,Bubeck-survey12,Hazan-OCO-book} for background. Most relevant are two algorithms for adversarial rewards/costs: Hedge for full feedback \citep{freund1999adaptive}, and EXP3 for bandit feedback \citep{bandits-exp3}; both are based on the weighted majority algorithm from \citep{LittWarm94}.

\StochasticBwK was introduced and optimally solved in \cite{BwK-focs13}. Subsequent work extended these results to soft supply/budget constraints \citep{AgrawalDevanur-ec14}, a more general notion of rewards%
\footnote{The total reward is determined by the time-averaged outcome vector, but can be an arbitrary Lischitz-concave function thereof.}
\citep{AgrawalDevanur-ec14},
combinatorial semi-bandits \citep{Karthik-aistats18}, and contextual bandits \citep{cBwK-colt14,CBwK-colt16,CBwK-nips16}. Several special cases with budget/supply constraints were studied previously:
dynamic pricing \citep{BZ09,DynPricing-ec12,BesbesZeevi-or12,Wang-OR14},
 dynamic procurement \citep{DynProcurement-ec12,Krause-www13} (a version of dynamic pricing where the algorithm is a buyer rather than a seller), dynamic ad allocation \citep{AdsWithBudgets-arxiv13,combes2015bandits}, and a version with a single resource and unlimited time \citep{Gyorgy-ijcai07,TranThanh-aaai10,TranThanh-aaai12,Qin-aaai13}.
 While all this work is on regret minimization,
\citet{GuhaM-icalp09,GuptaKMR-focs11} studied closely related Bayesian formulations.

\StochasticBwK was optimally solved using three different algorithms \citep{BwK-focs13,AgrawalDevanur-ec14}, with extremely technical and delicate analyses. All three algorithms involve inherently `stochastic' techniques such as ``successive elimination" and ``optimism under uncertainty", and do not appear to extend to the adversarial version. One of them, \term{PrimalDualBwK} from \cite{BwK-focs13}, is a primal-dual algorithm superficially similar to ours. Indeed, it decouples into two online learning algorithms: a ``primal" algorithm which chooses among arms, and a ``dual" algorithm similar to ours, which chooses among resources. However, the two algorithms are not playing a repeated game in any meaningful sense, let alone a zero-sum game. The primal algorithm operates under a much richer input: it takes the entire outcome vector for the chosen arm, as well as the ``dual distribution" -- the distribution over resources chosen by the dual algorithm. Further, the primal algorithm is very problem-specific: it interprets the dual distribution as a vector of costs over resources, and chooses arms with largest reward-to-cost ratios, estimated using ``optimism under uncertainty".

Our approach to using regret minimization in games can be traced to \cite{freund1996game,freund1999adaptive} (see Ch. 6 in \cite{SchapireBoosting}), who showed how a repeated zero-sum game played by two agents yields an approximate Nash equilibrium.
\OMIT{The framework of regret minimization in games (Lemma~\ref{lem:learningGames} and related results) was introduced in  and has been significant in several ways. On a face value, it is about a game between two agents that run regret-minimizing algorithms. Second, it advances regret minimization as a plausible behavioral model that leads to approximate Nash equilibria. Third, it gives a way compute an approximate Nash equilibrium. }
This approach has been used as a unifying algorithmic framework for several  problems: boosting~\citep{freund1996game}, linear programs~\citep{AroraHK}, maximum flow~\citep{Christiano11MaxFlow},
and  convex optimization~\citep{abernethy2017frank,wang2018acceleration}.
While we use a result with the $1/\sqrt{t}$ convergence rate for the equilibrium property, recent literature obtains faster convergence for cumulative payoffs (but not for the equilibrium property) under various assumptions
\citep{Rakhlin-nips13,Daskalakis-GEB15,Haipeng-colt18}.

	Repeated Lagrangian games, in conjunction with regret minimization in games, have been used in a series of recent papers
\citep{rogers2015inducing,hsu2016jointly,roth2016watch,pmlr-v80-kearns18a,agarwal2018reductions,roth2017multidimensional},
as an algorithmic tool to solve convex optimization problems; application domains range from differential privacy to algorithmic fairness to learning from revealed preferences. All these papers deal with deterministic games (\ie same game matrix in all rounds). Reframing the problem in terms of repeated Lagrangian games is a key technical insight in this work. 
Most related to our paper are \citet{roth2016watch,roth2017multidimensional}, where a repeated Lagrangian game is used as a subroutine (the ``inner loop") in an online algorithm; the other papers solve an offline problem. We depart from this prior work in several respects. Our main results are for the adversarial version, where the standard machinery does not apply and we provide a very different analysis. For the stochastic version, we use a stochastic game and we deal with some subtleties specific to \BwK.

\OMIT{KARTHIK'S READING REPORT, PLEASE DO NOT REMOVE!
\kaedit{Papers with repeated Lagrangian games \cite{rogers2015inducing,hsu2016jointly,roth2016watch,pmlr-v80-kearns18a,agarwal2018reductions,roth2017multidimensional}.
	\begin{itemize}
		\item Solving an online problem? - \cite{roth2016watch}, \cite{roth2017multidimensional}
		\item Playing on a stochastic game? - None
		\item Using no-regret rather than best-response for the dual? \cite{agarwal2018reductions,pmlr-v80-kearns18a}
		\item Solving the actual problem rather than some subroutine? \cite{hsu2016jointly,rogers2015inducing,pmlr-v80-kearns18a,agarwal2018reductions}.
And all they do, technically, is apply the hammer.
		\item Solving a problem whose relaxation is an LP very much like ours? - Most are general convex programs. Reduction to LP does not make sense in the problem setting.
		\item Doing some adversarial analysis? - Unsure what this means.
		\item Achieving smth optimal-ish? \cite{pmlr-v80-kearns18a,agarwal2018reductions}. Most of these papers do not cite a lower-bound. So I am not sure if its optimal.

	\end{itemize}
}} 


Online packing problems \citep[\eg][]{BuchbinderNaor09,buchbinder2009design,DevanurJSW-ec11}
can be seen as a special case of \AdvBwK with a much more permissive feedback model: the algorithm observes full feedback (the outcomes for all actions) before choosing an action. Online packing subsumes various \emph{online matching} problems, including the \emph{AdWords problem} \citep{MSVV07} motivated by ad allocation \cite[see][for a survey]{OnlineMatchingBook}. While we derive $O(\log T)$ competitive ratio against $\OPTFD$, online packing admits a similar result against $\OPTDP$.%
\OMIT{An extension to contextual BwK allows us to treat the full feedback as a context, and compete against the best fixed distribution over policies (see Section~\ref{sec:ext-CB}). The latter benchmark coincides with $\OPTDP$ for both the AdWords and the Online Packing problem. However, for AdWords one can obtain a constant competitive ratio, whereas our results only imply $O(\log T)$ ratio.}

Another related line of work concerns online convex optimization with constraints \citep{mahdavi2012trading,mahdavi2013stochastic,chen2017online,neely2017online,chen2018bandit}. Their setting differs from ours in several important respects. First, the action set is a convex subset of $\R^K$ (and the algorithms rely on the power to choose arbitrary actions in this set). In particular, there is no immediate way to handle discrete action sets.%
\footnote{Unless there is full feedback, in which case one can use a standard reduction whereby actions in online convex optimization correspond to distributions over actions in a $K$-armed bandit problem.}
Second, convexity/concavity is assumed on the rewards and resource consumption.
Third, full feedback is observed for the resource consumption. Moreover,  in all papers except \citet{chen2018bandit} one also observes either full feedback on rewards or the rewards gradient around the chosen action. Fourth, their algorithm only needs to satisfy the budget constraints at the time horizon (whereas in BwK the budget constraints hold for all rounds). Fifth, their fixed-distribution benchmark is weaker than ours: essentially, its time-averaged consumption must be small enough at each round $t$. Due to these differences, this setting admits sublinear regret for adversarial outcomes \citep{neely2017online}. The other papers in this line of work focus on  stochastic outcomes.

Logarithmic competitive ratios are quite common in prior work on approximation algorithms and online algorithms. Examples include:
set cover \citep{lovasz1975ratio,johnson1974approximation},
buy-at-bulk network design \citep{awerbuch1997buy},
sparsest cut \citep{arora2009expander},
the dial-a-ride problem \citep{charikar1998finite},
online k-server \citep{bansal2011polylogarithmic},
online packing/covering \citep{azar2016online},
online set cover \citep{alon2003online},
online network design \citep{umboh2015online},
and online paging \citep{fiat1991competitive}.

\OMIT{ KARTHIK'S READING REPORT, PLEASE DO NOT REMOVE!
\kaedit{notes for related work on Bandit convex optimization.
	All works below have the following setting. At each time, we choose a point $\vec{x}_t$. Then we are shown/receive a loss $f_t(x_t)$ (bandit or full-feedback (varies by work)) and receive cost functions $g_{t,1}, g_{t, 2}, .. g_{t, d}$ (full feedback). The goal is to minimize the total loss $\sum_{t \in [T]} f_t(x_t)$ such that the long-term constraint $\sum_{t \in [T]} g_{t, i} \leq 0$ for every $i \in [d]$. Note that in intermediate time-steps, one can violate the constraints.
	
	This setting differs from BwK in the following ways. First, we only receive bandit feedback in the constraints. Second, in our setting the constraints should be maintained at every time-step. The game stops the first time we violate a constraint. This is critical since, we cannot obtain a sublinear regret with hard constraints, while in their setting they obtain a sub-linear regret. There doesn't seem to be a natural reduction from BwK to their setting. Third, all action sets are convex (and hence continuous). They do not extend to discrete action steps.
	Fourth, the functions $g_{t, i}$ cannot all be positive. This is because, they assume that there exists a feasible solution in hind-sight. Thus, this implies we have both consumption and replenishment of resources.
	
	\begin{itemize}
		\item \cite{chen2017online}. Adversarial setting. Not concurrent work. Full-feedback on the objective.
		\item \cite{neely2017online}. Adversarial setting. Not concurrent work. Full-feedback on the objective. Unpublished.
		\item \cite{mahdavi2012trading}. Adversarial setting. Not concurrent work. Bandit feedback on the objective. The functions $g$ are all known beforehand. They assume that they have access to the gradient of the objective function at each time-step.
		\item  \cite{mahdavi2013stochastic} Stochastic setting. Not concurrent work. Full feedback on the objective. The functions $g$ are all linear.
		\item \cite{chen2018bandit} Adversarial setting. Published 22 May 2018. Bandit feedback on the objective. No gradient assumption. They obtain an estimate to gradient using one point access to the function.
	\end{itemize}
}} 

\subsection{Simultaneous and independent work}
\label{sec:related-simultaneous}

Three related papers have come to our attention after the initial version of our paper has appeared on {\tt arxiv.org} in Nov'18. At the time, \citet{rivera2018online,rangi2018unifying} have been available as yet unpublished technical reports, and \citet{Jake-icml19} has not yet appeared.

\citet{rivera2018online} consider online convex optimization with knapsacks: essentially, the problem defined in Section~\ref{sec:ext-BCO}, but with full feedback. Focusing on the stochastic version, they design an algorithm similar to \MainALG, and derive a regret bound similar to ours, using a similar analysis. They also claim an extension to bandit feedback, without providing any details (such as the  precise statement of Lemma~\ref{lem:learningGames} in terms of the regret property \eqref{eq:prelims-regret-weak}).

\citet{rangi2018unifying} consider \AdvBwK in the special case when there is only one constrained resource, including time. They attain sublinear regret, \ie a regret bound that is sublinear in $T$. They also assume a known lower bound $c_{\min}>0$ on realized per-round consumption of each resource, and their regret bound scales as $1/c_{\min}$. They also achieve $\text{polylog}(T)$ instance-dependent regret for the stochastic version using the same algorithm (matching results from prior work on the stochastic version). \BwK with only one constrained resource (including time) is a much easier problem, compared to the general case with multiple resources studied in this paper, in the following sense. First, the single-resource version admits much stronger performance guarantees ($\text{polylog}(T)$ vs. $\sqrt{T}$ regret bounds for \StochasticBwK, and sublinear regret vs. competitive ratio for \AdvBwK). Second, the optimal all-knowing time-invariant policy is the best fixed arm, rather than the best fixed distribution over arms.

\citet{Jake-icml19} study online learning in repeated adversarial zero-sum games (which is our main technical tool). They obtain a powerful result for arbitrary games: an online learning algorithm which controls both players and guarantees convergence to the Nash equilibrium. They apply their framework to train Generative Adversarial Networks (GANs). Interestingly, they achieve the competitive ratio of $1$, despite the adversarial setting.  Their algorithm can continue up to round $T$, with no stopping rule like in \BwK; for this reason, their results do not have an immediate bearing on our problem.

\OMIT{
\asedit{Dynamic pricing and, more generally, revenue management problems, have a rich literature in Operations Research \cite{Boer-survey15}, mainly focusing on parameterized demand distributions with Bayesian priors. The study of prior-independent dynamic pricing has been initiated in~\cite{Blum03,KleinbergL03}, without limited supply constraints. Dynamic pricing with a limited supply of $d=1$ product has been introduced in \cite{BZ09}. It has been optimally solved against the best fixed arm in \cite{DynPricing-ec12,Wang-OR14}, and finally against the best dynamic policy in \cite{BwK-focs13}. Dynamic pricing with a limited supply of multiple products was first studied in \cite{BesbesZeevi-or12}, and subsequently has been an important special case in all work on \BwK. Most of this work, starting from \cite{KleinbergL03}, uses the \emph{fixed discretization} technique, where the infinite set of prices / price vectors is replaced by a mesh.%
\footnote{This technique is also standard in the work on multi-armed bandits with Lipschitz assumption (see Chapter 4 in \cite{slivkins-MABbook}).} The granularity of this mesh can be adjusted in the analysis, as long as the discretization error can be upper-bounded. The latter, however, has not been previously achieved for $d>1$ products.}
} 

%
%


\OMIT{ 
\emph{Online Matching} and \emph{AdWords} problem \cite{karp1990optimal,devanur2013randomized,MSVV07,DevanurH-ec09,
Karande11Unknown,OnlineMatchingBook,vertexWeighted11,Feldman09IID,Goel08Budgetted} and more general online packing problems
\cite{DevanurJSW-ec11,BuchbinderNaor09,buchbinder2009design,agrawal2014dynamic,kesselheim2014primal,alon2006general}.

Primal-Dual algorithms classic papers \cite{LPPrimalDual,MaxFlow,kuhn1955hungarian}. See surveys \cite{williamson2002primal,goemans1997primal} for overview in approximation algorithms.
} 


\section{Preliminaries}
	\label{sec:prelims}

We use bold fonts to represent vectors and matrices.
	We use standard notation whereby, for a positive integer $K$, $[K]$ stands for $\{1,2 \LDOTS K\}$, and $\Delta_K$ denotes the set of all probability distributions on $[K]$.
Some of the notation introduced further is summarized in Appendix~\ref{sec:notations}.


\xhdr{Bandits with Knapsacks (\BwK).}
There are $T$ rounds, $K$ possible actions and $d$ resources, indexed as
    $[T], [K], [d]$,
respectively. In each round $t\in [T]$, the algorithm chooses an action $a_t\in [K]$ and receives an \emph{outcome vector}
    $\vec{o}_t = (r_t; \, c_{t,1} \LDOTS c_{t,d}) \in [0,1]^{d+1}$,
where $r_t$ is a reward and $c_{t,i}$ is consumption of each resource $i\in [d]$. Each resource $i$ is endowed with budget $B_i\leq T$. The game stops early, at some round $\taualg<T$, when/if the total consumption of any resource exceeds its budget. The algorithm's objective is to maximize its total reward. Without loss of generality all budgets are the same:
    $B_1 = B_2 = \ldots = B_d = B$.%
\footnote{To see that this is indeed w.l.o.g., for each resource $i$, divide all per-round consumptions $c_{t,i}$ by $B_i/B$, where
    $B := \min_{i \in [d]} B_i$
is the smallest budget. In the modified problem instance, all consumptions still lie in $[0, 1]$, and all the budgets are equal to $B$.}

The outcome vectors are chosen as follows. In each round $t$, the adversary chooses the \emph{outcome matrix}
    $\vM_t \in [0,1]^{K\times (d+1)}$,
where rows correspond to actions. The outcome vector $\vec{o}_t$ is defined as the $a_t$-th row of this matrix, denoted $\vM_t(a_t)$. Only this row is revealed to the algorithm. The adversary is deterministic and \emph{oblivious}, meaning that the entire sequence
    $\vM_1 \LDOTS \vM_T$
is chosen before round $1$. A problem instance of \BwK consists of (known) parameters $(d,K,T,B)$, and the (unknown) sequence $\vM_1 \LDOTS \vM_T$.

In the stochastic version of \BwK, henceforth termed \emph{\StochasticBwK}, each outcome matrix $\vM_t$ is chosen from some fixed but unknown distribution $\outcomeD$ over the outcome matrices. An instance of this problem consists of (known) parameters $(d,K,T,B)$, and the (unknown) distribution $\outcomeD$.

Following prior work \citep{BwK-focs13,AgrawalDevanur-ec14}, we assume, w.l.o.g., that one of the resources is a \emph{dummy resource} similar to time; formally, each action consumes $B/T$ units of this resource per round (we only need this for \StochasticBwK). Further, we posit that one of the actions is a \emph{null action}, which lets the algorithm skips a round: it has $0$ reward and consumes $0$ amount of each resource other than the dummy resource.
	
\xhdr{Benchmarks.}
Let $\REW(\ALG) = \sum_{t\in [\taualg]} r_t$ be the total reward of algorithm $\ALG$ in the \BwK problem. Our benchmark is the \emph{best fixed distribution}, a distribution over actions which maximizes $\E[\REW(\cdot)]$ for a particular problem instance. The expected total reward of this distribution is denoted $\OPTFD$.

For \StochasticBwK, one can compete with the \emph{best dynamic policy}: an algorithm that maximizes $\E[\REW(\cdot)]$ for a particular problem instance. Essentially, this algorithm knows the latent distribution $\outcomeD$ over outcome matrices. Its expected total reward is denoted $\OPTDP$.

%
%
%

\xhdr{Adversarial online learning.}
To state the framework of ``regret minimization in games" below, we need to introduce the protocol of  \emph{adversarial online learning}, see Figure~\ref{prob:adv}.

\begin{figure}[h]
\begin{framed}
{\bf Given:} action set $A$, payoff range $[\rmin,\rmax]$.\\
In each round $t\in [T]$,
\begin{OneLiners}
\item[1.] the adversary chooses a payoff vector $\vec{f}_t\in [\rmin,\rmax]^K$;
\item[2.] the algorithm chooses a distribution $\vec{p}_t$ over $A$, without observing $\vec{f}_t$,
\item[3.] algorithm's chosen action $a_t\in A$ is drawn independently from $\vec{p}_t$;
\item[4.] payoff  $f_t(a_t)$ is received by the algorithm.
\end{OneLiners}
\vspace{-1mm}
\end{framed}
\vspace{-1mm}
\caption{Adversarial online learning}
\label{prob:adv}
\end{figure}

In this protocol, the adversary can use previously chosen arms to choose the payoff vector $\vec{f}_t$, but not the algorithm's random seed. The distribution $\vec{f}_t$ is chosen as a deterministic function of history. (The history at round $t$ consists, for each round $s<t$, of the chosen action $a_s$ and the observed feedback in this round.) We focus on two feedback models: \emph{bandit feedback} (no auxiliary feedback) and \emph{full feedback} (the entire payoff vector $\vec{f}_t$). The version for costs can be defined similarly, by setting the payoffs to be the negative of costs.

We are interested in adversarial online learning algorithms with known upper bounds on \emph{regret},
\begin{align}
\regAOL(T):= \textstyle
       \sbr{ \max_{a\in A} \sum_{t\in [T]} f_t(a) }
    \;-\;  \sbr{ \sum_{t\in [T]} f_t(a_t) }.
\end{align}
The benchmark here is the total payoff of the best arm, according to the payoff vectors actually chosen by the adversary. More precisely, we assume high-probability regret bounds of the following form:
\begin{align}\label{eq:prelims-regret-weak}
\textstyle
\forall \delta>0 \qquad
\Pr\sbr{ \regAOL(T)  \leq (\rmax-\rmin)\, R_\delta(T) } \geq 1-\delta,
\end{align}
for some function $R_\delta(\cdot)$. We will actually use a stronger version implied by \eqref{eq:prelims-regret-weak},%
\footnote{Regret bound \eqref{eq:prelims-regret-strong} follows from \eqref{eq:prelims-regret-weak} using a simple ``zeroing-out" trick: for a given round $\tau\in [T]$, the adversary can set all future payoffs to some fixed value $x\in [\rmin,\rmax]$, in which case
    $\regAOL(\tau) = \regAOL(T)$.}
\begin{align}\label{eq:prelims-regret-strong}
\textstyle
\forall \delta>0 \qquad
\Pr\left[\; \forall \tau\in [T]\quad \regAOL(\tau)  \leq (\rmax-\rmin)\, R_{\delta/T}(T) \;\right] \geq 1-\delta.
\end{align}

\noindent Algorithms EXP3.P \citep{bandits-exp3} for bandit feedback, and Hedge \citep{FS97} for full feedback, satisfy \eqref{eq:prelims-regret-weak} with, resp.,
\begin{align}\label{eq:prelims-regrets}
    \text{$R_\delta(T) = O\left(\sqrt{|A|\,T\log (T/\delta)} \right)$
~~and~~
    $R_\delta(T)=O\left( \sqrt{T \log (|A|/\delta) } \right)$.}
\end{align}

\xhdr{Regret minimization in games.}
We build on the framework of \emph{regret minimization in games}. A \emph{zero-sum game} $(A_1,A_2,\vG)$ is a game between two players $i\in \{1,2\}$ with action sets $A_1$ and $A_2$ and payoff matrix $\vG\in \R^{A_1\times A_2}$. If each player $i$ chooses an action $a_i\in A_i$, the outcome is a number $G(a_1,a_2)$. Player $1$ receives this number as \emph{reward}, and player $2$ receives it as \emph{cost}. A \emph{repeated zero-sum game} $\mG$ with action sets $A_1$ and $A_2$, time horizon $T$ and game matrices
    $\vG_1 \LDOTS \vG_T\in \R^{A_1\times A_2}$
is a game between two algorithms, $\ALG_1$ and $\ALG_2$, which proceeds over $T$ rounds such that each round $t$ is a zero-sum game $(A_1,A_2,\vG_t)$. The goal of $\ALG_1$ is to maximize the total reward, and the goal of $\ALG_2$ is to minimize the  total cost.

The game $\mG$ is called \emph{stochastic} if the game matrix $\vG_t$  in each round $t$ is drawn independently from some fixed distribution. For such games, we are interested in the \emph{expected game}, defined by the expected game matrix $\vG = \E[\vG_t]$. We can relate the algorithms' performance to the minimax value of $\vG$.

\begin{restatable}{lemma}{learningGames}
\label{lem:learningGames}
Consider a stochastic repeated zero-sum game between algorithms $\ALG_1$ and $\ALG_2$, with payoff range $[\rmin,\rmax]$.  Assume that each $\ALG_j$, $j\in \{1,2\}$ is an algorithm for adversarial online learning, as per Figure~\ref{prob:adv}, which satisfies regret bound \eqref{eq:prelims-regret-weak} with $R_\delta(T) =R_{j,\delta}(T)$.

Let $\tau$ be some fixed round in the game. For each algorithm
    $\ALG_j$, $j\in \{1,2\}$,
let $A_j$ be its action set, let $p_{t,j}\in \Delta_{A_j}$ be the distribution chosen in each round $t$, and let
    $\bar{\vec{p}}_j = \frac{1}{\tau} \sum_{t\in [\tau]}\vec{p}_{t,j}$
be the average play distribution at round $\tau$. Let $v^*$ be the minimax value for the expected game $\vG = \E[\vG_t]$.

Then for each $\delta>0$, with probability at least $1-2\delta$ it holds that
\begin{align}
\forall \vec{p}_2\in \Delta_{A_2}\quad
\bar{\vec{p}}_1^\tran\, \vG\, \vec{p}_2 \geq v^*
    - \tfrac{1}{\tau}(\rmax-\rmin)\,
          \left(\; R_{1,\,\delta/T}(T)+R_{2,\,\delta/T}(T)
                + 4\sqrt{2T \log (T/\delta)} \; \right).
        \label{eq:prelims-games-play}
\end{align}
\end{restatable}

\refeq{eq:prelims-games-play} states that the average play of player $1$ is approximately optimal against any distribution chosen by player $2$.%
\footnote{If each player $j$ chooses distribution $p_j\in \Delta_{A_j}$, and the game matrix is $\vG$, then expected reward/cost is $\vec{p}_1^\tran \vG \vec{p}_2$.}
This lemma is well-known for the deterministic case (\ie when $\vec{G}_t = \vec{G}$ for each round $t$), and folklore for the stochastic case. We provide a proof in Appendix~\ref{app:pf-games} for the sake of completeness.

\OMIT{ 
		Consider the following setup. We have $T$ rounds and two players. At each time-step $t$, player $1$ and player $2$ chooses distributions $\vec{P}_t$, $\vec{Q}_t$ respectively, over their actions. We assume that player $1$ uses a no-regret bandit algorithm such as EXP3 \cite{bandits-exp3} while player $2$ uses a no-regret full-information algorithm such as Hedge \cite{FS97,freund1999adaptive}. We then have the following convergence lemma.
} 

%
%

\section{A new algorithm for Stochastic BwK}
\label{sec:IID}

We present a new algorithm for \StochasticBwK, based on the framework of regret minimization in games. This is a very natural algorithm once the single-shot game is set up, and it allows for a very clean regret analysis. We will also use this algorithm as a subroutine for the adversarial version.

On a high level, we define a stochastic zero-sum game for which a  mixed Nash equilibrium corresponds to an optimal solution for a linear relaxation of the original problem. Our algorithm consists of two regret-minimizing algorithms playing this game. The framework of regret minimization in games guarantees that the average primal and dual play distributions ($\bar{\vec{p}}_1$ and $\bar{\vec{p}}_2$ in Lemma~\ref{lem:learningGames}) approximate the mixed Nash equilibrium in the expected game, which correspondingly approximates the optimal solution.

\subsection{Linear relaxation and Lagrange functions}

We start with a linear relaxation of the problem that all prior work relies on. This relaxation is stated in terms of expected rewards/consumptions, \ie implicitly, in terms of the expected outcome matrix $\vM = \E[\vM_t]$. We explicitly formulate the relaxation in terms of $\vM$, and this is essential for the subsequent developments. For ease of notation, we write the $a$-th row of $\vM$, for each action $a\in[K]$, as
    \[ \vM(a) = (r^{\vM}(a); \, c^{\vM}_1(a) \LDOTS c^{\vM}_d(a)), \]
so that $r^{\vM}(a)$ is the expected reward and $c^{\vM}_i(a)$ is the expected consumption of each resource $i$.

Essentially, the relaxation assumes that each instantaneous outcome matrix $\vM_t$ is equal to the expected outcome matrix
    $\vM = \E[\vM_t]$. The relaxation seeks the best distribution over actions, focusing on a single round with budgets rescaled as $B/T$. This leads to the following linear program (LP):

\begin{equation}
\label{lp:primalAbstract}
\begin{array}{ll@{}ll}
	\text{maximize} \qquad
&
\sum_{a\in [K]} X(a)\;  r^{\vM}(a) & \text{such that}\\
 		& \sum_{a\in [K]} X(a) = 1  \\
\displaystyle \forall i \in [d]  \qquad
& \sum_{a\in[K]} X(a)\; c^{\vM}_i(a) \leq B/T \\
	\displaystyle \forall a \in [K] \qquad & \displaystyle 0 \leq X(a) \leq 1.
\end{array}
\end{equation}
We denote this LP by $\myLP{\vM}{B}{T}$. The solution $\vec{X}$ is the best fixed distribution over actions, according to the relaxation. The value of this LP, denoted $\OPTLP(\vM,B,T)$, is the expected per-round reward of this distribution. It is also the total reward of $\vec{X}$ in the relaxation, divided by $T$.
We know from \cite{BwK-focs13} that
\begin{align}\label{eq:IID-benchmarks}
T\cdot\OPTLP(\vM,B,T) \geq \OPTDP\geq \OPTFD,
\end{align}
where $\OPTDP$ and $\OPTFD$ are the total expected rewards of, respectively, the best dynamic policy and the best fixed distribution. In words, $\OPTDP$ is sandwiched between the total expected reward of the best fixed distribution and that of its linear relaxation.

Associated with the linear program $\myLP{\vM}{B}{T}$ is the \emph{Lagrange function}
$\mL = \myLag{\vM}{B}{T}$. It is a function
$\mL: \Delta_K \times \mathbb{R}^d_{\geq 0} \rightarrow \R$ defined as
\begin{align}\label{eq:LagrangianGeneral}	
\mL(\vec{X}, \vec{\lambda})
:=  \sum_{a\in [K]} X(a)\, r^{\vM}(a) +
    \sum_{i\in [d]} \lambda_i
    \left[ 1-\frac{T}{B}\; \sum_{a\in [K]} X(a)\, c^{\vM}_i(a) \right].
\end{align}
The values $ \lambda_1 \LDOTS \lambda_d$ in \refeq{eq:LagrangianGeneral} are called the \emph{dual variables}, as they correspond to the variables in the dual LP. Lagrange functions are meaningful due to their max-min property (\eg Theorem D.2.2 in \cite{ben2001lectures}):
\begin{align}\label{eq:LagrangeMinMax}	
\min_{\vec{\lambda}\geq 0} \max_{\vec{X} \in \Delta_K}
    \mL(\vec{X}, \vec{\lambda})
= \max_{\vec{X} \in \Delta_K} \min_{\vec{\lambda}\geq 0}
    \mL(\vec{X}, \vec{\lambda})
= \OPTLP(\vM,B,T).
	\end{align}
This property holds for our setting because $\myLP{\vM}{B}{T}$ has at least one feasible solution (namely, one that puts probability one on the null action), and the optimal value of the LP is bounded.

\begin{remark}
We use the linear program $\myLP{\vM}{B}{T}$ and the associated Lagrange function $\myLag{\vM}{B}{T}$ throughout the paper. Both are parameterized by an outcome matrix $\vM$, budget $B$ and time horizon $T$. In particular, we can plug in an arbitrary $\vM$, and we heavily use this ability throughout. For the adversarial version, it is essential to plug in parameter $T_0\leq T$ instead of the time horizon $T$. For the analysis of the high-probability result in \AdvBwK, we use a rescaled budget $B_0\leq B$ instead of budget $B$.
\end{remark}

\subsection{Our algorithm: repeated Lagrangian game}

The Lagrange function $\mL = \myLag{\vM}{B}{T}$ from \eqref{eq:LagrangianGeneral} defines the following zero-sum game: the \emph{primal player} chooses an arm $a$, the \emph{dual player} chooses a resource $i$, and the payoff is a number
\begin{align}\label{eq:IID-Lagrangian-simple}
    \mL(a,i) = r^{\vM}(a) + 1-\tfrac{T}{B}\; c^{\vM}_i(a).
\end{align}
The primal player receives this number as a reward, and the dual player receives it as cost. This game is termed the \emph{Lagrangian game} induced by $\myLag{\vM}{B}{T}$. This game will be crucial throughout the paper.

The Lagrangian game is related to the original linear program as follows:

\begin{lemma} \label{lm:gameToLP}
Assume one of the resources is the dummy resource. Consider the linear program $\myLP{\vM}{B}{T}$, for some outcome matrix $\vM$. Then the value of this LP equals the minimax value $v^*$ of the Lagrangian game induced by $\myLag{\vM}{B}{T}$. Further, if
    $(\vec{X},\vec{\lambda})$
is a mixed Nash equilibrium in the Lagrangian game, then $\vec{X}$ is an optimal solution to the LP.
\end{lemma}
	
The proof can be found in Appendix~\ref{app:eq:LagrangeMinMaxNew}. The idea is that because of the special structure of the LP, the second equality in \eqref{eq:LagrangeMinMax} also holds when the dual vector $\vec{\lambda}$ is restricted to distributions.


Consider a repeated version of the Lagrangian game. Formally, the \emph{repeated Lagrangian game} with parameters $B_0\leq B$ and $T_0\leq T$ is a repeated zero-sum game between the \emph{primal algorithm} that chooses among arms and the \emph{dual algorithm} that chooses among resources. Each round $t$ of this game is the Lagrangian game induced by the Lagrange function
    $\mL_t :=\myLag{\vM_t}{B_0}{T_0}$,
where $\vM_t$ is the round-$t$ outcome matrix. Note that we use parameters $B_0,T_0$ instead of budget $B$ and time horizon $T$.%
\footnote{These parameters are needed only for the adversarial version. For \StochasticBwK we use $B_0=B$ and $T_0=T$.}

\begin{remark}
Consider repeated Lagrangian game for \StochasticBwK (with $B_0=B$ and $T_0=T$). The payoffs in the expected game are defined by the expected Lagrange function
    $\mL := \E[\mL_t]$.
By linearity, $\mL$ is the Lagrange function for the expected outcome matrix
    $\vM = \E[\vM_t] $:
\begin{align}\label{eq:IID-linearity}
 \mL := \E[\mL_t] = \myLag{\vM}{B}{T}.
\end{align}
\end{remark}

\OMIT{ 
\begin{remark}
Why does repeated Lagrangian game makes sense for \StochasticBwK? What the framework of regret minimization in games accomplishes is that the average primal and dual play distributions ($\bar{\vec{p}}_1$ and $\bar{\vec{p}}_2$ in Lemma~\ref{lem:learningGames}) approximate the mixed Nash equilibrium in the expected game, in some specific sense. And we know from Lemma~\ref{lm:gameToLP} that if $(\bar{\vec{p}}_1,\bar{\vec{p}}_2)$ were an \emph{exact} mixed Nash equilibrium, then $\bar{\vec{p}}_1$ would be an optimal solution for
    $\myLP{\vM}{B}{T}$,
which is essentially what we want in \StochasticBwK.
\end{remark}
} 

Our algorithm, called \MainALG, is very simple: it is a repeated Lagrangian game in which the primal algorithm receives bandit feedback, and the dual algorithm receives full feedback.

To set up the notation, let $a_t$ and $i_t$ be, respectively, the chosen arm and resource in round $t$. The payoff is therefore
    $\mL_t(a_t,i_t)$.
It can be rewritten in terms of the observed outcome vector
    $\vec{o}_t = (r_t; c_{t,1} \LDOTS c_{t,d})$
(which corresponds to the $a_t$-th row of the instantaneous outcome matrix $\vM_t$):
\begin{align}\label{eq:IID-Lagrange-t}
    \mL_t(a_t,i_t) = r_t + 1- \tfrac{T_0}{B_0}\; c_{t,i_t} \in \sbr{-\tfrac{T_0}{B_0}+1, 2}.
\end{align}

\noindent Note that the payoff range is $[\rmin,\rmax] = [-\tfrac{T_0}{B_0}+1]$.

With this notation, the pseudocode for \MainALG is summarized in Algorithm~\ref{alg:LagrangianBwK}. The pseudocode is simple and self-contained, without referring to the formalism of repeated games and Lagrangian functions. Note that the algorithm is implementable, in the sense that the outcome vector $\vec{o}_t$ revealed in each round $t$ of the \BwK problem suffices to generate full feedback for the dual algorithm.

\begin{algorithm2e}[!h]
\caption{Algorithm \MainALG for \StochasticBwK.}
\label{alg:LagrangianBwK}
\DontPrintSemicolon
\SetKwInOut{Input}{input}
\Input{parameters $B_0,T_0$, primal algorithm $\ALG_1$, dual algorithm $\ALG_2$.}
\tcp{$\ALG_1$, $\ALG_2$ are adversarial online learning algorithms}
\tcp{~~~with bandit feedback and full feedback, respectively}

\For{round $t = 1, 2, 3, \; \ldots $}{
\begin{enumerate}
\item $\ALG_1$ returns arm $a_t\in [K]$, algorithm $\ALG_2$ returns resource $i_t\in [d]$.\;

\item arm $a_t$ is chosen, outcome vector
        $\vec{o}_t = (r_t(a_t); c_{t,1}(a_t) \LDOTS c_{t,d}(a_t)) \in [0,1]^{d+1}$
    is observed.\;

\item The payoff $\mL_t(a_t,i_t)$ from \eqref{eq:IID-Lagrange-t} is reported to $\ALG_1$ as reward, and to $\ALG_2$ as cost.\;

\item The payoff $\mL_t(a_t,i)$ is reported to $\ALG_2$ for each resource $i\in [d]$.
\end{enumerate}}
\end{algorithm2e}

\subsection{Performance guarantees}

We consider algorithm \MainALG with parameter $T_0=T$. We assume the existence of the dummy resource; this is to ensure that the crucial step, \refeq{eq:IID-analysis-B}, works out even if the algorithm stops at time $T$, without exhausting any actual resources. We obtain a regret bound that is non-trivial whenever $B>\Omega(\sqrt{T})$, and is optimal, up to log factors, in the regime when $\min(\OPTDP,B)>\Omega(T)$.

\begin{theorem}\label{thm:IID}
Consider \StochasticBwK with $K$ arms, $d$ resources, time horizon $T$, and budget $B$. Assume that one resource is the dummy resource (with consumption $\tfrac{B}{T}$ for each arm).
Fix the failure probability parameter $\delta\in (0,1)$. Consider algorithm \MainALG with parameters $B_0=B$, $T_0=T$.

If EXP3.P and Hedge are used as the primal and the dual algorithms, respectively, then the algorithm achieves the following regret bound, with probability at least $1-\delta$:
\begin{align}\label{eq:thm:IID-specific}
\textstyle \OPTDP - \REW(\MainALG) \leq
    O\left( \frac{T}{B}\; \sqrt{T K  \log (dT/\delta)} \right).
\end{align}

In general, suppose each algorithm $\ALG_j$ satisfies a regret bound \eqref{eq:prelims-regret-weak} with
    $R_\delta(T) = R_{j,\delta}(T)$
and payoff range $[\rmin,\rmax] = [-\tfrac{T}{B}+1, 2]$.
Then with probability at least $1-O(\delta T)$ it holds that
\begin{align}\label{eq:thm:IID-general}
\textstyle \OPTDP - \REW(\MainALG) \leq
    O\Paren{\tfrac{T}{B}}\left(R_{1,\,\delta/T}(T) + R_{2,\,\delta/T}(T) +
	\sqrt{T \log (dT/\delta)} \right).
\end{align}
\end{theorem}

\begin{remark}	
To obtain \eqref{eq:thm:IID-specific} from the ``black-box" result \eqref{eq:thm:IID-general}, we use regret bounds in \refeq{eq:prelims-regrets}.
\end{remark}

\begin{remark}
From \cite{BwK-focs13}, the optimal regret bound for \StochasticBwK  is
\[ \OPTDP-\E[\REW] \leq \tilde{O}\left( \sqrt{K \OPTDP}\;( 1+\sqrt{\OPTDP/B}) \right).\]
Thus, the regret bound \eqref{eq:thm:IID-specific} is near-optimal if  $\min(\OPTDP,B)>\Omega(T)$, and non-trivial if $B>\Omega(\sqrt{T})$.
\end{remark}

We next prove the ``black-box"  regret bound \eqref{eq:thm:IID-general}. For the sake of analysis, consider a version of the repeated Lagrangian game that continues up to the time horizon $T$. In what follows, we separate the ``easy steps" from what we believe is the crux of the proof.

\xhdr{Notation.} Let $\vec{X}_t$ be the distribution chosen in round $t$ by the primal algorithm $\ALG_1$.
Let
    $\oX := \frac{1}{\tau} \sum_{t\in [\tau]} \vec{X}_t$
be the distribution of average play up to round $\tau$. Let
    $\vM = \E[\vM_t]$
be the expected outcome matrix. Let
    $\vec{r} = (r^{\vM}(a):\; a\in [K])$
be the vector of expected rewards over the actions. Likewise,
    $\vec{c}_i = (c_i^{\vM}(a):\; a\in [K])$
be the vector of expected consumption of each resource $i\in[d]$.

\xhdr{Using Azuma-Hoeffding inequality.}
Consider the first $\tau$ rounds, for some $\tau\in[T]$. The average reward and resource-$i$ consumption over these rounds are close to
    $\oX \cdot \vec{r}$ and $\oX \cdot \vec{c}_i$,
respectively, with high probability. Specifically, a simple usage of Azuma-Hoeffding inequality (Lemma~\ref{lem:AzumaHoeffding}) implies that
\begin{align}
\textstyle \tfrac{1}{\tau}\;\sum_{t \in [\tau]}r_t
    &\geq \oX \cdot \vec{r} - R_{0}(\tau)/\tau, \label{eq:AzumaReward} \\
\textstyle \tfrac{1}{\tau}\; \sum_{t \in [\tau]} c_{i,t}
    &\leq \oX \cdot \vec{c}_i+ R_{0}(\tau)/\tau,
    \qquad \forall i \in [d],
    \label{eq:AzumaResource}
\end{align}
hold with probability at least $1-\delta$, where $R_0(\tau) = O( \sqrt{\tau \log (d/\delta)})$.

\xhdr{Regret minimization in games.}
Let us apply the machinery from regret minimization in games to the repeated Lagrangian game. Consider the game matrix $\vG$ of the expected game. Using \refeq{eq:IID-linearity} and Lemma~\ref{lm:gameToLP}, we conclude that the minimax value of $\vG$ is
    $v^* = \OPTLP(\vM,B,T)$.

We apply Lemma~\ref{lem:learningGames}, with a fixed stopping time $\tau\in [T]$. Recall that the payoff range is
    $\rmax-\rmin = \tfrac{T}{B}+1$.
Thus, with probability at least $1-2\delta$ it holds that
\begin{align}\label{eq:IID-games-play}
\vec{\lambda}\in \Delta_{d}:\quad
\oX^\tran\, \vG\, \vec{\lambda}
    \geq v^* - \tfrac{1}{\tau} (\tfrac{T}{B}+1)\cdot \reg(T),
\end{align}
where the regret term is
    $\reg(T) := R_{1,\,\delta/T}(T)+R_{2,\,\delta/T}(T)
                     +  4\sqrt{2T \log (T/\delta)}$.

\xhdr{Crux of the proof.}
Let us condition on the event that \eqref{eq:AzumaReward}, \eqref{eq:AzumaResource}, and \eqref{eq:IID-games-play} hold for each $\tau\in [T]$. By the union bound, this event holds with probability at least
    $1-3\delta T$.

Let $\tau$ denote the \emph{stopping time} of the algorithm, the first round when the total consumption of some resource exceeds its budget. Let $i$ be the resource for which this happens; hence,
\begin{align}\label{eq:IID-analysis-B}
\textstyle    \sum_{t\in [\tau]}\; c_{i,t} >B.
\end{align}
Let us use \refeq{eq:IID-games-play} with
    $\vec{\lambda} = \vec{\lambda}^{(i)}$,
the point distribution for this resource. Then
\begin{align*}
 \oX^\tran\, \vG\, \vec{\lambda}^{(i)}
    &= \myLag{\vM}{B}{T}(\oX,\vec{\lambda}^{(i)})
        &\EqComment{by \refeq{eq:IID-linearity}} \\
    &= \oX \cdot \vec{r} +1 - \tfrac{T}{B}\; \oX \cdot \vec{c}_i
        &\EqComment{by definition of Lagrange function} \\
    &\leq \textstyle
        \frac{1}{\tau}\; \rbr{
            \rbr{\sum_{t \in [\tau]}r_t}
            - \rbr{ \tfrac{T}{B}\, \sum_{t \in [\tau]} c_{i,t} }
            +\tau + (1+\tfrac{T}{B})\,R_{0}(\tau)
        }
         &\EqComment{plugging in \eqref{eq:AzumaReward} and \eqref{eq:AzumaResource}} \\
    &\leq
    \textstyle
        \frac{1}{\tau}\; \rbr{
            \rbr{\sum_{t \in [\tau]}r_t} + \tau - T
        + (1+\tfrac{T}{B})\,R_{0}(\tau)
        }.
         &\EqComment{plugging in \refeq{eq:IID-analysis-B}}
\end{align*}
Plugging this into \refeq{eq:IID-games-play} and rearranging, we obtain
\begin{align*}
\textstyle
\sum_{t \in [\tau]}r_t
    \geq \tau\,v^* + T-\tau - (1+\tfrac{T}{B})\cdot \reg(T) -  (1+\tfrac{T}{B})\,R_{0}(\tau).
\end{align*}

Since $v^*\leq 1$ (because $v^*=\OPTLP$, as we've proved above),
\begin{align*}
\textstyle
\REW(\MainALG) = \sum_{t \in [\tau]}r_t
    \geq T\,v^* - (1+\tfrac{T}{B})\cdot \reg(T) -(1+\tfrac{T}{B})\,R_{0}(\tau).
\end{align*}
The claimed regret bound \eqref{eq:thm:IID-general} follows by \refeq{eq:IID-benchmarks}, completing the proof of Theorem~\ref{thm:IID}.

\section{A simple algorithm for \AdvBwK}
\label{sec:adversarial}

\newcommand{\Scale}{} 
	
We present and analyze an algorithm for \AdvBwK which achieves
    $d\cdot \log T$
competitive ratio, in expectation, up to a low-order additive term. Our algorithm is very simple: we randomly guess the value of $\OPTFD$ and run \MainALG with parameter $T_0$ driven by this guess. The analysis is very different, however, since we cannot rely on the machinery from regret minimization in stochastic games. The crux of the analysis (Lemma~\ref{lm:adv-crux}) is re-used to analyze the high-probability algorithm in the next section.

In hindsight, the intuition for our algorithm can be explained as follows. Since \MainALG builds on adversarial online learning algorithms $\ALG_j$, it appears plausibly applicable to \AdvBwK. We analyze it for an arbitrary parameter $T_0$, and find that it performs best when $T_0$ is tailored to $\OPTFD$ up to a constant multiplicative factor. This is precisely what our algorithm achieves using the random guess.
	
\OMIT{ 
\xhdr{Motivation.} Here we briefly describe the intuition behind the guess $\guess$ that satisfies Property~\ref{eq:Guess}.  The key difficulty is that the algorithm doesn't know a-priori the value $\tau=\tau^*$. This is the intuition behind why the algorithm needs to make a \emph{guess}. We now justify the origin behind guessing an approximation to $\OPTFD$.

	Consider the Lagrange function in \refeq{eq:LagrangianGeneral}. Recall that
	\[
		\max_{1 \leq \tau \leq T} \tau \OPT_{\LP}(\tau) = \max_{1 \leq \tau \leq T} \max_{\vec{X} \in \Delta_K} \min_{\vec{\lambda} \in \Delta_d} \tau \mL_{\vM_\tau}(\vec{X}, \vec{\lambda}).
	\]

	Thus for any given $\vec{X} \in \Delta_K$ and $\vec{\lambda} \in \Delta_d$, is it illuminating to re-write $\tau^* \ast \mL_{\vM_{\tau^*}}(\vec{X}, \vec{\lambda})$ as follows.
	\begin{equation}
		\label{eq:LagrangeAdv}
		\sum_{t=1}^{\tau^*} \sum_{a=1}^K X(a) r_t(a) + \sum_{i=1}^d \lambda_i \left(1- \sum_{t=1}^{\tau^*} \sum_{a=1}^K X(a) \frac{c_{t, i}(a)}{B} \right).
	\end{equation}
	We will now show that $\sum_{i=1}^d \lambda_i$ can be larger than $1$ and in-fact as large as $\OPT_{\LP}(\tau^*)$. This follows by looking at the corresponding dual program to $\LP(\tau^*)$. Ignoring the constraint corresponding to the dummy resource (since this resource has $0$ consumption on all arms) this can be written as follows.

\begin{equation*}
\begin{array}{ll@{}ll}
\tag{Benchmark-LP-Dual}
\label{lp:dual_single}
\displaystyle \text{minimize} & \displaystyle \qquad \sum_{i=1}^{d-1} \lambda_i + \nu \\
\displaystyle \forall a \in [m] & \displaystyle \qquad \left( \sum_{t=1}^{\tau^*} \sum_{i=1}^{d-1} c_{t, i}(a)/B \right) \lambda + \nu \geq \sum_{t=1}^{\tau^*} r_t(a) \\
 & \displaystyle \qquad  0 \leq \lambda_1, \lambda_2, \ldots, \lambda_{d-1}, \nu
\end{array}
\end{equation*}

From strong duality we have that for an optimal dual solution pair $(\vec{\lambda}^*, \nu^*)$ we have $\sum_{i=1}^{d-1} \lambda_i^* + \nu^* = \OPT_{\LP}(\tau^*)$. However our dual algorithm chooses distributions and hence the sum is $1$. To remedy this difference in scale, we can re-write $\sum_{i=1}^{d-1} \lambda_i + \nu$ as $\OPT_{\LP}(\tau^*) \Paren{\sum_{i=1}^d \tilde{\lambda}_i}$ where $\sum_{i=1}^d \tilde{\lambda}_i = 1$. Thus if we guessed the value of $\OPT_{\LP}(\tau^*)$ (or equivalently $\OPTFD$) we could then learn the optimal distribution $\vec{\tilde{\lambda}}$ via the dual algorithm. We don't explicitly use these observations in our analysis, but only use it to motivate our algorithm.
} 

Our algorithm is presented as Algorithm~\ref{alg:LagrangianBwKAdv}. We guess the value of $\OPTFD$ within a given range $[\Gmin,\Gmax]$. We guess $\OPTFD$ uniformly on the ``exponential scale": we draw the exponent $u$ uniformly at random, and define the guess as $\guess = \Gmin\cdot \kappa^u$, for some scale parameter $\kappa>1$.%
\footnote{Somewhat surprisingly, our results do not depend on the value $\kappa$. This is because the dependence on $\kappa$ is captured via a normalized integral
    $\ln \kappa \cdot \int_0^{\log_\kappa x} \kappa^u du$,
and this expression does not depend on $\kappa$.}
 We call \MainALG with $T_0 = \guess/(d+1)$. Our analysis works as long as $\OPTFD \leq \Gmax$, with $\Gmin$ appearing in the additive term.

\begin{algorithm2e}[h]
\caption{A simple algorithm for \AdversarialBwK.}
\label{alg:LagrangianBwKAdv}
\DontPrintSemicolon
\SetKwInOut{Input}{input}
\Input{scale parameter $\kappa>1$, guess range $[\Gmin, \Gmax]$, primal and dual algorithms $\ALG_1$, $\ALG_2$}
\tcp{$\ALG_1$, $\ALG_2$ are adversarial online learning algorithms}
\tcp{~~~with bandit feedback and full feedback, resp.}
Choose $u$ uniformly at random from $[0, \RndMax]$, where
    $\RndMax = \log_{\kappa} \tfrac{\Gmax}{\Gmin} $.\;
Guess the value of $\OPTFD$ as $\guess = \Gmin\cdot\kappa^u$.\;
Run \MainALG with algorithms $\ALG_1$, $\ALG_2$ and parameters $B_0=B$ and $T_0=\guess/(d+1)$.
\end{algorithm2e}

\begin{theorem}\label{thm:AdvBwK-main}
Consider \AdvBwK with $K$ arms, $d$ resources, time horizon $T$, and budget $B$. Assume that one of the arms is a \emph{null arm} that has zero reward and zero resource consumption. Consider Algorithm~\ref{alg:LagrangianBwKAdv} with scale parameter $\kappa > 1$. Suppose algorithms $\ALG_j$ that satisfy the regret bound \eqref{eq:prelims-regret-weak} with $\delta = T^{-2}$ and regret term
    $R_\delta(T) = R_{j,\delta}(T)$,
for any known payoff range
    $[\rmin,\rmax]$.

\begin{itemize}
\item[(a)] If $\OPTFD \leq \Gmax$ then the expected reward of Algorithm~\ref{alg:LagrangianBwKAdv} satisfies
\begin{align}\label{eq:adv-thm-range}
\E[\REW]
    \geq \frac{\OPTFD - \Gmin}
        { (d+1)\;
        \ln \left( \frac{\Gmax}{\Gmin} \right)} -\reg-1,
\end{align}
where
$\reg = (1+\tfrac{\OPTFD}{d B} )
            \left( R_{1,\,\delta/T}(T) + R_{2,\,\delta/T}(T)\right)$.

\item[(b)] In particular, taking $[\Gmin,\Gmax] = [\sqrt{T},T]$,
we obtain
\begin{align}\label{eq:adv-thm-logT}
\E[\REW] \geq \frac{\OPTFD-\sqrt{T}}
    {\tfrac12\, (d+1)\; \ln(T)} -\reg-1.
\end{align}
\end{itemize}
\end{theorem}

\begin{remark}
One can use algorithms EXP3.P for $\ALG_1$ and Hedge for $\ALG_2$, with regret bounds given by \eqref{eq:prelims-regrets}, and achieve the regret term
    $\reg = O\Paren{1+ \tfrac{\OPTFD}{d B}}\;
        \sqrt{T K \log (Td/\delta)}$.
We obtain a meaningful performance guarantee as long as, say, $\reg<\OPTFD/2$; this requires
   $\OPTFD$ and $B$
to be at least
    $\widetilde{\Omega}(\sqrt{TK})$.
\end{remark}

\OMIT{ 
\begin{remark}
In retrospect, the intuition for our algorithm can be explained as follows. \MainALG builds on adversarial online learning algorithms $\ALG_i$, and appears plausibly applicable to the adversarial BwK problem. We analyze it for an arbitrary parameter $T_0$ in Lemma~\ref{lm:adv-crux} (which is the crux of the proof), and find that it performs best when $T_0$ is tailored to $\OPTFD$, up to a constant multiplicative factor. A simple way to achieve the latter is to guess $\OPTFD$ randomly, as in Algorithm~\ref{alg:LagrangianBwKAdv}.
\end{remark}
} 

\begin{remark}
We define the outcome matrices slightly differently compared to Section~\ref{sec:IID} in that we do not posit a dummy resource. Formally, we assume that the null arm has zero consumption in every resource. This is essential for case 1 (\emph{i.e.,} when $\taualg \leq \sigma)$ in the analysis of Lemma~\ref{lm:adv-crux}.
\end{remark}

\begin{remark}\label{rem:factor-d}
The $\log(T)$ appears in the competitive ratio because the algorithm needs to guess $\OPTFD$ up to a constant factor. The factor of $d$ can be traced to a pessimistic over-estimate in \eqref{eq:adv-crux-case2-dB}.
\end{remark}

\begin{remark}\label{rem:simplifies}
The algorithm simplifies when $d=1$, \ie if there is only one resource other than the dummy resource. Then the outcome matrices have only one resource, so the dual algorithm $\ALG_2$ is no longer needed.
\end{remark}

\begin{remark}\label{rem:reduction}
The problem can be reduced to the case $d=1$, which simplifies the algorithm, as per Remark~\ref{rem:simplifies}, but increases the competitive ratio. The reduction is very simple: replacing all ``true resources" (\ie all resources other than the dummy resource) with the ``maximal resource" whose consumption is the maximum over the true resources.
The competitive ratio, \ie the denominator in \refeq{eq:adv-thm-range}, increases by the factor of
    $\frac{2d}{d+1}$.
Moreover, the reduction can be wasteful if the maximal consumption (across all resources) is much larger than a ``typical" consumption of each resource. The analysis compares algorithm's reward to the benchmark for the ``fake problem" with $d=1$, then compares the said benchmark to $\OPTFD$. The former step is essentially the analysis in  Section~\ref{sec:adv-analysis}, albeit in a slightly simpler form. We omit the easy details.
\end{remark}

If a problem instance of \AdvBwK is actually an instance of adversarial bandits, then we recover the optimal $\tilde{O}(\sqrt{KT})$ regret. (This easily follows by examining the proof of Lemma~\ref{lm:adv-crux}.)


\begin{lemma}\label{lm:adv-recover}
Consider \MainALG, with algorithms EXP3.P for $\ALG_1$ and Hedge for $\ALG_2$,
for an instance of \AdvBwK with zero resource consumption. This algorithm obtains $\tilde{O}(\sqrt{KT})$ regret, for any parameters $B_0,T_0>0$. Accordingly, so does Algorithm~\ref{alg:LagrangianBwKAdv} with any scale parameter $\kappa>0$.
\end{lemma}

\subsection{Analysis: proof of Theorem~\ref{thm:AdvBwK-main} and Lemma~\ref{lm:adv-recover}}
\label{sec:adv-analysis}

\xhdr{Stopped linear program.}
Let us set up a linear relaxation that is suitable to the adversarial setting. The expected outcome matrix is no longer available. Instead, we use \emph{average} outcome matrices:
\begin{align}\label{eq:adv-average-M}
    \textstyle  \bvM_\tau = \tfrac{1}{\tau}\; \sum_{t\in [\tau]}\; \vM_t,
\end{align}
the average up to a given intermediate round $\tau\in [T]$. Similar to the stochastic case, the relaxation assumes that each instantaneous outcome matrix $\vM_t$ is equal to the average outcome matrix $\bvM_\tau$. What is different now is that the relaxation depends on $\tau$: using $\bvM_\tau$ is tantamount to stopping precisely at this round.

With this intuition in mind, for a particular end-time $\tau$ we consider the linear program \eqref{lp:primalAbstract}, parameterized by the time horizon $\tau$ and the average outcome matrix $\bvM_\tau$.
Its value,
    $\OPTLP(\bvM_\tau, B,\tau)$,
represents the per-round expected reward, so it needs to be scaled by the factor of $\tau$ to obtain the total expected reward. Finally, we maximize over $\tau$. Thus, our linear relaxation for \AdvBwK is defined as follows:
\begin{align}\label{eq:adv-stoppedLP}
\textstyle
\OPTLPfull := \max_{\tau\in [T]}\;\tau\cdot \OPTLP(\bvM_\tau, B,\tau)
    \geq \OPTFD.
\end{align}
The inequality in \eqref{eq:adv-stoppedLP} is proved in the appendix (Section~\ref{sec:appxAdversarial}).


\xhdr{Regret bounds for $\ALG_j$.}
Since each algorithm $\ALG_j$, $j\in \{1,2\}$ satisfies regret bound \eqref{eq:prelims-regret-weak} with $\delta = T^{-2}$ and
    $R_\delta(T) = R_{j,\delta}(T)$,
it also satisfies a stronger version~\eqref{eq:prelims-regret-strong} with the same parameters.
Recall from \eqref{eq:IID-Lagrange-t} that the payoff range is
    $[\rmin,\rmax] = [-\tfrac{T_0}{B}+1, 2]$.
For succinctness, let
    $U_j(T|T_0) = (1+\tfrac{T_0}{B})\, R_{j,\,\delta/T}(T)$
denote the respective regret term in \eqref{eq:prelims-regret-strong}.

Let us apply these regret bounds to our setting.  Let $a_t\in [K]$ and $i_t\in [d]$ be, resp., the chosen arm and resource in round $t$. We represent the outcomes as vectors over arms:
    $\vec{r}_t, \vec{c}_{t,i}\in [0,1]^K$
denote, resp., reward vector and resource-$i$ consumption vector for a given round $t$. Recall that the round-$t$ payoffs in \MainALG are given by the Lagrange function
    $\mL_t :=\myLag{\vM_t}{B}{T_0}$
such that
\begin{align}\label{eq:ADV-Lagrange}
    \mL_t(a,i) = r_t(a) + 1- \tfrac{T_0}{B}\; c_{t,i}(a)
\end{align}
for each arm $a$ and resource $i$.
Consider the total Lagrangian payoff at a given round $\tau\in [T]$:
\begin{align}\label{eq:adv-total-payoff}
\textstyle
\sum_{t\in [\tau]} \mL_t(a_t,i_t)
= \REW_\tau +\tau - W_\tau,
\end{align}
where
    $\REW_\tau =\sum_{t\in[\tau]} r_t(a_t)$
is the total reward up to round $\tau$, and
    $W_\tau = \tfrac{T_0}{B} \sum_{t\in[\tau]} c_{t,i_t}(a_t)$
is the \emph{consumption term}. The regret bounds sandwich  \eqref{eq:adv-total-payoff} from above and below:
\begin{align}\label{eq:adv-regrets}
\left(\max_{a\in [K]}\; \sum_{t\in [\tau]} \mL_t(a,i_t)\right) - U_1(T|T_0)
\leq \REW_\tau +\tau - W_\tau
\leq \left(\min_{i\in [d]}\; \sum_{t\in [\tau]} \mL_t(a_t,i)\right) + U_2(T|T_0).
\end{align}
This holds for all $\tau\in [T]$, with probability at least $1-2\delta$. The first inequality in \eqref{eq:adv-regrets} is due to the primal algorithm, and the second is due to the dual algorithm. Call them \emph{primal} and \emph{dual} inequality, respectively.

\xhdr{Crux of the proof.}
We condition on the event that \eqref{eq:adv-regrets} holds for all $\tau\in [T]$, which we call the \emph{clean event}.
The crux of the analysis is encapsulated in the following lemma, which analyzes an execution of \MainALG with an arbitrary parameter $T_0$ under the clean event.

\begin{lemma}\label{lm:adv-crux}
Consider an execution of \MainALG with $B_0=B$ and an arbitrary parameter $T_0$ such that the clean event holds. Fix an arbitrary round $\sigma\in [T]$, and consider the LP value relative to this round:
\begin{align}\label{eq:lm:adv-crux-benchmark}
f(\sigma) := \OPTLP(\bvM_\sigma, B,\sigma).
\end{align}
The algorithm's reward up to round $\sigma$ satisfies
\begin{align}\label{eq:lm:adv-crux-sigma}
    \REW_\sigma \geq \min(T_0,\, \sigma\cdot f(\sigma) -dT_0) -
        \left(\; U_1(T|T_0) + U_2(T|T_0) \;\right).
\end{align}
Taking $\sigma$ to be the maximizer in \eqref{eq:adv-stoppedLP}, algorithm's reward satisfies
\begin{align}\label{eq:lm:adv-crux}
    \REW \geq \min(T_0,\OPTFD-dT_0) -
        \left(\; U_1(T|T_0) + U_2(T|T_0) \;\right).
\end{align}
\end{lemma}

\refeq{eq:lm:adv-crux-sigma} is used, with a different $\sigma$, for the high-probability analysis in Section~\ref{sec:HP}.

\begin{proof}
Let $\taualg$ be the stopping time of the algorithm. We consider two cases, depending on whether some resource is exhausted at time $\sigma$. In both cases, we focus on the round
    $\min(\taualg,\sigma)$.

\textbf{Case 1: $\taualg\leq\sigma$ and some resource is exhausted.}
Let us focus on round $\tau=\taualg$. If $i$ is the exhausted resource, then
    $\sum_{t \in [\tau]} c_{t, i}(a_t) > B$.
Let us apply the dual inequality in \eqref{eq:adv-regrets} for this resource:
\begin{align*}
\REW_\tau +\tau - W_\tau - U_2(T|T_0)
    &\leq \textstyle  \sum_{t\in [\tau]} \mL_t(a_t,i) \\
    &= \textstyle \textstyle  \REW_\tau + \tau -
        \tfrac{T_0}{B} \sum_{t \in [\tau]} c_{t, i}(a_t) \\
    & \leq \REW_\tau + \tau - T_0.
\end{align*}
It follows that $W_\tau \geq T_0-U_2(T|T_0)$.

Now, let us apply the primal inequality in \eqref{eq:adv-regrets} for the null arm. Recall that the reward and consumption for this arm is $0$, so $\mL_t(\term{null}, i_t) = 1$ for each round $t$. Therefore,
\begin{align*}
\REW_\tau +\tau - W_\tau + U_1(T|T_0)
     &\geq \textstyle  \sum_{t\in [\tau]} \mL_t(\term{null},i_t)
     = \tau.
\end{align*}
We conclude that
    $\REW_\tau \geq  W_\tau - U_1(T|T_0) \geq T_0 - U_1(T|T_0) - U_2(T|T_0)$.

\textbf{Case 2: $\taualg \geq \sigma$.}
Let us focus on round $\sigma$. Consider the linear program
    $\myLP{\bvM_\sigma}{B}{\sigma}$,
and let $\vec{X^*}\in \Delta_K$ be an optimal solution to this LP. The primal inequality in \eqref{eq:adv-regrets} implies that
\begin{align}
\REW_\sigma +\sigma - W_\sigma + U_1(\sigma)
    &\geq \textstyle
        \max_{a\in [K]}\; \sum_{t\in [\sigma]} \mL_t(a,i_t) \nonumber\\
    &\geq \textstyle \sum_{t\in [\sigma]} \sum_{a\in [K]} X^*(a)\;\mL_t(a,i_t) \nonumber\\
    &= \textstyle
        \sigma +
        \sum_{t\in [\sigma]} \vec{X^*}\cdot \vec{r}_t -
        \tfrac{T_0}{B}\,\sum_{t\in [\sigma]} \vec{X^*}\cdot \vec{c}_{t,i_t}\nonumber\\
\REW_\sigma
    &\geq \textstyle \sigma\cdot f(\sigma) -
        \tfrac{T_0}{B}\, \sum_{t\in [\sigma]} \vec{X^*}\cdot \vec{c}_{t,i_t} - U_1(T|T_0).
            \label{eq:adv-crux-case2}
\end{align}
In the last inequality we used the fact that
    $\sum_{t \in [\sigma]} \vec{X^*}\cdot \vec{r_t} = \sigma\cdot f(\sigma)$
by optimality of $\vec{X^*}$.

    $\sum_{t\in [\sigma]}\vec{X^*}\cdot \vec{c}_{t,i}\leq B$
for each resource $i$, since $\vec{X^*}$ is a feasible solution for
    $\OPTLP(\bvM_\sigma, B,\sigma)$.
Then,
\begin{align}\label{eq:adv-crux-case2-dB}
\textstyle
\sum_{t\in [\sigma]} \vec{X^*}\cdot \vec{c}_{t,i_t}
    \leq  \textstyle
        \sum_{i\in[d]}\;\sum_{t\in [\sigma]} \vec{X^*}\cdot \vec{c}_{t,i}
    \leq dB.
\end{align}
Plugging \eqref{eq:adv-crux-case2-dB} into \eqref{eq:adv-crux-case2}, we conclude that
    $\REW_\sigma \geq \textstyle \sigma\cdot f(\sigma) -dT_0-U_1(T|T_0)$.

Conclusions from the two cases imply \eqref{eq:lm:adv-crux}, as claimed.
\end{proof}

\xhdr{Wrapping up (the easy version).}
$\OPTFD \in [\Gmin,\Gmax]$, then some guess $\guess$ is approximately correct:
\begin{align}\label{eq:adv:wrapup:guess}
 \OPTFD/\kappa \leq \guess \leq \OPTFD.
\end{align}
By Lemma~\ref{lm:adv-crux}, the algorithm's execution with this guess, assuming the clean event, satisfies \eqref{eq:lm:adv-crux}, where, recalling that $T_0 = \guess/(d+1)$, we have
\[ \min(T_0,\OPTFD-dT_0) \geq \frac{\OPTFD}{\kappa (d+1)}
\quad\text{and}\quad
T_0 \leq \frac{\OPTFD}{d+1}.\]

\noindent The regret term for this guess is
\begin{align*}
\reg = U_1(T|T_0) + U_2(T|T_0)
    \leq (1+\tfrac{\OPTFD}{(d+1)\,B})\,
        (R_{1,\,\delta/T}(T) + R_{2,\,\delta/T}(T)).
\end{align*}

\noindent To complete the proof of \eqref{eq:adv-thm-range} (with a much  larger constant in the denominator), note that we obtain a suitable guess $\guess$ with probability
    $1/\left\lceil \log_{\kappa} \tfrac{\Gmax}{\Gmin} \right\rceil$.

\xhdr{Wrapping up (optimizing the constants).}
Let us go beyond the ``approximately right guess" in \eqref{eq:adv:wrapup:guess}, and account for contributions of \emph{every} guess. In other words, let us integrate over the guesses.

Assume that $\OPTFD$ lies in the guess range $[\Gmin,\Gmax]$. Recall that the algorithm samples $u$ uniformly at random from the interval $[0,\RndMax]$. Write
\begin{align*}
    T_0  &= T_0(u) = \Gmin\cdot \kappa^u/(d+1),\\
    \reg &= \reg(u) = U_1(T|T_0(u)) + U_2(T|T_0(u)),\\
    \Lambda(u) &= \min(T_0(u),\max(0,\OPTFD-dT_0(u))).
\end{align*}
Then by Lemma~\ref{lm:adv-crux}, for a particular choice of $u$, the algorithm's reward satisfies
\begin{align}\label{eq:adv:wrapup:crux}
    \E[\REW \mid u] \geq \Lambda(u) - \reg(u)-1,
\end{align}
where the `-1' term accounts for the complement of the ``clean event".

The $\Lambda(u)$ term can be split into three cases as follows:
	\begin{align}\label{eq:adv:wrapup:split}
    \Lambda(u) = \begin{cases}
       T_0(u) &\quad\text{if }
       0 \leq u \leq
        \log_{\kappa} \tfrac{\OPTFD}{\Gmin},\\
      \OPTFD - d T_0(u) &\quad\text{if }
        \log_{\kappa} \tfrac{\OPTFD}{\Gmin}
        < u \leq \log_{\kappa} \Paren{\tfrac{d+1}{d} \cdot \tfrac{\OPTFD}{\Gmin}} ,\\
       0 &\quad\text{otherwise}.
     \end{cases}
	\end{align}
\noindent So, we are only interested in
    $u\leq u^*:= \log_{\kappa} \Paren{\tfrac{\OPTFD}{\Gmin}}$.

Integrating the right-hand side of \eqref{eq:adv:wrapup:crux} over $u$, we obtain:
\begin{align}\label{eq:adv:wrapup:int}
\E[\REW]
    \geq \frac{1}{\RndMax} \int_0^{u^*} \E[\REW\mid u]\; \dd u
    = \frac{1}{\RndMax} \int_0^{u^*}
        \Paren{\Lambda(u) - \reg(u) - 1}\; \dd u,
\end{align}
where
    $\RndMax = \log_{\kappa} \tfrac{\Gmax}{\Gmin} $
as per the algorithm's specification.

Using \eqref{eq:adv:wrapup:split} and omitting the easy details, the main term $\Lambda(u)$ integrates as follows:
\begin{align}\label{eq:adv:wrapup:int-main}
 \int_0^{u^*} \Lambda(u)\, \dd u
    \geq \frac{\OPTFD-\Gmin}{(\ln \kappa) (d+1)}.
 \end{align}
(We've only used the first ``regime" in \eqref{eq:adv:wrapup:split}. As for the second "regime" in \eqref{eq:adv:wrapup:split}, integrating over it only improves \eqref{eq:adv:wrapup:int-main} by a small additive term.)

To handle the regret term $\reg(u)$, note that it is non-decreasing with $u$, so
\begin{align*}
\frac{1}{\RndMax} \int_0^{u^*} \reg(u)\, \dd u
    \leq \frac{u^*\cdot \reg(u^*)}{\RndMax}
    \leq \reg(u^*).
 \end{align*}

\noindent Plugging this into \eqref{eq:adv:wrapup:int}, we obtain
 \begin{align}\label{eq:adv:wrapup:final}
\E[\REW]
    \geq \frac{1}{\RndMax} \left( \frac{\OPTFD-\Gmin}{(\ln \kappa) (d+1)} \right) - \reg(u^*) - 1.
\end{align}
Recalling that $T_0 = T_0(u^*) = \OPTFD/(d+1)$, we have
\begin{align*}
\reg(u^*) = U_1(T|T_0) + U_2(T|T_0)
    \leq (1+\tfrac{\OPTFD}{(d+1)\,B})\,
        (R_{1,\,\delta/T}(T) + R_{2,\,\delta/T}(T)).
\end{align*}
Finally, because of the $-\Gmin$ term in \eqref{eq:adv:wrapup:final}, the assumption $\OPTFD\geq \Gmin$ is redundant. This completes the proof of Theorem~\ref{thm:AdvBwK-main}(a).

\OMIT{ 
    \kaedit{Finally, to obtain the corollary, we minimize w.r.t $\kappa > 1$. Note that when $[\Gmin, \Gmax] = [1, T]$ the competitive ratio is $\kappa (d+1) \Cel{\log T/\log \kappa}$. Thus, we want to minimize the expression $\kappa/\log \kappa$ for $\kappa > 1$. Differentiating w.r.t. $\kappa$ and setting it to $0$, we have that the minimizer is when $\kappa = e$. Thus, the competitive ratio at this minimizer is $e (d+1) \Cel{\log T}$.}
} 

\xhdr{Proof Sketch of Lemma~\ref{lm:adv-recover}.}
Recall that in the adversarial bandit setting we have $\vec{c}_{i, t} = 0$ for every $i \in [d]$ and every $t \in [T]$. We re-analyze Lemma~\ref{lm:adv-crux} with $\sigma=T$. Notice that case 1 never occurs. Thus we obtain obtain \refeq{eq:adv-crux-case2} in case 2. Note that $\tfrac{T_0}{B}\, \sum_{t\in [\sigma]} \vec{X^*}\cdot \vec{c}_{t,i_t} = 0$ since $\vec{c}_{i, t}=0$. Therefore, we obtain
	\[ \REW_T \geq T \cdot f(T) - U_1(T|T_0). \]	

	We now argue that $T \cdot f(T) = \max_{a \in [K]} \sum_{t \in [T]} r_t(a)$. Let $\vec{X}^*$ be the optimal distribution over the arms. Thus $\sum_{t \in [T]} \vec{X^*}\cdot \vec{r_t} = T\cdot f(T)$. Note that since $\vec{c}_{i, t}=0$ the only constraint on $\vec{X}^*$ is that it lies in $\Delta_K$. Therefore the maximizer is a point distribution on $\max_{a \in [K]} \sum_{t \in [T]} r_t(a)$. This proof does not rely on any specific value for $B_0, T_0$. The payoff range is
    $[\rmax,\rmin] = [1, 2]$,
so
    $U_1(T|T_0) = \tilde{O}\Paren{\sqrt{K T}}$.

\section{High-probability algorithm for \AdvBwK}
\label{sec:HP}

\newcommand{\tvM}{\widetilde{\vM}}
\newcommand{\ips}{\term{ips}}

We recover the $O(d\log T)$ competitive ratio for \AdvBwK, but with high probability rather than merely in expectation. Our algorithm uses \MainALG as a subroutine, and re-uses the adversarial analysis thereof (Lemma~\ref{lm:adv-crux}). We do not optimize the regret term or the constant in the competitive ratio.

The algorithm is considerably more complicated compared to Algorithm~\ref{alg:LagrangianBwKAdv}. Instead of making one random guess $\guess$ for the value of $\OPTLPfull$, we iteratively refine this guess over time. The algorithm proceeds in phases. In the beginning of each phase, we start a fresh instance of \MainALG with parameter $T_0$ defined by the current value of $\guess$.%
\footnote{The idea of restarting the algorithm in each phase is similar to the standard ``doubling trick" in the online machine learning literature, but much more delicate in our setting.}
We update the guess $\guess$ in each round (in a way specified later), and stop the phase once $\guess$ becomes too large compared to its initial value in this phase. We invoke \MainALG with a rescaled budget
    $B_0 = B/\Theta \Paren{\log T }$.
Within each phase, we simulate the \BwK problem with budget $B_0$: we stop \MainALG once the consumption of some resource in this phase exceeds $B_0$. For the remainder of the phase, we play the null arm with probability $1-\gamma_0$ and do uniform exploration with the remaining probability, for some parameter $\gamma_0\in (0,1)$ (here and elsewhere, \emph{uniform exploration} refers to choosing each action with equal probability). The pseudocode is summarized in Algorithm~\ref{alg:LagrangianBwKAdvHighP}.

\begin{algorithm2e}[!h]
\caption{High-probability algorithm for \AdvBwK.}
\label{alg:LagrangianBwKAdvHighP}
\DontPrintSemicolon
\SetKwInOut{Input}{input}
\Input{scale parameter $\kappa$, exploration parameter $\gamma_0$, primal algorithm $\ALG_1$, dual algorithm $\ALG_2$}
\tcp{$\ALG_1$, $\ALG_2$ are adversarial online learning algorithms}
\tcp{~~~with bandit feedback and full feedback, resp.}

Initialize $\guess =1$.\;
\For{each phase}{
Start a fresh instance \ALG of \MainALG\;
~~~~with parameters $B_0=B/2\Cel{\log_\kappa T}$ and
    $T_0=\guess/\left( \revedit{3}\,d\lceil \log_{\kappa} T \rceil\right)$. \;
Define $\guessold := \guess$.\;
\For{each round in this phase}{
Recompute the guess $\guess$ \;
{\bf if} $\guess > \kappa \cdot \guessold$ {\bf then} start a new phase\;
\uIf{consumption of all resources in this phase does not exceed $B_0$}
    {Play the action chosen by \ALG,
    observe the outcome and report it back to \ALG.
     }
\Else{Choose the null arm with probability $1-\gamma_0$, do uniform exploration otherwise}
}}
\end{algorithm2e}

To complete algorithm's specification, let us define how to update the guess $\guess$ in each round $t$. The guess, denoted $\guess_t$, is an estimate for
    $\OPTLPfull[t]$,
-as defined in \eqref{eq:adv-stoppedLP}. We form this estimate using a standard \emph{inverse propensity scoring} (\emph{IPS}) technique. Let $\vec{p}_t$ and $a_t$ be, resp., the distribution and the arm chosen by the primal algorithm in round $t$. The instantaneous outcome matrix $\vM_t$ is estimated by matrix $\vM_t^\ips\in [0,\infty)^{K\times d}$ such that each row $\vM_t^\ips(a)$ is defined as follows:
	\begin{equation}\label{eq:IPSEst}
		 \vM_t^\ips(a) := \indicator{a_t=a}\; \tfrac{1}{p_t(a_t)}\; \vM_t(a).
	\end{equation}

For a given end-time $\tau$, the average outcome matrix $\bvM_\tau$ from \eqref{eq:adv-average-M} is estimated as
\begin{equation}\label{eq:IPSEst-ave}
\bvM_\tau^\ips := \tfrac{1}{\tau}\; \textstyle \sum_{t\in [\tau]} \vM_t^\ips.
\end{equation}

Finally, we plug this estimate into \eqref{eq:adv-average-M} and define
\begin{align}\label{eq:hp-guess}
\textstyle
\guess_t := \max_{\tau\in [t]}\;\tau\cdot \OPTLP(\bvM_\tau^\ips, B,\tau).
\end{align}
For the analysis, we will assume that the primal algorithm does some uniform exploration:
\begin{align}\label{eq:hp-explore}
\textstyle
p_t(a) \geq \gamma>0 \quad \text{for each arm $a\in [K]$ and each round $t\in [T]$}.
\end{align}

\begin{theorem}\label{thm:AdvBwK-HP}
Consider \AdvBwK with $K$ arms, $d$ resources, time horizon $T$, and budget $B$. Let $\delta>0$ be the failure probability parameter. Assume that $B > 5 K \, T^{3/4} \log(2T^2)$.
Suppose that one of the arms is a \emph{null arm} that has zero reward and zero resource consumption.

Consider Algorithm~\ref{alg:LagrangianBwKAdvHighP} with parameters
    $\kappa = 2$ and $\gamma_0 = T^{-1/4}$.
 Assume that each algorithm $\ALG_j$, $j\in \{1,2\}$, satisfies the regret bound \eqref{eq:prelims-regret-weak} with payoff range $[\rmin,\rmax] = [-\tfrac{T}{B}+1, 2]$
and regret term
    $R_\delta(T) = R_{j,\delta}(T)$.
Assume that the primal algorithm $\ALG_1$ satisfies \eqref{eq:hp-explore} with parameter $\gamma \geq T^{-1/4}$.

Then the total reward $\REW$ collected by Algorithm~\ref{alg:LagrangianBwKAdvHighP} satisfies
\begin{align}\label{eq:thm:AdvBwK-HP}
\Pr\left[\; \REW \geq
        \frac{\OPTFD}
            {O(d\log T) } - O(\reg)\;
\right] \geq 1-O(\delta T),
\end{align}
where the regret term is
    $\term{reg} =
        \tfrac{T}{B}\left(\;
            K \;T^{3/4}\; \log^{1/2}(\tfrac{1}{\delta})
            + R_{1,\,\delta/T}(T) + R_{2,\,\delta/T}(T)\;\right)$.
\end{theorem}

\begin{remark}
Using algorithms EXP3.P for $\ALG_1$ and Hedge for $\ALG_2$, we can achieve \eqref{eq:thm:AdvBwK-HP} with
    \[ \term{reg}
        = O\Paren{ \tfrac{TK}{B}} \; T^{3/4}\; \sqrt{\log (T/\delta)}.\]
This is because EXP3.P, with appropriately modified uniform exploration term $\gamma=T^{-1/4}$, satisfies the regret bound \eqref{eq:prelims-regret-weak} with
    $R_\delta(\tau) = O(T^{3/4}) \sqrt{K\log \tfrac{T}{\delta}}$,
and for Hedge we can (still) use \refeq{eq:prelims-regrets}. The theorem is meaningful whenever, say, $\term{reg}<\OPTFD/2$. The latter requires
    $\OPTFD \cdot \tfrac{B}{K}>\widetilde{\Omega}(T^{7/4})$.
\end{remark}

\begin{remark}
Like in Theorem~\ref{thm:AdvBwK-main}, we posit that the null arm does not consume any resources.
\end{remark}

\begin{remark}
For the sake of intuition, let us clarify the choice of parameters $B_0$ and $T_0$ in the algorithm. First, $\Cel{\log_\kappa T}$ appears in both $B_0$ and $T_0$, because it is an upper bound on the number of phases. Second, $d$ is needed in $T_0$ to counteract the dependence on $d$ in Lemma~\ref{lm:adv-crux}, the adversarial analysis of \MainALG. Third, the $\tfrac12$ appears in $B_0$ because we allow half of the budget to be spent on uniform exploration. Finally, the $\tfrac13$ in $T_0$ is needed to enable a specific step deep down in the analysis, the transition from \refeq{eq:partTwo} to \refeq{eq:HighPFinal1}.
\end{remark}

\begin{proof}[Proof Sketch of Theorem~\ref{thm:AdvBwK-HP}]

The proof consists of several steps. First, we argue that the guess $\guess_t$ is close to $\OPTLPfull[t]$ with high probability. This argument only relies on the uniform exploration property \eqref{eq:hp-guess} and the definition of IPS estimators, not on any properties of the algorithm. We immediately obtain concentration for the average outcome matrices. With some work, we derive concentration on the respective LP-values.

Next, we focus on a particular phase in the execution of the algorithm. We say that a phase is \emph{full} if the stopping condition $\guess_t> \kappa \cdot \guessold$ has fired. We focus on the last full phase. We prove there is enough reward to be collected in this phase. Essentially, letting $\tau_1,\tau_2$ be, resp., the start and end time of this phase, we consider the \BwK problem restricted to time interval $[\tau_1,\tau_2]$, and lower-bound the LP-value of this problem in terms of the LP-value of the original problem. Finally, we use the adversarial analysis of \MainALG (Lemma~\ref{lm:adv-crux}) to guarantee that our algorithm actually collects that value.

Because of the stopping condition $\guess_t>\kappa \cdot \guessold$, there can be at most $\Cel{\log_\kappa T}$ phases.  Therefore, rescaling the budget to $B_0/2\Cel{\log_\kappa T}$ guarantees that the algorithm consumes at most $B/2$ of the budget. We then argue that, with high-probability, the additional uniform exploration in each phase, consumes a budget of at most $B/2$ with high-probability. Thus, the algorithm never runs out of budget.
\end{proof}

\subsection{Full proof of Theorem~\ref{thm:AdvBwK-HP}}

\newcommand{\ratio}{\term{ratio}}

We now describe the full proof of Theorem~\ref{thm:AdvBwK-HP}, following the plan outlined in the proof sketch. We decompose the analysis into several distinct pieces, present them one by one, and then show how to put them together. Each piece is presented as a lemma, with appropriate notation and intuition. 

For clarity, most of the analysis is presented for an arbitrary parameter $\kappa>1$ in the algorithm, as long as it is an absolute constant, and an arbitrary parameter $\gamma$ in \refeq{eq:hp-explore}. We only plug in $\kappa=2$ and $\gamma\geq T^{-1/4}$ in the very end, in \refeq{eq:HighPFinal1} and right after.


\xhdr{Extended notation.}
To argue about a given phase, we extend some of our notation to refer to arbitrary time intervals, not just $[1,\tau]$. In what follows, fix time interval $[\tau_1,\tau_2]$, and let $\tDelta$$\tau = \tau_2-\tau_1+1$.
Let
\begin{align*}
\bvM_{\Brac{\tau_1, \tau_2}}
    &:= \frac{1}{\tDelta \tau}\,\sum_{t=\tau_1}^{\tau_2}\; \vM_t, \\
\bvM^{\ips}_{\Brac{\tau_1, \tau_2}}
    &:= \frac{1}{\tDelta \tau}\,\sum_{t=\tau_1}^{\tau_2}\; \vM^{\ips}_t
\end{align*}
be, resp., the average outcome matrix and its IPS-estimate on this time interval. Define
\begin{align}
\OBJ(\Brac{\tau_1, \tau_2})
    &:= \tDelta \tau \cdot \OPTLP(\bvM_{\Brac{\tau_1, \tau_2}}, B, \vec{\tDelta}\tau),	
    \label{eq:OBJ} \\
\OBJ^{\ips}(\Brac{\tau_1, \tau_2})
    &:=  \tDelta\tau \cdot \OPTLP(\bvM^{\ips}_{\Brac{\tau_1, \tau_2}}, B, \vec{\tDelta} \tau).
    \label{eq:OBJips}
\end{align}
We use short-hand
    $\OBJ(\tau_2) = \OBJ(\Brac{1, \tau_2})$
and
    $\OBJ^{\ips}(\tau_2) = \OBJ^{\ips}(\Brac{1, \tau_2})$.
We think of these quantities, resp., as the LP-objective given the stopping time at $\tau_2$, and the IPS-estimate thereof. Recall that
\begin{align}
\OPTLPfull[\tau] &:= \max_{t \in [\tau]}\;\OBJ(t). \label{eq:HP-analysis-OPTLPfull}\\
\guess_\tau &:= \max_{t \in [\tau]}\;\OBJ^{\ips}(t).
\label{eq:HP-analysis-guess}
\end{align}

\xhdr{Uniform exploration does not exhaust budget.}
The uniform exploration in  Algorithm~\ref{alg:LagrangianBwKAdvHighP} happens for at most $\gamma_0\,T$ rounds in expectation, and therefore for at most
    $\gamma_0\,T + 3 \sqrt{\gamma_0 T \ln (1/\delta)}$ rounds with probability at least $1-\delta$.%
\footnote{By an easy application of Chernoff-Hoeffding bounds (Lemma~\ref{lem:Chernoff}).}
It does not consume more than $B/2$ units of each resource,
    since $\gamma_0 = T^{-1/4}$ and $B>4\, T^{3/4}$.
	
\xhdr{IPS estimators are good.}  We argue that, essentially, the guess $\guess_\tau$ is close to $\OPTLPfull[\tau]$ with high probability. To this end, we prove that $\OBJ(\tau)$ is close to its IPS estimator, for any given $\tau\in [T]$.






\begin{lemma}\label{lem:hatLPwithActual}
With probability at least $1-d\delta T$ it holds that
\begin{equation}\label{eq:tailhatLP}
\forall \tau \in [T]\qquad
\left| \OBJ^{\ips}(\tau) - \OBJ(\tau) \right|
    \leq \Dev(\tau)
        := \left( 1+ \frac{2\OBJ(\tau)}{B} \right)\frac{K}{\gamma} \sqrt{2\tau \log \tfrac{T}{\delta}},
\end{equation}

If the event \eqref{eq:tailhatLP} holds, then $\guess_{\tau}$ and
$\OPTLPfull[\tau]$ are indeed close:
\begin{equation}\label{eq:tailhatLP-guess}
\forall \tau \in [T]\qquad
\left| \guess_{\tau} - \OPTLPfull[\tau] \right|
\leq \DevM :=
	\frac{KT}{\gamma B} \sqrt{18 T \log \tfrac{T}{\delta}}
\end{equation}
\end{lemma}

The proof of this lemma is deferred Section~\ref{subsec:IPS}. It only relies on the uniform exploration property \eqref{eq:hp-guess} and the definition of IPS estimators, not on anything that the algorithm does. A somewhat subtle point is to derive concentration on the respective LP-values from concentration of the average outcome matrices.

\xhdr{IPS estimates do not change too fast.}
We use the phase-stopping condition in the algorithm to argue that algorithm's guesses $\guess_t$ and estimates $\OBJ^{\ips}(t)$ do not change too fast.


\begin{lemma}\label{lem:HP-stopping}
Consider a full phase in the execution of the algorithm. Let $\tau$ be the first round in this phase, let $\tau'$ be any other round in this phase, and let $\tau''$ be any round in the next phase. Then
\[ \guess_{\tau'}
        \leq   \kappa\cdot \guess_\tau< \guess_{\tau''}.\]
\end{lemma}

\begin{proof}
The first inequality holds because the phase-stopping condition did not fire at round $\tau'$. For the second inequality, let $t$ denote the first round in the next phase. Then
\begin{align*}
\guess_{\tau''}
&\geq \guess_t
    &\EqComment{since $(\guess_t)$ is monotone by \refeq{eq:HP-analysis-guess}}
    \\
&>  \kappa \cdot \guess_{\tau}
    &\EqComment{by the phase-stopping condition}. \qquad\qquad\qedhere
\end{align*}
\end{proof}


\begin{claim}\label{clm:aux-1}
$\guess_t = \OBJ^{\ips}(t) >\OBJ^{\ips}(t-1)$
for any round $t$ such that
    $\guess_t > \guess_{t-1}$.
The latter condition holds, in particular, if $t$ is the first round in some phase.
\end{claim}

\begin{proof}
The claim follows from \refeq{eq:HP-analysis-guess} and the phase-stopping condition. In particular, if $t$ is the first round in some phase and     $\guess_t = \guess_{t-1}$, then this phase would have started earlier.
\end{proof}

\begin{claim}\label{clm:singleStepJump}
For any round $t$, we have
$\guess_{t+1} - \guess_t
\leq \OBJ^{\ips}(t+1) - \OBJ^{\ips}(t)
    \leq K/\gamma$.
\end{claim}

\begin{proof}
Fix round $t$. If $\guess_{t+1} > \guess_t$, then
    $\guess_{t+1} = \OBJ^{\ips}(t+1)$
by Claim~\ref{clm:aux-1}. Since
    $\guess_t \geq \OBJ^{\ips}(t)$
by \refeq{eq:HP-analysis-guess}, it follows that 	
\begin{align}\label{eq:UPEq}	
\guess_{t+1} - \guess_t
    \leq  \OBJ^{\ips}(t+1) - \OBJ^{\ips}(t).
\end{align}

Let us analyze the right-hand side in \refeq{eq:UPEq}. Recall the IPS-estimate matrices defined in \refeq{eq:IPSEst} and \refeq{eq:IPSEst-ave}: the round-$t$ matrix $\vM_t^\ips$ and the time-average matrix $\bvM_t^\ips$. Let $\vec{X}^*$ denote the optimal solution to the LP induced by $\bvM_{t+1}^\ips$. This is the LP that determines $\OBJ^{\ips}(t+1)$, as per \refeq{eq:OBJips}. So,
    $\OBJ^{\ips}(t+1) = \sum_{\tau=1}^{t+1} \vec{X}^* \cdot  \vec{r}^{\ips}_\tau$,
where
     $\vec{r}^{\ips}_\tau$
is the reward vector in $\vM_{\tau}^\ips$.

Recall that the constraint in this LP is $\vec{X} \cdot \sum_{\tau=1}^{t+1} \vec{c}^{\ips}_{\tau} \leq \vec{B}$. Since $\vec{X}^*$ satisfies this constraint and $\vec{c}^{\ips}_{t+1} \geq \vec{0}$, we have  $\vec{X}^* \cdot \sum_{\tau=1}^{t} \vec{c}^{\ips}_{\tau} \leq \vec{B}$. So, $\vec{X}^*$ is also feasible to the LP induced by $\bvM_t^\ips$.
It follows that
    $\OBJ^{\ips}(t) \geq \sum_{\tau=1}^{t} \vec{X}^* \cdot \vec{r}^{\ips}_\tau$.

Putting this together, the right-hand side of \refeq{eq:UPEq} is at most
    $\vec{X}^* \cdot \vec{r}^{\ips}_{t+1}$.
This is at most $K/\gamma$, since the IPS-estimated reward of each arm is at most $1/\gamma$ by \refeq{eq:hp-explore}.
\end{proof}

%
%
%
%
%

\xhdr{Last full phase offers sufficient rewards.}
Recall that a phase in the execution of the algorithm is called \emph{full} if the stopping condition $\guess_t>\kappa \cdot \guessold$ has fired. We focus on the last full phase; let $\tau_{start}, \tau_{end}$ denote the first and last time-steps of this phase. We prove there is enough reward to be collected in this phase.

Let $\tau^*$ denote the maximizer in \refeq{eq:adv-stoppedLP} which we interpret as the \emph{optimal stopping time}. Essentially, we compare the LP value for the time interval $[\tau_{start}, \tau_{end}]$ with the LP value for the time interval $[1,\tau^*]$. The former is expressed as
    $\OBJ(\Brac{\tau_{start}, \tau_{end}})$
and the latter as $\OPTLPfull$. Note that the time horizon $T$ lies in the subsequent phase (so we can apply Lemma~\ref{lem:HP-stopping}).

\begin{lemma}\label{lem:remOPT}
Consider a run of the algorithm such that event \eqref{eq:tailhatLP} holds. Then
\begin{equation}\label{eq:fracRewardsHighP}	
\OBJ(\Brac{\tau_{start}, \tau_{end}}) \geq \left( \frac{1}{\kappa} - \frac{1}{\kappa^2} \right) \OPTLPfull - O(\DevM).
\end{equation}
\end{lemma}
The proof of this lemma is deferred Section~\ref{subsec:lastphase}.

\xhdr{Adversarial analysis of \MainALG.}
Let us plug in the adversarial analysis of \MainALG, as encapsulated in  Lemma~\ref{lm:adv-crux}. We focus on the last full phase in the execution. We interpret it as an execution of algorithm \MainALG with parameters $B_0,T_0$ on an instance of \AdvBwK with budget $B_0$ that starts at round $\tau_{start}$ of the original problem. Let $\guess = \guess_{\tau_{start}}$ be the guess at the first round of the phase. Then the parameters are $B_0 = B/\ratio$ and
    $T_0 = \guess/(\revedit{3}d\cdot \ratio)$,
where
    $\ratio = \Cel{\log_\kappa T}$.

We apply Lemma~\ref{lm:adv-crux} for round $\sigma=\tau_{end}-\tau_{start}+1$  in the execution of \MainALG. Restated in our notation, $f(\sigma)$ in Lemma~\ref{lm:adv-crux} becomes
\[ f(\sigma) = \OPTLP(\bvM_{[\tau_{start}, \tau_{end}]}, B_0, \sigma). \]
Thus, we obtain that with probability at least $1-\delta$ we have
\begin{equation} \label{eq:FromAdvGuarantees}
\REW \geq \sum_{t=\tau_{start}}^{\tau_{end}} r_t(a_t) \geq \min\Paren{\frac{\guess}{\revedit{3}d \lceil \log_{\kappa} T \rceil},\;
\sigma f(\sigma) - \frac{d \guess}{\revedit{3}d \lceil \log_{\kappa} T \rceil}} - \reg(T),
\end{equation}
where the regret term is
$\reg(T) := (1+\tfrac{T}{B})
    \left( R_{1,\,\delta/T}(T) + R_{2,\,\delta/T}(T) \right)$.

\xhdr{Rescaling the budget.} 
Since we use rescaled budget $B_0$, we need to connect the corresponding LP-values to those for the original budget $B$. We use the following general fact, observed in \citet{AgrawalDevanur-ec14}:
for any outcome matrix $\vM$, budget $B$, time horizon $T$, and rescaling factor $\psi \in (0, 1]$, 
\begin{align}\label{eq:scaledBudgetOPTVals1}
\OPTLP\Paren{\vM, \psi B, T}
    \geq \psi \cdot \OPTLP\Paren{\vM, B, T}.
\end{align}
This holds because an optimal solution $\vec{\mu}$ to $\LP_{\vM, B, T}$, the vector $\psi\,\vec{\mu} $ is feasible to $\LP_{\vM, \psi B, T}$.

\xhdr{Putting it all together.}
Let us show how to complete the proof of Theorem~\ref{thm:AdvBwK-HP} using the tools derived above.
Throughout, we condition on the high-probability events in Lemma~\ref{lem:hatLPwithActual} and \refeq{eq:FromAdvGuarantees}.

Recall that $\sbr{\tau_{start}, \tau_{end}}$ denotes the last full phase, and let $\guess$ denote the guess at the beginning of this phase. Recall that
     $\guess = \OBJ^{\ips}(\tau_{start})$.

From \refeq{eq:scaledBudgetOPTVals1} we have that $\sigma f(\sigma) \geq  \frac{1}{\ratio} \OBJ(\Brac{\tau_{start}, \tau_{end}})$ since $B_0 = \frac{B}{\ratio}$. Combining this with \refeq{eq:FromAdvGuarantees} we obtain
		\begin{equation}
			\label{eq:HighPSumCond}
			\REW = \sum_{t=\tau_{start}}^{\tau_{end}} r_t(a_t) \geq \frac{1}{\ratio} \min\Paren{\frac{\guess}{\revedit{3}d}, \OBJ(\Brac{\tau_{start}, \tau_{end}}) - \frac{d \guess}{\revedit{3}d}} - \reg(T).
		\end{equation}

By Lemma~\ref{lem:remOPT}, \asdelete{with $\kappa=2$,} we can re-write \refeq{eq:HighPSumCond} as		
		\begin{equation}
			\label{eq:partTwo}
		 \REW \geq \frac{1}{\ratio} \min\Paren{\frac{\guess}{\revedit{3}\,d}, \Paren{\frac{1}{\kappa}-\frac{1}{\kappa^2}} \OPTLPfull - \frac{\guess}{\revedit{3}}} - \reg(T) -\revedit{O(\DevM)}.
		\end{equation}


Let us characterize how the guess $\guess$ deviates from
$\OPTLPfull$:
\begin{align}\label{eq:guess-bounds} 	
\OPTLPfull\big/\kappa^2 -O(\DevM)
    \leq \guess \leq
\OPTLPfull\big/\kappa + O(\DevM).
\end{align}


To prove the upper bound in \refeq{eq:guess-bounds},
\begin{align*}
\guess
    &\leq \guess_{T}/\kappa
    &\EqComment{by Lemma~\ref{lem:HP-stopping}}
    \\
    &\leq \OPTLPfull\big/\kappa + \DevM/\kappa
    &\EqComment{by Lemma~\ref{lem:hatLPwithActual}}.
\end{align*}

For the lower bound in \refeq{eq:guess-bounds}, we observe that
\begin{align*}
\guess_T
&\leq \kappa \cdot \guess_{\tau_{end} + 1}
    &\EqComment{a phase didn't start at $T$}
    \\
&\leq \kappa \rbr{\guess_{\tau_{end}} + K/\gamma}
    &\EqComment{By Claim~\ref{clm:singleStepJump}}
    \\
&\leq \kappa \rbr{ \kappa \cdot \guess + K/\gamma}
    &\EqComment{a phase didn't start at $\tau_{end}$}
    \\
\guess
    &\geq \guess_T/\kappa^2 - \DevM/\kappa^2 - K/\gamma
    \\
    &\geq \OPTLPfull/\kappa^2 - O(\DevM)
    &\EqComment{by \refeq{eq:tailhatLP-guess}}.
\end{align*}
This completes the proof of \refeq{eq:guess-bounds}.

Plugging \refeq{eq:guess-bounds} back into \refeq{eq:partTwo} and using $\kappa = 2$ we get,
		
		\begin{equation}
			\label{eq:HighPFinal1}
			\textstyle \REW \geq \frac{1}{\ratio} \min\Paren{\frac{\OPTLPfull}{\revedit{12} d},  \frac{\OPTLPfull}{\revedit{12}}} - \reg(T) - O(\DevM).
		\end{equation}
		
Moreover, $\OPTLPfull \geq \OPTFD$ by \refeq{eq:adv-stoppedLP} . Plugging this into \refeq{eq:HighPFinal1} \asedit{and using $\gamma\geq T^{-1/4}$}, we obtain \refeq{eq:thm:AdvBwK-HP}, completing the proof of the theorem.

\subsection{Proof of Lemma~\ref{lem:remOPT}: last full phase offers sufficient rewards}
\label{subsec:lastphase}

First, we decompose the objective on a time interval as a difference between the interval's endpoints:
\begin{equation}\label{eq:diffOPTToFull}	
	\OBJ([T_1, T_2]) \geq \OBJ(T_2)-\OBJ(T_1-1).
\end{equation}
This step holds for any two rounds $T_1 < T_2\leq T$. It is proved similarly to Claim~\ref{clm:singleStepJump}.

		
\begin{proof}[Proof of \refeq{eq:diffOPTToFull}]
Let $\vec{\mu}$ denote the optimal solution to $\LP_{\bvM_{T_2}, B, T_2}$. In particular, it satisfies the consumption constraint
    $\vec{\mu} \cdot \sum_{t \in [T_2]} \vec{c}_{t, i} \leq B$.
Since resource consumption is always non-negative,
    $\vec{\mu} \cdot \sum_{t \in [T_1-1]} \vec{c}_{t, i} \leq B$,
which is the consumption constraint for $\LP_{\bvM_{T_1-1}, B, T_1-1}$. So, $\vec{\mu}$ is feasible for that LP as well. Consequently,
    \[ \textstyle \sum_{t \in [T_1-1]} \vec{\mu} \cdot \vec{r}_t \leq \OBJ(T_1-1).\]

Likewise $\vec{\mu}$ is also feasible for
    $\LP_{\bvM_{[T_1, T_2]}, B, [T_1, T_2]}$,
and consequently	
\begin{align*}
\OBJ([T_1, T_2])
    & \geq \sum_{t\in [T_1, T_2]} \vec{\mu} \cdot \vec{r}_t
	   = \sum_{t \in [T_2]} \vec{\mu} \cdot \vec{r}_t
        - \sum_{t \in [T_1-1]} \vec{\mu} \cdot \vec{r}_t
    \\
	& \geq \OBJ(T_2) - \OBJ(T_1-1).\qedhere
\end{align*}
\end{proof}	
		
The rest of the proof is specific to the time interval being the last full phase.
\begin{align*}	
\OBJ([\tau_{start}, \tau_{end}])
&\geq \OBJ(\tau_{end}) - \OBJ(\tau_{start}-1)
    &\EqComment{by \refeq{eq:diffOPTToFull}}
    \\
&\geq \OBJ^{\ips}(\tau_{end})
        - \OBJ^{\ips}(\tau_{start}-1) - 2\cdot \DevM
    &\EqComment{by Lemma~\ref{lem:hatLPwithActual}}
    \\
&\geq \OBJ^{\ips}(\tau_{end}+1)
        - \OBJ^{\ips}(\tau_{start}) - 2\cdot \DevM - K/\gamma.
\end{align*}
In the last inequality, we control
    $\OBJ^{\ips}(\tau_{start})$ and $\OBJ^{\ips}(\tau_{end})$
using, resp., Claim~\ref{clm:aux-1} and Claim~\ref{clm:singleStepJump}.

Let us transition from $\OBJ^{\ips}(\cdot)$ to guesses $\guess$, and use the machinery for comparing the guesses across time. For a more succinct notation, write $t = \tau_{end}+1$. Then
\begin{align*}	
\OBJ\rbr{[\tau_{start}, \tau_{end}]} +2\cdot \DevM + K/\gamma
&\geq \guess_t - \guess_{\tau_{start}}
    &\EqComment{by Claim~\ref{clm:aux-1}}
    \\
&> \guess_t\cdot (1-\nicefrac{1}{\kappa})
    &\EqComment{by Lemma~\ref{lem:HP-stopping}}
    \\
&\geq \guess_T \cdot (\nicefrac{1}{\kappa}-\nicefrac{1}{\kappa^2})
    &\EqComment{by Lemma~\ref{lem:HP-stopping}}
    \\
&\geq \rbr{\OPTLPfull[T] - \DevM}
    \cdot (\nicefrac{1}{\kappa}-\nicefrac{1}{\kappa^2})
    &\EqComment{by Lemma~\ref{lem:hatLPwithActual}}.
\end{align*}
Rearranging, we complete the proof of Lemma~\ref{lem:remOPT} as follows.
\begin{align*}
\OBJ\rbr{[\tau_{start}, \tau_{end}]}
&\geq \OPTLPfull[T] \cdot (\nicefrac{1}{\kappa}-\nicefrac{1}{\kappa^2})
    - \DevM\cdot (\nicefrac{1}{\kappa}-\nicefrac{1}{\kappa^2}+2) + K/\gamma.
\end{align*}

\subsection{Proof of Lemma~\ref{lem:hatLPwithActual} (IPS estimators are good)} \label{subsec:IPS}

	Recall that for every $t \in [T]$ and $a \in [K]$ we have that $p_t(a)$, the probability that arm $a$ is chosen at time $t$ is at least $\frac{\gamma}{K}$. We now prove Lemma~\ref{lem:highProbabilityDeviations} which relates linear sums of rewards and consumptions computed using the unbiased estimates and the true values. Denote $R_{\gamma, \delta}(\tau) := \frac{K}{\gamma} \sqrt{2 \tau \ln(T/\delta) }$.
	
	\begin{claim}
			\label{lem:highProbabilityDeviations}
			Let $\delta >0 $, $\gamma > 0$ used by the EXP3.P($\gamma$) be given parameters. Then we have the following statements for any fixed $\vec{z} \in \Delta_K$.
			\begin{align}
				&\textstyle \Pr \left[ \exists \tau \in [T] \quad \left| \sum_{t \in [\tau]} \vec{z} \cdot \Brac{\vec{r}^\ips_t - \vec{r}_t} \right| > R_{\gamma, \delta}(\tau) \right] \leq \delta \label{eq:rew}\\
				\forall i \in [d] \qquad &\textstyle  \Pr \left[ \exists \tau \in [T] \quad \left| \sum_{t \in [\tau]} \vec{z} \cdot \Brac{\vec{c}^\ips_{t, i} - \vec{c}_{t, i}(a)} \right| > R_{\gamma, \delta}(\tau)  \right] \leq \delta \label{eq:cons}
			\end{align}
		\end{claim}
		
			\begin{proof}
			The proof of this follows directly from the invocation of the Azuma-Hoeffding inequality. We will show this for Equation~\eqref{eq:rew}. Define $Y_t := \vec{z}\cdot \Brac{\vec{r}^\ips_t - \vec{r}_t}$ (like-wise for the lower-tail use $Y_t := \vec{z} \cdot (\vec{r}_t-\vec{r}^\ips_t)$). Note that this forms a martingale difference sequence since $\mathbb{E}[\vec{z} \cdot (\vec{r}^\ips_t - \vec{r}_t)~|~\mathcal{H}_{t-1}] = \vec{z} \cdot \Brac{\vec{r}_t-\vec{r}_t} = 0$. Here we used the fact that $\vec{z}$ is not random and fixed before the start of the algorithm. Also we have that $|Y_t| \leq \frac{K}{\gamma}$. Using Lemma~\ref{lem:AzumaHoeffding} and taking a union bound over all $\tau \in [T]$ we have the desired equation.
		\end{proof}

		We will now prove the two inequalities in \refeq{eq:tailhatLP}. We will first prove the first inequality in \refeq{eq:tailhatLP}.
		
		Let $\vec{\mu}^*$ denote the optimal solution to $\OPTLP\Paren{\bvM_{\tau}, B\Paren{1 - \frac{R_{\gamma, \delta}(\tau)}{B} }, \tau}$. Note this is valid whenever $B > \Omega\left( \frac{K}{\gamma} \sqrt{\tau \log \tfrac{T}{\delta}} \right)$. From Equation~\eqref{eq:cons} we have that with probability at least $1-\delta$ for every $ i \in [d]$,
		
		\begin{align}
			\textstyle \sum_{t \in [\tau]} \vec{\mu}^*\cdot \vec{c}^\ips_{t, i} &\textstyle \leq \sum_{t \in [\tau]} \vec{\mu}^* \cdot \vec{c}_{t, i} + R_{\gamma, \delta}(\tau). \nonumber \\
			&\textstyle \leq B \label{eq:HighPUTResource}
		\end{align}

		\refeq{eq:HighPUTResource} used the fact that $\sum_{t \in [\tau]} \vec{\mu}^* \cdot \vec{c}_{t, i} \leq B\Paren{1 -  R_{\gamma, \delta}(\tau)}$.

		Using Equation~\eqref{eq:rew}, we have that with probability at least $1-\delta$,
		\begin{equation*}
			\textstyle \sum_{t \in [\tau]} \vec{\mu}^* \cdot \vec{r}_t \leq \sum_{t \in [\tau]} \vec{\mu}^* \cdot \vec{r}^\ips_t + R_{\gamma, \delta}(\tau).
		\end{equation*}
		
		Using the fact that,
		\[
			\textstyle \sum_{t \in [\tau]} \vec{\mu}^* \cdot \vec{r}_t =  \OPTLP\Paren{\bvM_{\tau}, B\Paren{1 - \frac{ R_{\gamma, \delta}(\tau)}{B}}, \tau},
		\]
		we have the following.
		
		\begin{equation}
			\textstyle \OPTLP\Paren{\bvM_{\tau}, B\Paren{1 - \frac{R_{\gamma, \delta}(\tau)}{B}}, \tau} -  R_{\gamma, \delta}(\tau) \leq \sum_{t \in [\tau]} \vec{\mu}^* \cdot \vec{r}^\ips_t. \label{eq:HighPUTReward}
		\end{equation}

		From \refeq{eq:HighPUTResource} we have that $\vec{\mu}^*$ is feasible to $\OPTLP^{\ips}\Paren{\tau}$ and from \refeq{eq:HighPUTReward} this implies that
		
		\begin{equation}
			\label{eq:HighPWTLOPT}
			\textstyle \OBJ^{\ips}\Paren{\tau} \geq \OPTLP\Paren{\bvM_{\tau}, B\Paren{1 - \frac{ R_{\gamma, \delta}(\tau)}{B}}, \tau} -  R_{\gamma, \delta}(\tau).
		\end{equation}

		Finally from \refeq{eq:scaledBudgetOPTVals1} we have
		\begin{equation}
			\label{eq:HighPScaled1}
			\textstyle \OPTLP\Paren{\bvM_{\tau}, B\Paren{1 - \frac{ R_{\gamma, \delta}(\tau)}{B}}, \tau} \geq \Paren{1 - \frac{ R_{\gamma, \delta}(\tau)}{B}} \OBJ\Paren{\tau}.
		\end{equation}
		
		From \refeq{eq:HighPWTLOPT} and \refeq{eq:HighPScaled1} we have
		\[
				\textstyle \OBJ^{\ips}(\tau) \geq \OBJ(\tau) - \underbrace{R_{\gamma, \delta}(\tau)\Paren{1+ \frac{\OBJ(\tau)}{B}}}_{ \leq \left( 1+ \frac{2\OBJ(\tau)}{B} \right) \frac{K}{\gamma} \sqrt{2\tau \log \tfrac{T}{\delta}} },
		\]
		which gives the lower-tail in \refeq{eq:tailhatLP}.

	We will now prove the second inequality in \refeq{eq:tailhatLP} in a similar fashion. Let $\tilde{\vec{\mu}}^*$ denote the optimal solution to
	$\OBJ^{\ips}\Paren{\bvM_{\tau}, B\Paren{1 - \frac{R_{\gamma, \delta}(\tau)}{B} }, \tau}$.
		
	From Equation~\eqref{eq:cons} we have that with probability at least $1-\delta$ for every $ i \in [d]$,		
		\begin{align}
			\textstyle \sum_{t \in [\tau]} \tilde{\vec{\mu}}^* \cdot \vec{c}_{t, i} &\textstyle \leq \sum_{t \in [\tau]} \tilde{\vec{\mu}}^* \cdot \vec{c}^\ips_{t, i} + R_{\gamma, \delta}(\tau). \nonumber \\
			&\textstyle \leq B \label{eq:HighPLTResource}
		\end{align}

		\refeq{eq:HighPLTResource} used the fact that $\sum_{t \in [\tau]} \tilde{\vec{\mu}}^* \cdot \vec{c}^\ips_{t, i} \leq B\Paren{1 -  R_{\gamma, \delta}(\tau)}$.	
		
		Using Equation~\eqref{eq:rew}, we have that with probability at least $1-\delta$,
		\begin{equation*}
			\textstyle \sum_{t \in [\tau]} \tilde{\vec{\mu}}^* \cdot \vec{r}^\ips_t \leq \sum_{t \in [\tau]} \tilde{\vec{\mu}}^* \cdot \vec{r}_t + R_{\gamma, \delta}(\tau).
		\end{equation*}
		
		From the fact that
		\[
			\textstyle \sum_{t \in [\tau]} \tilde{\vec{\mu}}^* \cdot \vec{r}^\ips_t =  \OPTLP^{\ips}\Paren{\bvM^\ips_{\tau}, B\Paren{1 - \frac{ R_{\gamma, \delta}(\tau)}{B}}, \tau},
		\]
		we get the following.
		
		\begin{equation}
			\textstyle \OPTLP^{\ips}\Paren{\bvM^\ips_{\tau}, B\Paren{1 - \frac{R_{\gamma, \delta}(\tau)}{B}}, \tau} -  R_{\gamma, \delta}(\tau) \leq \sum_{t \in [\tau]} \tilde{\vec{\mu}}^* \cdot \vec{r}_t. \label{eq:HighPLTReward}
		\end{equation}

		From \refeq{eq:HighPLTResource} we have that $\tilde{\vec{\mu}}^*_{j}$ is feasible to $\OPTLP\Paren{\tau}$ and from \refeq{eq:HighPLTReward} this implies that
		
		\begin{equation}
			\label{eq:HighPWTUOPT}
			\textstyle \OPTLP^{\ips}\Paren{\bvM^\ips_{\tau}, B\Paren{1 - \frac{ R_{\gamma, \delta}(\tau)}{B}}, \tau} \leq \OBJ\Paren{\tau} +  R_{\gamma, \delta}(\tau).
		\end{equation}

		Finally from \refeq{eq:scaledBudgetOPTVals1} we have
		\begin{equation}
			\label{eq:HighPScaled2}
			\textstyle \OPTLP^{\ips}\Paren{\bvM^\ips_{\tau}, B\Paren{1 - \frac{ R_{\gamma, \delta}(\tau)}{B}}, \tau} \geq \Paren{1 - \frac{ R_{\gamma, \delta}(\tau)}{B}} \OBJ^{\ips}\Paren{\tau}.
		\end{equation}
		
		Combining \refeq{eq:HighPScaled2} and \refeq{eq:HighPWTUOPT} we get,
		
		\begin{equation}
			\label{eq:HighP2Scaled3}
			\OBJ^{\ips}(\tau) \leq \OBJ(\tau) + R_{\gamma, \delta}(\tau) + \frac{R_{\gamma, \delta}(\tau)}{B - R_{\gamma, \delta}(\tau)} \Paren{\OBJ(\tau) + R_{\gamma, \delta}(\tau)}.	
		\end{equation}
		
		Since $B > 2 R_{\gamma, \delta}(\tau)$ we get,
		
		\[
				\OBJ^{\ips}(\tau) \leq \OBJ(\tau) + \underbrace{ \frac{2R_{\gamma, \delta}(\tau)}{B} \Paren{\OBJ(\tau) + R_{\gamma, \delta}(\tau)}}_{\leq \left(1+\frac{2 \OBJ(\tau)}{B} \right) \frac{K}{\gamma} \sqrt{2 \tau \log \tfrac{T}{\delta} }},
		\]
		
		and thus we get the upper-tail in \refeq{eq:tailhatLP}.
		
		We will now prove \refeq{eq:tailhatLP-guess}. Recall that $\guess_{\tau} := \max_{t \in [\tau]} \OBJ^{\ips}(t)$. Moreover $\OPTLPfull[\tau] = \max_{t \in [\tau]} \OBJ(t)$.
		
		Consider $\guess_{\tau} - \OPTLPfull[\tau]$. We have,
		
		\begin{align*}
			\guess_{\tau} - \OPTLPfull[\tau] & = \max_{t \in [\tau]} \OBJ^{\ips}(t) - \OPTLPfull[\tau] \\
			& \leq \max_{t \in [\tau]} (\OBJ(t) + \Dev(t)) - \OPTLPfull[\tau] \\
			& \leq \max_{t \in [\tau]} \OBJ(t) + \max_{t \in [\tau]} \Dev(t) - \OPTLPfull[\tau]\\
			& = \max_{t \in [\tau]} \Dev(t).
		\end{align*}

		Now consider $\OPTLPfull[\tau]-\guess_{\tau}$. We have,
		
		\begin{align*}
			\OPTLPfull[\tau]-\guess_{\tau} &\leq 	\OPTLPfull[\tau] - \max_{t \in [\tau]} (\OBJ(t) - \Dev(t)) \\
			& \leq \OPTLPfull[\tau] - \max_{t \in [\tau]} \OBJ(t) + \max_{t \in [\tau]} \Dev(t) \\
			& = \max_{t \in [\tau]} \Dev(t).
		\end{align*}
		
	This completes the proof of Lemma~\ref{lem:hatLPwithActual}.

\section{Extensions}
\label{sec:ext}

\newcommand{\CorIID}{(C1)\xspace}
\newcommand{\CorAdv}{(C2)\xspace}

We obtain several extensions which highlight the modularity of \MainALG: we apply Theorem~\ref{thm:IID} and Theorem~\ref{thm:AdvBwK-main} with appropriately chosen primal algorithm $\ALG_1$, and immediately obtain strong performance guarantees.%
\footnote{For these theorems to hold, $\ALG_1$ needs to satisfy regret bound \eqref{eq:prelims-regret-weak} only against adaptive adversaries that arise in the repeated Lagrange game in the corresponding extension, not against \emph{arbitrary} adaptive adversaries.}
We tackle four well-known scenarios:
\begin{itemize}
\item \emph{full feedback} \citep[\eg][]{LittWarm94,FS97,AroraHK}: in each round, the algorithm chooses an action and observes the outcomes of all possible actions; this is a classic scenario in online machine learning.

\item \emph{combinatorial semi-bandits} \cite[\eg][]{Gyorgy-jmlr07,Kale-nips10,Audibert-colt11}: actions are feasible subsets of ``atoms". The atoms in the chosen action have individual outcomes that are observed and add up to the action's total outcome. Typical motivating example are subsets/lists of news articles, ads, or web search results.

\item \emph{contextual bandits with policy sets} \cite[\eg][]{Langford-nips07,policy_elim,monster-icml14}: before each round, a \emph{context} is observed, and the algorithm competes against the best policy (mapping from context to actions) in a given policy class. In a typical application scenario, the context includes known features of the current user.
\item \emph{bandit convex optimization} (starting from \citet{Bobby-nips04,FlaxmanKM-soda05}, with recent advances \citet{Bubeck-colt15,Hazan-nips14,bubeck2017kernel}). Here the set of actions is a convex set $\mX \subset \R^K$. For each round $t$, the adversary chooses a concave function
        $f_t: \mathcal{X}\to [0,1]$
    such that the reward for chosen action $\vec{x}\in \mX$ is $f_t(\vec{x})$.
\end{itemize}

\xhdr{Formalities.} To simplify the statements, we make the following assumptions without further mention:
\begin{itemize}
\item The dual algorithm, $\ALG_2$, is always Hedge, with the associated regret bound from \refeq{eq:prelims-regrets}. For high-probability regret bounds, $\delta = \tfrac{1}{T}$ is a fixed and known  failure probability parameter.
\item For \StochasticBwK, one resource is the dummy resource (with consumption $\tfrac{B}{T}$ for each arm). Algorithm \MainALG is run with parameters $B_0 = B$ and $T_0=T$.

\item For \AdvBwK, one of the arms is a \emph{null arm} that has zero reward and zero resource consumption. Algorithm~\ref{alg:LagrangianBwKAdv} is run with any $\kappa> 1$ and range $[\Gmin,\Gmax] = [\sqrt{T},T]$, as in Theorem~\ref{thm:AdvBwK-main}(b).
\end{itemize}

\xhdr{A typical corollary.} All our corollaries have the following shape, for some regret term $\reg$:
\begin{itemize}
\item[\CorIID] In the stochastic version, algorithm \MainALG achieves, with probability at least $1-\tfrac{1}{T}$,
\begin{align*}
\OPTDP - \REW \leq
    O\left( \tfrac{T}{B}\cdot \reg\right).
\end{align*}

\item[\CorAdv] In the adversarial version, Algorithm~\ref{alg:LagrangianBwKAdv} achieves
\begin{align*}
\frac{\E[\REW]}{\OPTFD}
    \geq  \frac{2}{(d+1)\;  \ln(T)} - O(\reg)
            \left( \tfrac{1}{\OPTFD} + \tfrac{1}{d B}\right).
\end{align*}
\end{itemize}

\OMIT{For the stochastic case, it is unsurprising given prior work, in all three extensions, that the regret bound in \CorIID is attained by some algorithm. The significance here is that we obtain this result as an immediate corollary.}

\noindent Corollaries similar to \CorAdv can be achieved for Algorithm~\ref{alg:LagrangianBwKAdvHighP}, too; we omit them for ease of exposition.

\subsection{BwK with full feedback}

In the full-feedback version of BwK, the entire outcome matrix $\vM_t$ is revealed to the algorithm after each round $t$. Accordingly, we can use Hedge as the primal algorithm $\ALG_1$. The effect is, essentially, that the dependence on $K$, the number of arms, in the regret term becomes logarithmic rather than $\sqrt{K}$.

\begin{corollary}\label{cor:fullFB}
Consider BwK with full feedback. Using Hedge as the primal algorithm, we obtain corollaries \CorIID and \CorAdv with regret term
    $\reg = \sqrt{T \ln \left( dKT\right) }$.
\end{corollary}

%

\AdvBwK with full feedback have not been studied before. For the stochastic version, the regret bound is unsurprising: one expects to obtain a similar improvement with each of the three other algorithms in the prior work on \StochasticBwK by tracing the ``confidence terms" through the analysis. The significance here is that we obtain this result as an immediate corollary.

\subsection{Combinatorial Semi-Bandits with Knapsacks}
\label{sec:ext-semi}

Following \cite{Karthik-aistats18}, we consider \emph{Combinatorial Semi-BwK}, a common generalization of BwK and \emph{combinatorial semi-bandits} \cite[\eg][]{Gyorgy-jmlr07,Kale-nips10,Audibert-colt11}. In this problem, actions correspond to subsets of some finite ground set $\Omega$ of size $n$, whose elements are called \emph{atoms}. There is a fixed family $\mF\subset 2^{\Omega}$ of feasible actions. For each round $t$, there is an outcome vector $\vec{o}_{t,e} \in [0,\tfrac{1}{n}]^{d+1}$ for each round atom $e\in \Omega$, with the same semantics as the actions' outcome vectors. If an action $S\subset \Omega$ is chosen, the outcome vectors $\vec{o}_{t,e}$ are observed for all atoms $e\in S$, and the action's outcome is the sum
    $\vM_t(S) = \sum_{e\in S} \vec{o}_{t,e} \in [0,1]^d$.
In the adversarial case, all outcome vectors $\vec{o}_{t,e}$, $t\in [T]$, $e\in \Omega$ are chosen by an adversary arbitrarily before round $1$. In the stochastic case, the atomic outcome matrix
    $(\vec{o}_{t,e}:\; e\in \Omega)$
is drawn independently in each round $t$ from some fixed distribution. Combinatorial semi-bandits, as studied previously, is a special case with no resource constraints ($d=0$).

Typical motivating examples involve ad/content personalization scenarios. Atoms can correspond to items news articles, ads, or web search results, and actions are subsets that satisfy some constraints on quantity, relevance, or diversity of items. One can also model ranked lists of atoms: then atoms are rank-item pairs, and each feasible action $S\subset \Omega$ satisfies a constraint that each rank between $1$ and $|S|$ is present in exactly one chosen rank-item pair.

A naive application of our main results suffers from regret terms that are proportional to $\sqrt{|\mF|}$, which may be exponential in the number of atoms $n$. Instead, the work on combinatorial semi-bandits features regret bounds that scale polynomially in $n$. This is what we achieve, too. We use an algorithm from \cite{neu2016importance} which solves combinatorial semi-bandits in the absence of resource constraints. This algorithm satisfies a high-probability regret bound \eqref{eq:prelims-regret-weak} against an adaptive adversary, with
    $R_\delta(T) = O(\sqrt{nT \log(1/\delta)})$.~
\footnote{\label{fn:ext-semi}Prior work \citep{neu2016importance,Karthik-aistats18} posits that atoms' per-round rewards/consumptions lie in the range $[0,1]$, rather than $[0,\tfrac{1}{n}]$, so their stated regret bounds should be recaled accordingly.}

\begin{corollary}\label{cor:SemiBwK}
Consider Combinatorial Semi-BwK with $n$ atoms. Using the algorithm from \cite{neu2016importance} as the primal algorithm, we obtain corollaries \CorIID and \CorAdv with regret term
    $\reg = \sqrt{nT \log T}$.
\end{corollary}

The adversarial version of Combinatorial Semi-BwK has not been studied before.

The stochastic version has been studied in \cite{Karthik-aistats18} when the action set is a matroid, achieving regret
    \[ \tilde{O}\left( \OPTDP\,\sqrt{n/B} + \sqrt{T/n}+ \sqrt{\OPTDP} \right). \]
    This regret bound becomes
    $ \tilde{O}(\sqrt{nT})$
in the regime when $B$ and $\OPTDP$ are $\Omega(T)$ (see Footnote~\ref{fn:ext-semi}). We achieve the same regret bound for this regime, without the matroid assumption and without any extra work. However, the regret bound in \cite{Karthik-aistats18} can be substantially better than ours when $\OPTDP \ll T$.

\subsection{Contextual Bandits with Knapsacks}
\label{sec:ext-CB}

Following \cite{cBwK-colt14,CBwK-colt16}, we consider \emph{Contextual Bandits with Knapsacks} (\emph{cBwK}), a common generalization of BwK and \emph{contextual bandits with policy sets} \cite[\eg][]{Langford-nips07,policy_elim,monster-icml14}.
The only change in the protocol, compared to BwK, is that in the beginning of each round $t$ a context $x_t\in \mX$ arrives and is observed by the algorithm before it chooses an action. The context set $\mX$ is arbitrary and known. In the adversarial version (\emph{Adversarial cBwK}) both $x_t$ and the outcome matrix $\vM_t$ is chosen by an adversary. In the stochastic version (\emph{Stochastic cBwK}) the pair $(x_t,\vM_t)$ is chosen independently from some fixed but unknown distribution over such pairs.

In cBwK one is also given a finite set $\Pi$ of \emph{policies}: deterministic mappings from contexts to actions.%
\footnote{W.l.o.g. assume that $\Pi$ contains all \emph{constant policies}, \ie all policies that always evaluate to the same action.}
 Essentially, the algorithm competes with the best course of action restricted to these policies. For a formal definition, let us interpret cBwK as a BwK problem with action set $\Pi$, denote this problem as $\term{BwK}(\Pi)$. In other words, actions in $\term{BwK}(\Pi)$ are policies in cBwK. An algorithm for $\term{BwK}(\Pi)$ is oblivious to context arrivals. It chooses a policy $\pi_t\in \Pi$ in each round $t$, and receives an outcome for this policy: namely, the outcome for action $\pi(x_t)$. We are interested in the usual benchmarks for this problem, the best algorithm $\OPTDP$ and the best fixed distribution $\OPTFD$ (where both benchmarks are constrained to use policies in $\Pi$); denote them $\OPTDP(\Pi)$ and $\OPTFD(\Pi)$, respectively.

Without budget constraints (\ie with $B=T$), this is exactly contextual bandits with policy set $\Pi$. Both benchmarks then reduce to the standard benchmark of the best fixed policy.

\xhdr{Background: algorithm EXP4.P.} We use EXP4.P \citep{exp4p}, an algorithm for the contextual version of adversarial online learning with bandit feedback. The algorithm operates according to the protocol in Figure~\ref{prob:CB-adv}.

\begin{figure}[ht]
\begin{framed}
{\bf Given:} action set $[K]$, context set $\mX$, policy set $\Pi$, payoff range $[\rmin,\rmax]$.\\
In each round $t\in [T]$,
\begin{OneLiners}
\item[1.] the adversary chooses a context $x_t\in \mX$ and a payoff vector $\vec{g}_t\in [\rmin,\rmax]^K$;
\item[2.] the algorithm chooses a distribution $\vec{p}_t$ over $\Pi$ (without seeing $x_t$);
\item[3.] a policy $\pi_t\in \Pi$ is drawn independently from $\vec{p}_t$;
\item[4.] algorithm's chosen action is defined as $a_t = \pi_t(x_t) \in [K]$;
\item[5.] payoff $g_t(a_t)$ is received by the algorithm.
\end{OneLiners}
\vspace{-1mm}
\end{framed}
\caption{Adversarial contextual bandits}
\label{prob:CB-adv}
\end{figure}

We are interested in regret bounds for EXP4.P relative to the best fixed policy:
\[ \textstyle \OPT_\Pi = \max_{\pi\in \Pi} \sum_{t\in [T]} f_t(\pi(x_t)). \]
For each round $t$, the pair $(x_t,\vec{g}_t)$ induces a payoff vector $\vec{f}_t \in [\rmin,\rmax]^\Pi$ on policies:
\[ f_t(\pi) = g_t(\pi(x_t)) \qquad \forall \pi \in \Pi.  \]

\begin{theorem}[\cite{exp4p}]\label{thm:EXP4}
Fix failure probability $\delta>0$, policy set $\Pi$, and payoff range $[\rmin,\rmax]$.
Then algorithm EXP4.P (appropriately tuned) satisfies the following regret bound:
\begin{align}\label{eq:thm:EXP4}
\textstyle
\Pr\left[\;
  \OPT_\Pi
   \;-\;  \sum_{t\in [T]} f_t(\pi_t)  \leq (\rmax-\rmin)\, R_\delta(T)
\;\right] \geq 1-\delta,
 \end{align}
with regret term
    $ R_\delta(T) = O\left( \sqrt{\tau K \log (KT\, |\Pi|/\delta )} \right)$.
\end{theorem}

\xhdr{Our solution for cBwK.}
We solve cBwK by reducing it to $\term{BwK}(\Pi)$, and treating it as a BwK problem. A naive solution simply posits $|\Pi|$ arms and directly applies the machinery developed earlier in this paper. This results in $\sqrt{|\Pi|}$ dependence in regret bounds, which is unsatisfactory, as the policy set may be very large. Instead, we use EXP4.P as the primal algorithm ($\ALG_1$). We interpret EXP4.P as an algorithm for (non-contextual) adversarial online learning, as defined in Section~\ref{sec:prelims}, with action set $\Pi$. It is easy to see that Theorem~\ref{thm:EXP4} provides regret bound
\eqref{eq:prelims-regret-weak}
under this interpretation. Therefore, we obtain the following:

\begin{corollary}\label{cor:CB}
Consider contextual bandits with knapsacks, with policy set $\Pi$. Using EXP4.P as the primal algorithm,
we obtain corollaries \CorIID and \CorAdv with regret term
    $\reg = \sqrt{T K\ln \left( dKT\,|\Pi| \right)}$.
The benchmarks are
    $\OPTDP = \OPTDP(\Pi)$ and $\OPTFD = \OPTFD(\Pi)$.
\end{corollary}

%

Adversarial cBwK has not been studied before. The regret bound for the adversarial case is meaningful only if $B>\sqrt{T}$. This is essentially inevitable in light of the lower bound in Theorem~\ref{thm:LB-CB}.

Stochastic cBwK has been studied in \cite{cBwK-colt14,CBwK-colt16}, achieving regret
\begin{align}\label{eq:ext-CB-regret-prior}
O(\reg)(1+\OPTDP(\Pi)/B),
\end{align}
where the $\reg$ term is the same as in Corollary~\ref{cor:CB}. Whereas the regret bound from Corollary~\ref{cor:CB} is
    $O(\reg\cdot T/B)$.
Note that we match \eqref{eq:ext-CB-regret-prior} in the regime
    $\OPTDP(\Pi) > \Omega(T)$.
Our regret bound is optimal, up to logarithmic factors, in the regime $B>\Omega(T)$. This is due to the
    $\min\left( T,\,\Omega(\sqrt{KT \log(|\Pi|)/\log(K)} \right)$
lower bound on regret, which holds for contextual bandits \citep{RegressorElim-aistats12}.

\xhdr{Discussion.}
Our algorithms are slow, as the per-round running time of EXP4.P is proportional to $|\Pi|$. The state-of-art approach to computational efficiency in contextual bandits is \emph{oracle-efficient algorithms}, which make only a small number of calls to an oracle that finds the best policy in $\Pi$ for a given data set. In particular, prior work for Stochastic cBwK \citep{CBwK-colt16} obtains an oracle-efficient algorithm with regret bound as in \eqref{eq:ext-CB-regret-prior}. To obtain oracle-efficient algorithms for cBwK in our framework, both for the stochastic and adversarial versions, it suffices to replace EXP4.P with an oracle-efficient algorithm for adversarial contextual bandits that obtains regret bound \eqref{eq:thm:EXP4}, possibly with a larger regret term $R_\delta$. Such algorithms \emph{almost} exist: a recent breakthrough \citep{Syrgkanis-AdvCB-icml16,Rakhlin-AdvCB-icml16,Syrgkanis-AdvCB-nips16} obtains algorithms with similar regret bounds, but only for expected regret.

\subsection{Bandit Convex Optimization with Knapsacks}
\label{sec:ext-BCO}

We consider \emph{Bandit Convex Optimization with Knapsacks} (\emph{BCOwK}), a common generalization of BwK and \emph{bandit convex optimization}. We define BCOwK as a version of BwK, where the action set $\mX$ is a convex subset of $\mathbb{R}^K$. For each round $t$, there is a concave function
        $f_t: \mathcal{X}\to [0,1]$
and convex functions
    $g_{t, i}: \mathcal{X} \rightarrow [0, 1]$, for each resource $i$, so that
the reward for choosing action $\vec{x}\in \mX$ in this round is $f_t(\vec{x})$ and consumption of each resource $i$ is $g_{t,i}(\vec{x})$.
In the stochastic version, the tuple of functions
    $(f_t; g_{t, 1} \LDOTS g_{t,d})$
is sampled independently in each round $t$ from some fixed distribution (which is not known to the algorithm). In the adversarial version, all these tuples are chosen by an adversary before the first round.

Neither stochastic nor adversarial version of BCOwK have been studied in prior work (but see the discussion of constrained online convex optimization in Section~\ref{sec:related}). Bandit convex optimization, as studied previously, is a special case with no resource constraints ($d=0$).

\OMIT{the algorithm chooses a point $\vec{x}_t \in \mX$. For a chosen action $\vec{x}_t$, the algorithm obtains reward $f_t(\vec{x}_t)$ and consumes $g_{t, 1}(\vec{x}_t), g_{t, 2}(\vec{x}_t), \ldots, g_{t, d}(\vec{x}_t)$ amount of each resource. We assume that the function $f_t: \mathcal{X} \rightarrow [0, 1]$ is concave and functions $g_{t, i}: \mathcal{X} \rightarrow [0, 1]$ are convex.\footnote{Since linear functions are both convex and concave, BwK is a special case of BCOwK.} The game stops at time $\tau$ if for some resource $j$, we have $\sum_{t \in [\tau]} g_{t, j}(x_t) > B$.

from a latent distribution with means $f, \{g_i\}_{i \in [d]}$ respectively. In the adversarial setting, the sequence of functions are arbitrary. In the stochastic version, we assume that there exists a point $x \in \mX$ such that $f(x) =0$ and $g_i(x) = 0$ for every $i \in [d]$. In the adversarial setting, we assume that there exists a point $x \in \mX$ such that for every $t \in [T]$, $f_t(x)=0$ and $g_{t, i}(x) = 0$ for every $i \in [d]$.

These requirements are stronger than the \emph{Slater's} condition required for constrained convex optimization.\footnote{The requirements implies that the point $x \in \mX$ satisfies the Slater's constraint. The converse is not necessarily true.} The algorithm competes with the usual benchmarks given the sequence of functions; the reward obtained by the best sequence of actions in $\mX$ ($\OPTDP$) and the best fixed distribution over points in $\mX$ ($\OPTFD$).
} 

The primal algorithm $\ALG_1$ in \MainALG faces an instance of BCO (with an adaptive adversary). This is because the Lagrange function~\eqref{eq:LagrangianGeneral} is a concave function of the action, as sum of concave functions. For our primal algorithm, we use a recent breakthrough on BCO due to \cite{bubeck2017kernel}. This algorithm satisfies the high-probability regret bound \eqref{eq:prelims-regret-weak} against an adaptive adversary, with regret term
\[ R_\delta(T) = O(K^{9.5} \log^{7}(T) \sqrt{T \log(1/\delta)}).\]

We assume the existence of a \emph{null arm}: a point $\vec{x}\in \mX$ such that
    $f_t(\vec{x})=g_{t,i}(\vec{x}) = 0$
for each resource $i$ except the ``dummy resource". (Recall that we posit the ``dummy resource" --  a resource whose consumption is $B/T$ for each arm -- for the stochastic version.) Unlike elsewhere in this paper, this assumption is not without loss of generality: indeed, the null arm should be ``embedded" into $\mX$ without breaking the convexity/concavity properties. Moreover, we assume that the null arm lies in the interior of $\mX$.

\begin{corollary}\label{cor:BCO}
	Consider BCOwK for a given convex set $\mX \subset \mathbb{R}^K$. Using the algorithm from \cite{bubeck2017kernel} as the primal algorithm, we obtain corollaries \CorIID and \CorAdv with regret term $\reg = K^{9.5} \log^{7.5}(T) \sqrt{T}$.
\end{corollary}

\begin{remark}\label{rem:ext-BCO-infinite}
\MainALG framework extends to infinite action sets: everything carries over, as long as \refeq{eq:LagrangeMinMax} holds. (Essentially, we never take union bounds over actions, and we can replace $\max$ and sums over actions with $\sup$ and integrals.) For BCOwK, \refeq{eq:LagrangeMinMax} is a statement about constrained convex optimization programs. According to \emph{Slater's condition} \citep[see Eq. (5.27) in][]{boyd2004convex}, it suffices to have a point $\vec{x}$ in the interior of $\mX$ such that $g_{t,i}(\vec{x}) < B/T$ for each resource $i \in [d]$ other than the dummy resource (or any other resource whose consumption is the same in all rounds).
One such point is the null arm.
\end{remark}


%
%

\section{Lower bounds}
\label{sec:LB}

We prove the lower bounds on the competitive ratio that we have claimed in Section~\ref{sec:intro}: the $\Omega(\log T)$ lower bound w.r.t. the best fixed distribution benchmark ($\OPTFD$), the $\Omega(T)$ lower bound w.r.t. the best dynamic policy benchmark ($\OPTDP$), and the $\Omega(K)$ lower bound w.r.t. the best fixed arm benchmark ($\OPTFA$). As a warm-up, we analyze the simple example from Section~\ref{sec:intro} that leads to the $\tfrac54$ lower bound w.r.t. $\OPTFD$. All lower-bounds are for a randomized algorithm against an oblivious adversary. We summarize all these lower bounds in the following theorem:


\begin{theorem}\label{thm:LB-main}
Consider \AdvBwK with a single resource ($d=1$), $K$ arms, budget $B$, and time horizon $T$. Consider any randomized algorithm for this problem, and let $\REW$ denote its reward. Then:
\begin{itemize}
\item[(a)]
    $\OPTFD/\E[\REW]\geq \tfrac54 - o(1)$ for some problem instance (from the example in the Introduction).\\
This holds even if $\OPTFD\geq T/4$ and $B = T/2$.

\item[(b)] $\OPTFD/\E[\REW]\geq \Omega(\log T)$ for some problem instance with $K=2$ arms.\\
This holds for any given budget
    $B\in [c_0\,\log^3(T), O(T^{1-\alpha})]$,
 even if $\OPTFD\geq B^2/T$. \\
 Here $\alpha\in (0,1)$ is an arbitrary absolute constant, and $c_0$ is any large enough absolute constant.

\item[(c)] $\OPTDP/\E[\REW]\geq T/B^2$ for some problem instance with $K=2$ arms. \\
This holds for any given budget $B<\sqrt{T}$, with $\OPTFD=B$.

\item[(d)] $\OPTFA/\E[\REW]\geq \Omega(K)$ for some problem instance.\\
This holds for any given budget $B$, with $\OPTFA > B/K^K$.
\end{itemize}
\end{theorem}

\begin{remark}
The lower bounds for parts (a,b,c) hold (even) for problem instances with $K=2$ arms, one of which is the ``null arm" with no rewards and no resource consumption. The lower bounds in parts (a,b) hold even for a much more permissive feedback model from the online packing literature, namely, when the algorithm observes the outcome vector for all actions in a given round, and moreover does it \emph{before} it chooses an arm in this round.
\end{remark}

We tweak our construction from Theorem~\ref{thm:LB-main}(c) to obtain a strong lower bound for the contextual version of \AdvBwK (a.k.a. \emph{Adversarial cBwK}), as studied in Section~\ref{sec:ext-CB}. This lower bound implies that Adversarial cBwK is essentially hopeless in the regime $B<\sqrt{T}$, complementing a strong positive result (Corollary~\ref{cor:CB}) for the regime $B>\tilde{\Omega}(\sqrt{T})$. It is proved in Section~\ref{sec:LB-DP}, along with Theorem~\ref{thm:LB-main}(c).

\begin{theorem}\label{thm:LB-CB}
Consider adversarial contextual bandits with knapsacks (Adversarial cBwK), with policy class $\Pi$, a single resource ($d=1$), $K=2$ arms, and any given budget $B<\sqrt{T}$. Consider any randomized algorithm for this problem, and let $\REW$ denote its reward. Then
 \[ \OPTFD(\Pi)/\E[\REW]\geq T/B^2\quad  \text{for some problem instance}. \]
\end{theorem}

\noindent

\xhdr{Notation.} In the proof of lower-bounds below, we use the following notation. Given an instance $\mI$, we denote $\OPTFD(\mI)$, $\OPTFA(\mI)$ and $\OPTDP(\mI)$ to denote the optimal value of the best fixed distribution, best fixed arm and best dynamic policy respectively, for instance $\mI$. Likewise let $\OPTLPfull(\mI)$ denote the optimal $\LP$ value for instance $\mI$ and given an algorithm $\mA$ and an instance $\mI$, let $\E[\REW(\mA, \mI)]$ denote the expected reward obtained by $\mA$ on instance $\mI$, where the expectation is over the internal randomness of the algorithm.

\subsection{Warm-up: example from the Introduction}
\label{sec:LB-warmup}

As a warm-up, let us recap and analyze the example from the Introduction.

\begin{construction}\label{con:LB-simple}
There are two arms and one resource with budget $B=\tfrac{T}{2}$. Arm $1$ has zero rewards and zero consumption. Arm $2$  has consumption $1$ in each round, and offers reward $\tfrac12$ in each round of the first half-time ($\tfrac{T}{2}$  rounds). In the second half-time, arm $2$ offers either reward $1$ in all rounds, or reward $0$ in all rounds. More formally, there are two problem instances, call them $\mI_1$ and $\mI_2$, that coincide for the first half-time and differ in the second half-time.
\end{construction}

\begin{lemma}\label{thm:simpleLowerBound}
Any algorithm suffers
$\OPTFD/\E[\REW]\geq \tfrac54 - o(1)$
on some instance in  Construction~\ref{con:LB-simple}.
\end{lemma}

The intuition is that given a random instance as input the algorithm needs to choose how much budget to invest in the first half-time, without knowing what comes in the second, and any choice (in expectation) leads to the claimed competitive ratio.

To prove Lemma~\ref{thm:simpleLowerBound} (as well we the main lower bound in Theorem~\ref{thm:LB-main}(b)) we compare algorithm's performance to $\OPTLPfull$, and invoke the following lemma:

\begin{lemma}\label{lem:optfull-optfd}
$\OPTFD \geq \OPTLPfull - O\Paren{\OPTLPfull \cdot \sqrt{\frac{\log dT}{B}}}.$
\end{lemma}

\begin{proof}
Let $\tau^*$ denote the time-step at which $\OPTLPfull$ is maximized. Let $\vec{p}$ denote the optimal solution to $\tau^* \cdot \OPTLP(\bvM_{\tau^*}, B(1-\epsilon), \tau^*)$ where $\epsilon=\sqrt{\frac{\log dT}{B}}$. Note that $\OPTFD$ is at least as large as the expected total reward obtained by the distribution $\vec{p}$. From the Chernoff-Hoeffding bounds (Lemma~\ref{lem:Chernoff}), with probability at least $1-dT^{-2}$ we have	
		\[
				\textstyle \forall i \in [d] \qquad \sum_{t \in [\tau^*]} \vec{p} \cdot \vec{c}_{t, i} \leq B.
		\]
Conditioning on this event the expected total reward obtained by $\vec{p}$ is
\[ \textstyle \sum_{t \in [\tau^*]} \;\vec{p} \cdot \vec{r}_t = \tau^* \cdot \OPTLP(\bvM_{\tau^*}, B(1-\epsilon), \tau^*).\] Thus the expected total reward obtained by $\vec{p}$ is at least $\tau^* \cdot \OPTLP(\bvM_{\tau^*}, B(1-\epsilon), \tau^*)$.~ \footnote{With probability $T^{-2}$ we assume that $\vec{p}$ has an expected reward of $0$.} Moreover from \refeq{eq:scaledBudgetOPTVals1} we have that
\begin{align*}
\OPTFD
    &\geq  \tau^* \cdot \OPTLP(\bvM_{\tau^*}, B(1-\epsilon), \tau^*) \\
    &\geq (1-\epsilon) \tau^* \cdot \OPTLP(\bvM_{\tau^*}, B, \tau^*) \\
    &\geq \textstyle \OPTLPfull - O\Paren{\OPTLPfull \sqrt{\frac{\log dT}{B}}}. \qedhere
\end{align*}
\end{proof}

\begin{proof}[Proof of Lemma~\ref{thm:simpleLowerBound}]
Denote the two arms by $A_1$ and $A_0$ where $A_0$ denotes the null arm. The consumption for arm $A_1$ is $1$ for all rounds in both $\mI_1$ and $\mI_2$. Thus the only difference between the two instances is the rewards obtained for playing arm $A_1$ in each round. The instances have two phases where each phase lasts for $\tfrac{T}{2}$ rounds. In phase $1$, in both $\mI_1$ and $\mI_2$ playing arm $A_1$ fetches a reward $\tfrac{1}{2}$. In the second phase, in $\mI_1$, the reward for playing arm $A_1$ is $0$ while in $\mI_2$ the reward for playing arm $A_1$ is $1$. Thus the \emph{outcome} matrix for the first $\tfrac{T}{2}$ time-steps is the same in instances $\mI_1$ and $\mI_2$.
			
			Consider a randomized algorithm $\mA$. Let $\alpha_1$ be the expected number of times arm $A_1$ is played by $\mA$ in the first $\tfrac{T}{2}$ rounds on instances $\mI_1$ and $\mI_2$. Note since the outcome matrix is same, the expected number of times the arm is played should be same in both the instances. Let $\alpha_{2, 1}$, $\alpha_{2, 2}$ denote the expected number of times arm $A_1$ is played in the second phase in instances $\mI_1$ and $\mI_2$ respectively.

			Recall that in this section we are interested in a lower-bound on the competitive ratio $\OPTFD/\E[\REW]$ for every instance.
			Consider $\OPTLPfull(\mI_1)$, the optimal value of the best fixed distribution on $\mI_1$.  Using \refeq{eq:adv-stoppedLP} with $\tau=\tfrac{T}{2}$ this equals $\frac{T}{2} \cdot \OPTLP\Paren{\bvM_{\tfrac{T}{2}}, B, \tfrac{T}{2}}$ which evaluates to $\tfrac{T}{4}$.
			Likewise $\OPTLPfull(\mI_2)$ equals $T \cdot \OPTLP\Paren{\bvM_{T}, B, T}$, which evaluates to $\tfrac{3T}{8}$.
			Consider the performance of $\mA$ on $\mI_1$. We have,
						\begin{equation}
				\label{eq:simple-lower-cr1}
				\textstyle \tfrac{\OPTLPfull(\mI_1)}{\E[\REW(\mA, \mI_1)]} \geq \Paren{\tfrac{T}{4}}/\Paren{ \tfrac{\alpha_1}{2} }.
			\end{equation}
			
			Likewise on $\mI_2$ we have,
			\begin{equation}
				\label{eq:simple-lower-cr2}
				\textstyle \tfrac{\OPTLPfull(\mI_2)}{\E[\REW(\mA, \mI_2)]}  \geq \Paren{ \tfrac{3T}{8} }/\Paren{ \tfrac{\alpha_1}{2} + \alpha_{2, 2} }.	
			\end{equation}

			Thus the competitive ratio of $\mA$ is at least the maximum of the ratios in \refeq{eq:simple-lower-cr1} and \refeq{eq:simple-lower-cr2}. Thus we want to minimize this maximum and is achieved when the two ratios are equal to each other.
			
			Notice that the term $\alpha_{2, 1}$ does not appear in \refeq{eq:simple-lower-cr1} and \refeq{eq:simple-lower-cr2}. By setting the term in \refeq{eq:simple-lower-cr1} equal to the term in \refeq{eq:simple-lower-cr2} and re-arranging,
						\begin{equation}
				\label{eq:otherEq-simple-lower}
				\textstyle	\alpha_1 = 4 \alpha_{2, 2}.
			\end{equation}

			Moreover we have $\alpha_1 + \alpha_{2, 2} \leq B$. Combining this with \refeq{eq:otherEq-simple-lower} we get $\alpha_1 \leq \tfrac{4B}{5} = \tfrac{2T}{5}$ and the corresponding competitive ratio is at least $\Paren{ \tfrac{T}{4} }/\Paren{ \tfrac{\alpha_1}{2} } \geq \frac{5}{4}$.
			By Lemma~\ref{lem:optfull-optfd} with $d=1$, for every $j \in [2]$,
			\[
					\textstyle \OPTFD(\mI_j)/\E[\REW(\mA, \mI_j)] \geq \frac{5}{4} - O\Paren{\frac{\OPTLPfull}{T} \sqrt{\frac{\log T}{B}}}. \qedhere
			\]
\end{proof}

\subsection{The main lower bound: proof of Theorem~\ref{thm:LB-main}(b)}
\label{sec:LB-main}

To obtain the $\Omega(\log T)$ lower bound in Theorem~\ref{thm:LB-main}(b), we extend Construction~\ref{con:LB-simple} to one with $\Omega(\log T)$ phases rather than just two. As before, the algorithm needs to decide how much budget to save for the subsequent phases; without knowing whether they would bring higher rewards or nothing. The construction is as follows, see Figure~\ref{fig:lowerBound} for a pictorial representation:

\begin{construction}\label{con:LB-main}
There is one resource with budget $B$, and two arms, denoted $A_0,A_1$. Arm $A_0$ is the ``null arm'' that has zero reward and zero consumption. The consumption of arm $A_1$ is $1$ in all rounds. The rewards of $A_1$ are defined as follows. We partition the time into $\tfrac{T}{B}$ phases of duration $B$ each (for simplicity, assume that $B$ divides $T$). We consider $\tfrac{T}{B}$ problem instances; for each instance
    $\mI_\tau$, $\tau \in \Brac{\tfrac{T}{B}}$
arm $A_1$ has positive rewards up to and including phase $\tau$; after that all rewards are $0$. In each phase $\sigma\in [\tau]$, arm $A_1$ has reward $\sigma B/T$ in each round.
\end{construction}

	\begin{figure}[h]
    \centering
    \includegraphics[scale=1]{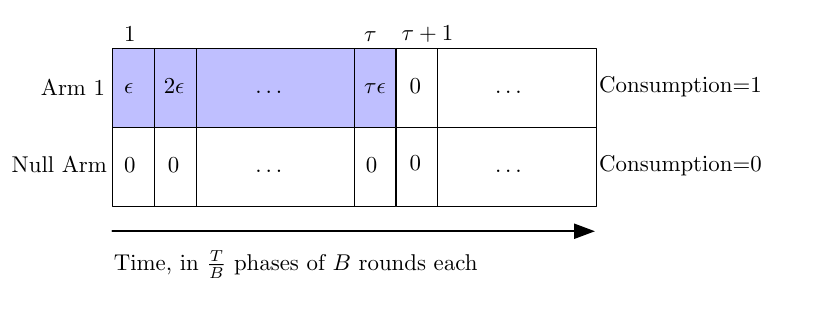}
    \caption{The lower-bounding construction for the $\Omega(\log T)$ lower bound. Here $\epsilon = \frac{B}{T}$.}
    \label{fig:lowerBound}
	\end{figure}

The lower bound holds for a wide range of budgets $B$, as expressed by the following lemma:

\begin{lemma} \label{lm:LB-main}
For any budget $B$ and any algorithm there is a problem instance in Construction~\ref{con:LB-main} such that
\begin{align}\label{eq:lm:LB-main}
    \frac{\OPTFD}{\E[\REW]}
        \geq \frac{1}{2} \cdot \ln ( \lfloor T/B \rfloor )
        + \zeta
        -  O\Paren{\frac{\log^{1.5} T}{\sqrt{B}}},
    \end{align}
where $\zeta = 0.577...$ is the Euler-Masceroni constant, and $\OPTFD\geq B^2/T$.
\end{lemma}



In the rest of this subsection we prove Lemma~\ref{lm:LB-main}.
Fix any randomized algorithm $\mA$. As before in this sub-section we are interested in the ratio $\OPTFD/\mathbb{E}[\REW(\mA)]$. We argue that it has the claimed competitive ratio on at least one instance $\mI_\tau$ in the construction~\ref{con:LB-main}. The proof proceeds in two parts. We first argue about the solution structure of the optimal distribution for the construction~\ref{con:LB-main} (we prove this formally in Lemma~\ref{lem:optimalSolutionLB}). Next we characterize the expected number of times arm $A_1$ is played if $\mA$  optimal algorithm in each of the phases. Combining the two we get Lemma~\ref{lm:LB-main}.

\xhdr{Structure of the optimal solution.} Define \(\OPTLP(\bvM_{\tau^*}, B, \tau^*)\) to be the optimal value of $\LP$~\ref{lp:primalAbstract} on the instance \(\mI_\tau\). Then we have the following Lemma.

	\begin{lemma}
		\label{lem:optimalSolutionLB}
		For a given instance $\mI_{\tau}$ we have \(\OPTLP(\bvM_{\tau^*}, B, \tau^*) = \frac{\epsilon B (\tau + 1)}{2} \) .
	\end{lemma}
	
	\begin{proof}
	Let $\mP(t)$ denote the non-zero reward on arm $A_1$ at time-step $t$ (\ie $\mP(t) = \lceil \tfrac{t}{B} \rceil \epsilon$). It suffices to prove that the optimal stopping time $\tau^* = B \tau$. Indeed, given that the stopping time is $B\tau$, the optimal solution is to set $X(1) = \frac{1}{\tau}$ and $X(0) = 1-\frac{1}{\tau}$ thus obtaining a total reward of $\frac{1}{\tau} \sum_{t \in [B \tau]} \mP(t) \epsilon$. From the definition of $\mP(t)$ we have that $\frac{1}{\tau} \sum_{t \in [B \tau]} \mP(t) \epsilon = \frac{1}{\tau} \sum_{j \in [\tau]} \epsilon B j$. Using the fact that $\sum_{j \in [\tau]} j = \frac{\tau (\tau+1)}{2}$ we get the statement of the Lemma. Thus it remains to prove that the optimal stopping time $\tau^* = B \tau$.
	
	First it is easy to prove that $\tau^* \leq B\tau$. Since there are no rewards after time-step $\tau^*$, we have
	\[
			\textstyle \forall t' > 0 \qquad  \OPTLP(\bvM_{\tau*+t'}, B, \tau^*+t') = \frac{1}{\tau+t'} \sum_{t \in [\tau^*]} \mP(t) \epsilon < \frac{1}{\tau} \sum_{t \in [B \tau]} \mP(t) \epsilon.
	\]
	
	Now we will argue that the optimal stopping time cannot be strictly lesser than $\tau^*$. To do so, first we argue that for two stopping times $t_1 < t_2$ within the same phase, the maximum objective is achieved for the stopping time $t_2$. This implies that the optimal stopping time has to be the last time step of some phase.
	
		Consider times $t_1 < t_2$ such that $\mP(t_1) = \mP(t_2) = \tau$. Then we want to claim that
		\[
			\textstyle \frac{B}{t_1} \Paren{\sum_{t \in [t_1]} \mP(t) \epsilon} \leq \frac{B}{t_2} \Paren{\sum_{t \in [t_2]} \mP(t) \epsilon}.
		\]	
		For contradiction assume the inequality does not hold. Then we have the following.
		
		\[
			\textstyle \sum_{t \in [t_1]} \mP(t) > \frac{t_1}{t_2} \Paren{\sum_{t \in [t_2]} \mP(t)}.
		\]
	
		Note that $\sum_{t \in [B(\tau-1)]}\mP(t) = \sum_{t' \in [\tau]} B t' = \frac{B(\tau-1)\tau}{2}$. Thus we have
		\begin{align*}
			\textstyle \sum_{t \in [t_1]} \mP(t) & \textstyle = \frac{B(\tau-1)\tau}{2} + (t_1-B(\tau-1))\tau,\\
			\textstyle \sum_{t \in [t_2]} \mP(t) & \textstyle = \frac{B(\tau-1)\tau}{2} + (t_2-B(\tau-1))\tau.
		\end{align*}
			
		Therefore we have,
		\[
			\textstyle \frac{B(\tau-1)\tau}{2} + (t_1-B(\tau-1))\tau > \frac{t_1B(\tau-1)\tau}{2t_2} + \tfrac{t_1}{t_2} \cdot (t_2-B(\tau-1))\tau.
		\]
		
		Further re-arranging, we get $B(\tau-1) > t_1$. This is a contradiction since $t_1$ is in phase $\tau$, so $t_1\geq B(\tau-1)$.
	
		Next we argue that the optimal value is achieved when the stopping time is in the last \emph{non-zero rewards phase}. Consider two phases $\tau_1 < \tau_2$. Thus the ending times are $B \tau_1$ and $B \tau_2$. To prove that the optimal value increases by stopping at $B \tau_2$, as opposed to $B \tau_1$, we want to show that
		
		\[
				\textstyle \frac{1}{\tau_1} \sum_{t \in [\tau_1]} B t \epsilon \leq \frac{1}{\tau_2} \sum_{t \in [\tau_2]} B t \epsilon.
		\]
		
		As before assume for a contradiction that this doesn't hold. Then re-arranging we get,
		$\frac{\tau_1(\tau_1 + 1)}{2} > \frac{\tau_1 (\tau_2 + 1)}{2}$,
which implies $\tau_1 > \tau_2$, contradiction. We conclude that the stopping time is $\tau^* = B\tau$.
	\end{proof}

\xhdr{Expected behavior of the optimal algorithm.} Consider any randomized algorithm $\mA$. The performance of $\mA$ is then as follows. From the definition of $\OPTLPfull$ we have,
	\begin{equation}
		\label{eq:lowerBoundMain}
		\frac{\OPTLPfull}{\E[\REW(\mA)]} = \max_{1 \leq \tau \leq T/B} \frac{B \tau \cdot \OPTLP(\bvM_{B\tau}, B, B \tau)}{\E[\REW(\mA)]}.
	\end{equation}
	
	
	We will now show that for any algorithm $\mA$, there exists an instance $j \in \Brac{\frac{T}{B}}$,
	\begin{equation}
		\label{eq:MainEqLBMain}
				\frac{\OPTLPfull(\mI_j)}{\E[\REW(\mA, \mI_j)]} \geq \Omega(\log T).
	\end{equation}
	
	Consider two consecutive instances \(\mI_{\tau}\) and \(\mI_{\tau+1}\). The outcome matrices in the phases $1, 2 \LDOTS \tau$ look identical in both these instances. This implies that any randomized algorithm cannot distinguish the two instances (in expectation). Thus, the expected number of times arm $A_1$ is chosen by algorithm $\mA$ in phases $1, 2 \LDOTS \tau$ is identical. Let \(\alpha_{\tau}\) denote the expected number of times $\mA$ plays arm \(A_1\) in phase \(\tau\). Note that this is the same for all instances $\mI_{\tau}, \mI_{\tau+1}, \LDOTS \mI_{T/B}$, as just argued. Thus, we can write
	\begin{equation}
		\label{eq:ExpressionREW}
		\E[\REW(\mA, \mI_{\tau})] = \sum_{j \in [\tau]} j \epsilon \alpha_j.
	\end{equation}
	
		Note that the expected number of times arm \(A_1\) is played in phase \(\tau\) on instances $\mI_{1}, \mI_{2} \LDOTS \mI_{\tau-1}$ does not appear in this expression and thus, is irrelevant for our purposes. Additionally, WLOG we only consider algorithms that exhaust its budget $B$. Indeed, an algorithm can instead choose only arm $A_1$ when the number of steps remaining equals its residual budget, without any degradation in the total reward.
 	Combining Eq.~\eqref{eq:ExpressionREW} with Lemma~\ref{lem:optimalSolutionLB}, the LHS in Eq.~\eqref{eq:lowerBoundMain} can be lower-bounded by,

		\begin{equation}
			\label{eq:MainLBOPTP3}	
				\frac{\OPTLPfull}{\E[\REW(\mA)]} \geq
				\frac{\epsilon B}{2} \cdot \left( \min_{\substack{\vec{\alpha} \geq \vec{0}:\\ \langle \vec{\alpha}, \vec{1} \rangle = B}} \max_{1 \leq \tau \leq T/B} \frac{\tau + 1}{\sum_{j \in [\tau]} j \epsilon \alpha_j} \right).
		\end{equation}
		
We can characterize the optimal solution $\vec{\alpha}$ in Eq.~\eqref{eq:MainLBOPTP3} as follows. Since the objective is a minimum over $\tfrac{T}{B}$ convex functions with a single equality constraint on the sum of the variables, from complementary slackness condition the minimum is attained when
			\begin{equation}
			\label{eq:EqualLowerBound}
\text{for each $\tau \in [T/B]$, the expression
    $\textstyle \left( \sum_{j \in [\tau]} j \alpha_j \right) \cdot \frac{1}{\tau +1}$
is the same .}
	\end{equation}
			
We will now prove that Eq.~\eqref{eq:EqualLowerBound} leads to the following recurrence for the maximizing values of
\(\alpha_j\).
\begin{equation}
	\label{eq:lowerBoundRecurrence}	
	\forall j \geq 2 \qquad \alpha_j = \frac{\alpha_1}{2j}.
\end{equation}

We will prove the recurrence \refeq{eq:lowerBoundRecurrence} via induction.
The base case is when \(j=2\). By \refeq{eq:EqualLowerBound},
\[
	\frac{1}{\alpha_1} = \frac{3/2}{\alpha_1 + 2\alpha_2},
\]
which implies that \(\alpha_2 = \frac{1}{4}\alpha_1\), and we are done. Now consider the inductive case; let all \(\alpha\) up to \(\alpha_{\tau}\) satisfy the recurrence \refeq{eq:lowerBoundRecurrence}. Consider the instance $\mI_{\tau}$ and $\mI_{\tau+1}$. From \refeq{eq:EqualLowerBound} we have,
\[
	 \frac{\alpha_1 + \sum_{j=2}^\tau \alpha_1/2}{\tau + 1} =  \frac{\alpha_1 + \sum_{j=2}^{\tau} \alpha_1/2 + (\tau+1) \alpha_{\tau+1}}{\tau+2}.
\]

\noindent Rearranging,
\(\alpha_{\tau+1} = \tfrac{1}{2(\tau+1)} \alpha_1\). This completes the inductive step, and the proof of \refeq{eq:lowerBoundRecurrence}.

We complete the proof of the lemma as follows. As argued in Eq.~\eqref{eq:lowerBoundRecurrence}, for the minimum value of $\{ \alpha_j \}_{j \in [T/B]}$, the expression $\frac{\epsilon B}{2} \cdot \frac{\tau +1}{\sum_{j \in [\tau]} j \epsilon \alpha_j}$, which is the RHS in Eq.~\eqref{eq:MainLBOPTP3}, is the same for all $\tau$ and in particular for $\tau = 1$. Substituting $\tau=1$, this evaluates to  $B/\alpha_1$. Since
\(\alpha_1 (1 + 1/4 + 1/6 + \ldots + B/2T) \leq B\) it follows that
\(\alpha_1 \leq 2B / H(\tfrac{T}{B}) \), where $H(n)$ denotes the $n^{th}$ Harmonic number. So, the right-hand side of \refeq{eq:lowerBoundMain} is at least $\tfrac12 H(\tfrac{T}{B})$.
Finally, $H(n) \geq \ln(n) +\zeta$, where $\zeta = 0.577...$ is the Euler-Masceroni constant. Combining this with Lemma~\ref{lem:optfull-optfd} we obtain Eq.~\eqref{eq:lm:LB-main}.

\subsection{Best dynamic policy: proof of Theorem~\ref{thm:LB-main}(c) and Theorem~\ref{thm:LB-CB}}
\label{sec:LB-DP}

	Consider the following construction of the lower-bound example.
		
	\begin{construction}\label{con:LB-DP}
		There is one resource with budget $B$, and two arms, denoted $A_0, A_1$. Arm $A_0$ is the `null arm' that has zero reward and zero consumption. The consumption of arm $A_1$ is $1$ in all rounds. The rewards of $A_1$ are defined as follows. We partition the time into $\tfrac{T}{B}$ phases of duration $B$ each (for simplicity, assume that $B$ divides $T$). We consider $\tfrac{T}{B}$ problem instances; for each instance $\mI_\tau$, $\tau \in \Brac{T/B}$ arm $A_1$ has $0$ reward
			in all phases except phase $\tau$; in phase $\tau$ it has a reward of $1$ in each round.
	\end{construction}
	
\begin{lemma}\label{lm:dynamicPolicy}
Consider Construction~\ref{con:LB-DP} with any given time horizon $T \geq 2$ and budget $B \leq \sqrt{T}$. Let \ALG be an arbitrary randomized algorithm for BwK. Then for one of the problem instances,
\begin{align}\label{eq:lm:dynamicPolicy}
    \OPTDP/\E[\REW] \geq T/B^2.
\end{align}
\end{lemma}

	\begin{proof}
Let $n = \tfrac{T}{B}$ be the number of phases in Construction~\ref{con:LB-DP}. Let \ALG be a deterministic algorithm. Let $\REW$ denote its total reward, and let $\E_\tau[\cdot]$ denote the expectation over the uniform-at-random choice of the problem instance $\mI_\tau$. We claim that
\begin{align}\label{eq:lm:dynamicPolicy:pf}
    \textstyle \OPTDP/\E_\tau[\REW] \geq T/B^2.
\end{align}
Assume that \ALG maximizes $\E_\tau[\REW]$ (over deterministic algorithms). Then it satisfies the following:
\begin{itemize}
\item  Within each phase, if \ALG ever chooses to play arm $A_1$, it does so in the first round of the phase. If it receives a reward of $1$ in this round, it plays $A_1$ for the rest of the phase. Else, it never plays $A_1$ for the rest of this phase.
		\end{itemize}
		
For each $\tau\in [n]$, let $\alpha_{\tau}$ denote the number of times \ALG chooses arm $A_1$ in phase $\tau$ in problem instance $\mI_{\tau}$.
The expected reward of \ALG over the uniform-at-random choice of the problem instance $\mI_\tau$ is
    $ \E[\REW] = \tfrac{1}{n} \sum_{i\in [n]} \alpha_i$.
Let
    $(\alpha_{\pi(1)}, \alpha_{\pi(2)} \LDOTS \alpha_{\pi(k)})$
be the subsequence of $(\alpha_1 \LDOTS \alpha_n)$ which contains all elements with non-zero values.

The key observation is as follows. The problem instances $\mI_{\pi(\tau-1)}$ and $\mI_{\pi(\tau)}$ are identical until phase $\pi(\tau-1)-1$. Since the feedback received by \ALG until the first time it chooses arm $A_1$ in phase $\pi(\tau-1)$ is identical, it follows that $\alpha_{\pi(\tau-1)} - \alpha_{\pi(\tau)} = 1$.
Therefore,		
\[ \textstyle
    \sum_{i\in[n]} \alpha_i
    = \sum_{i\in[k]} \alpha_{\pi(i)}
	= k\cdot \alpha_{\pi(1)} - \frac{k (k-1)}{2}.
\]
Noting that $\alpha_1\leq B$ and $k\leq \min(B,n) = B$ , we have:
\[
\textstyle
\E[\REW]
    \leq \tfrac{1}{n} \sum_{i\in[n]} \alpha_i
    < B^2/n = B^3/T.
\]
Since $\OPTDP=B$ for every problem instance $\mI_\tau$, \refeq{eq:lm:dynamicPolicy:pf} holds for $\ALG$, and therefore for any other deterministic algorithm. By Yao's lemma, for every randomized algorithm \ALG there exists a problem instance $\mI_\tau$ such that \eqref{eq:lm:dynamicPolicy} holds.
	\end{proof}

We now use the same construction to prove Theorem~\ref{thm:LB-CB}.
\begin{proof}[Proof sketch of Theorem~\ref{thm:LB-CB}]
		We prove the Theorem by contradiction. Let $B \leq \sqrt{T}$. For contradiction, consider an algorithm $\ALG$ for cBwK on a policy set $\Pi$ such that $\OPTFD(\Pi)/\E[\REW(\ALG)] < T/B^2$. We will now use $\ALG$ to construct an algorithm $\mA$ for the Construction~\ref{con:LB-DP} such that $\OPTDP/\E[\REW(\mA)] < T/B^2$ for every instance. This contradicts Lemma~\ref{lm:dynamicPolicy}.
		
		Consider a policy set $\Pi$ with $|n|$ policies. Every policy $\pi \in \Pi$ maps contexts in the range $[1, T]$ to the action set $\{A_1, A_0\}$. In particular, a policy $\pi_{\tau} \in \Pi$ maps contexts that lie in the range $[B*(\tau-1)+1, B*\tau]$ to arm $A_1$ and all other contexts to $A_0$. $\mA$ invokes $\ALG$ as a sub-routine with the policy set $\Pi$. At each time-step $t$, $\mA$ gives the context $x_t =t$ to $\ALG$ and plays the arm chosen by $\ALG$.
		
		For an instance $\mI_{\tau}$ in Construction~\ref{con:LB-DP}, $\OPTFD(\Pi)$ is the total reward obtained by choosing the action given by $\pi_{\tau}$ in all time-steps. The total reward obtained is $B$, which equals $\OPTDP(\mI_{\tau})$. Therefore, $\OPTFD(\Pi)/\E[\REW(\ALG)] < T/B^2$ implies we have $\OPTDP/\E[\REW(\mA)] < T/B^2$ for every instance $\mI_{\tau}$, which is a contradiction.
	\end{proof}

 \subsection{Best fixed arm: proof of Theorem~\ref{thm:LB-main}(d)}
 \label{sec:LB-FA}

We use the following construction for the lower-bound.
	
	\begin{construction}\label{con:LB-FA}
		There is one resource with budget $B$, and $K$ arms denoted by $A_1, A_2 \LDOTS A_K$. Arm $A_K$ is the `null arm' that has zero reward and zero consumption. There are $K$ instances in the family. In each instance, the time-steps are divided into $T/K$ equally spaced phases. In instance $\mI_j$, all arms $A_{j'}$ where $j' > j$ have $0$ reward and $0$ consumption in all time-steps. Consider an instance $\mI_j$ for some $j \in [K-1]$ and an arm $j' \leq j$. Arm $A_{j'}$ has a reward of $\tfrac{1}{K^{K-j'}}$ and consumption of $1$ in all time-steps in phase $j'$ and has a reward of $0$ and consumption of $0$ in every other time-step. Thus the rewards and consumption are bounded in the interval $[0, 1]$ for every arm and every time-step in all instances in this family. 		
	\end{construction}
	
	\begin{lemma}
		\label{thm:bestFixedArmLB}
		Let $T \geq 2$, $2 \leq B \leq T$, $K \geq 3$ be given parameters of the $\AdversarialBwK$ problem. We show that there exists a family of instances with $d=1$ shared resource such that for every randomized algorithm $\mA$ we have $\tfrac{\OPTFA}{\E[\REW(\mA)]}$ is at least $\Omega(K)$ on one of these instances.
	\end{lemma}
	
	\begin{proof}
		First note that the best fixed arm in instance $\mI_j$ is to pick arm $A_j$ which yields a total reward of $\frac{B}{K^{K-j}}$.
		
		Consider a randomized algorithm $\mA$. Observe that in the first $j$ phases, the instances $\mI_{j-1}$ and $\mI_j$ have identical outcome matrices. Thus the expected number of times any arm $A_k$ for $k \in [K]$ is chosen in phases $\Set{1, 2 \LDOTS j}$ should be the same in both the instances. Let $\alpha_k$ denote the expected number of times arm $k$ is played by $\mA$ in phase $k$ on instances $\mI_k, \mI_{k+1} \LDOTS \mI_{K-1}$ \footnote{This has to be the same in all instances since the outcome matrix is identical until phase $k$ in all these instances}. Moreover we have that $\alpha_1 + \alpha_2 \LDOTS \alpha_{K-1} \leq B$.
		
		To show the lower-bound we want to minimize the competitive ratio on every instance for all possible values of $\alpha_1, \alpha_2 \LDOTS \alpha_{K-1}$. For ease of notation denote $r_j := \frac{1}{K^{K-j}}$. Let $\alpha_{\mathcal{B}}$ denote the set of values to $\Set{\alpha_k}_{k \in [K-1]}$ such that $\sum_{k \in [K-1]} \alpha_k \leq B$. Thus,
				\begin{equation}
			\label{eq:bestFixedArmRatio}
			\textstyle	\tfrac{\OPTFA}{\E[\REW(\mA)]} \geq \min_{\alpha_{\mathcal{B}}} \frac{r_k B}{\sum_{j \in [k]} r_j \alpha_j}.
		\end{equation}
		
		The ratio is minimized when all ratios in \refeq{eq:bestFixedArmRatio} are equal. We will show via induction that this yields the following recurrence,
		\begin{equation}
			\label{eq:FArecurrence}
			\textstyle \forall k \geq 2 \qquad \alpha_k = \Paren{1-\frac{r_{k-1}}{r_k}} \alpha_1.
		\end{equation}
		
		Combining this with the condition that $\sum_{k \in [K-1]} \alpha_k \leq B$, this yields the condition $\alpha_1 \leq \tfrac{B}{K-\tfrac{1}{K}}$. Moreover the minimizing value in \refeq{eq:bestFixedArmRatio} is $K-\tfrac{1}{K}$ which proves Lemma~\ref{thm:bestFixedArmLB}.
		
		We will now prove the recurrence \refeq{eq:FArecurrence}. Consider the base case with $k=2$. Equalizing the first two terms in \refeq{eq:bestFixedArmRatio} we get
		\[
			\textstyle	\frac{r_1 B}{r_1\alpha_1} = \frac{r_2 B}{r_1 \alpha_1 + r_2 \alpha_2}.
		\]
		
		Re-arranging we obtain that $\alpha_2 = \Paren{1- \frac{r_1}{r_2}} \alpha_1$. We will now prove the inductive case. Let the recurrence be true for all $1 \leq k \leq k'$. Consider the case $k=k'+1$. Setting the $k'$ and $k'+1$ ratios in \refeq{eq:bestFixedArmRatio} equal, we obtain
		\begin{equation}
			\label{eq:FAinductive}
			\textstyle \frac{r_{k'} B}{\sum_{j \in [k']} r_j \alpha_j} = \frac{r_{k'+1} B}{\sum_{j \in [k' + 1]} r_j \alpha_j}.
		\end{equation}

		Moreover from the inductive hypothesis we have $\alpha_j = \Paren{1-\tfrac{r_{j-1}}{r_j}} \alpha_1$ for every $j \leq k'$. Thus we have
		\begin{align*}
				\textstyle \sum_{j \in [k']} r_j \alpha_j & = r_{k'} \alpha_1\\
				\textstyle \sum_{j \in [k' + 1]} r_j \alpha_j &= r_{k'} \alpha_1 + r_{k'+1} \alpha_{k'+1}.
		\end{align*}
	
		Plugging this back in \refeq{eq:FAinductive} we get
		\[
				\textstyle \frac{r_{k'} B}{r_{k'} \alpha_1} = \frac{r_{k'+1} B}{r_{k'} \alpha_1 + r_{k'+1} \alpha_{k'+1}}.
		\]
		
		Rearranging we get $\alpha_{k'+1} = \Paren{1-\tfrac{r_k'}{r_{k'+1}}} \alpha_1$. This completes the induction.
	\end{proof}

\section{Open Questions and Follow-Up Work}
\label{sec:conclusions}

We use essentially the same algorithm, \MainALG, to solve both stochastic and adversarial version of bandits with knapsacks. Yet, we use it with different parameter $T_0$ (randomly guessed in the adversarial version) and a slightly different definition of the outcome matrices.%
\footnote{Recall that in the stochastic setting there a `dummy resource' with strictly positive consumption for all arms, whereas in the adversarial version the null arm must have zero consumption for all resources.}
Can we solve both versions with \emph{exactly} the same algorithm? One  concrete goal would be to achieve $O(\log T)$ competitive ratio in the adversarial version, and $o(T)$ regret for the stochastic version in the regime $\min(B,\OPTFD)\geq \Omega(T)$. A similar ``best of both worlds" result has been obtained for bandits without budget/supply constraints: one algorithm that achieves optimal regret rates for both adversarial bandits and stochastic bandits, without knowing which environment it is in \citep{BestofBoth-colt12,BestofBoth-icml14,Auer-colt16}. Further developments focused on mostly stochastic environments with a small amount of adversarial behavior \citep{BestofBoth-icml14,Seldin-colt17,Thodoris-stoc18,Haipeng-colt18}; similar questions are relevant to BwK as well.


\OMIT{ 
	Two more extensions would follow from plausible advances in other areas of multi-armed bandits. First, consider the contextual version of BwK (cBwK), as defined in Section~\ref{sec:ext-CB}. Our framework yields algorithms and regret bounds for cBwK, by using a suitable ``primal algorithm" in \MainALG. Namely, we need an algorithm for adversarial contextual bandits (see Figure~\ref{prob:CB-adv} on page \pageref{prob:CB-adv}), with high-probability regret bounds. We use one such algorithm to derive Corollary~\ref{cor:CB}, but it is computationally inefficient. Computationally efficient (more specifically, oracle-efficient) algorithms for adversarial contextual bandits exist
\cite{Syrgkanis-AdvCB-icml16,Rakhlin-AdvCB-icml16,Syrgkanis-AdvCB-nips16},
but only with expected regret bounds. Extending these results to high-probability regret bounds would imply oracle-efficient algorithms for cBwK.
Second, consider BwK problems with large action-dependent reward ranges, \eg the dynamic pricing problem with price range $[0, p_{\max}]$. (In this problem, actions are prices, and each price $p$ yields reward $0$ or $p$.) The current results, both ours and in the prior work on \StochasticBwK, give regret terms that are proportional to $p_{\max}$. A recent advance in multi-armed bandits replaces this dependence on $p_{\max}$ with a similar dependence on $p^*$, the best price \cite{MultiScaleMAB-ec17}. We could use the algorithm from \cite{MultiScaleMAB-ec17} as the primal algorithm in \MainALG, and obtain (with a little extra work) a similar improvement for BwK, if the regret bounds \cite{MultiScaleMAB-ec17} were extended from expected regret to high-probability regret.
} 

\OMIT{
We interpret the adversarial analysis of \MainALG (Lemma~\ref{lm:adv-crux}) as an extension of the ``regret minimization in games" framework to adversarial repeated games. Indeed, we apply this framework and derive a convergence result for rewards
}

Given our upper and lower bounds, the competitive ratio $\frac{\OPTFD-\reg}{\E[\REW]}$ can potentially be improved in several regimes. Some concrete questions left open by our paper are as follows:
\begin{itemize}
\item obtain \emph{constant} competitive ratio in the regime $B=\Omega(T)$.
\item obtain \emph{sublinear} dependence of the competitive ratio on $d$, the number of resources.
\item obtain \emph{constant} competitive ratio for problem instances with ``large enough" $\OPTFD$.
\item obtain \emph{optimal} constant competitive ratio when $\OPTFD$ is known up to a constant factor.
\item obtain competitive ratio in Theorem~\ref{thm:AdvBwK-main} \emph{uniformly} over $\Gmin$ (\ie for all $\Gmin$ simultaneously). Equivalently, obtain  competitive ratio $O(\ln(T/\OPTFD))$.
\end{itemize}

Several of these questions have been resolved in follow-up work. \cite{Singla-colt20} resolve the optimal dependence on $d$, achieving competitive ratio $O(\log(d)\log(T))$ via a more careful analysis of \MainALG (among other results), and prove a matching lower bound. \citet{Castiglioni-icml22} use a version of \MainALG to obtain $\nicefrac{T}{B}$ competitive ratio, which is a mere constant when $B=\Omega(T)$.
Interestingly, they use fixed parameters $(B_0,T_0) =(B,T)$, without the random guessing in Algorithm~\ref{alg:LagrangianBwKAdv} or the multi-phased ``meta-algorithm" in Algorithm~\ref{alg:LagrangianBwKAdvHighP}; moreover, their result holds with high probability. \citet{Castiglioni-icml22} also analyze the very same algorithm in the stochastic setting, matching our regret bound from Theorem~\ref{thm:IID} and therefore achieving a ``best of both worlds" result. Their version of \MainALG optimizes the dual vector $\vec{\lambda}\in [0,\nicefrac{T}{B}]^d$, whereas ours optimizes $\vec{\lambda}$ over all distributions.

Can one still achieve meaningful regret bounds for \AdvBwK, without the competitive ratio? One way to interpret our impossibility results is that the fixed-distribution benchmark is just too harsh. It could be productive to define a weaker (and perhaps fairer) benchmark for the algorithm to compete against, so as to achieve competitive ratio of $1$ relative to this benchmark. \citet{Gaitonde-auctions-arxiv22} achieve one such result for the special case of budget-constrained bidding in a repeated auction.


In terms of extensions to ``richer" application scenarios, as in Section~\ref{sec:ext}, \citet{Castiglioni-icml22} spell out two more extensions: to repeated Stackelberg games and to repeated first-price auctions. The main open question is to achieve similar results using a ``stochastic" primal algorithm, \ie a primal algorithm designed (only) for the stochastic version of a particular application scenario.

\OMIT{ 
Given the importance of online learning in repeated stochastic  zero-sum games, online learning in repeated \emph{adversarial} zero-sum games is of independent interest. While we made use of it in the context of \AdvBwK, it is tantalizing to study it in more general contexts, and use it in a fruitful way in other problems. After the initial publication of this paper on {\tt arxiv.org}, \citet{Jake-icml19} obtained a powerful result in this direction, as discussed in Section~\ref{sec:related-simultaneous}.
} 


\addcontentsline{toc}{section}{References}

\begin{small}
\bibliographystyle{plainnat}
\bibliography{bib-abbrv,bib-AGT,bib-bandits,bib-ML,bib-random,bib-slivkins,refs}
\end{small}

\newpage

\appendix
\section{Standard tools}
\label{app:LP}

Our exposition in the body of the paper relies on some tools that are either known or can easily be derived using standard techniques. We state (and sometimes derive) these tools in this appendix.

\subsection{Concentration Inequalities}
\label{sec:techLemmas}

	\begin{lemma}[Azuma-Hoeffding inequality]
			\label{lem:AzumaHoeffding}
			Let $Y_1, Y_2, \ldots, Y_T$ be a martingale difference sequence (\emph{i.e.,} $\mathbb{E}[Y_t ~|~Y_1, Y_2 \LDOTS Y_{t-1}] = 0$). Suppose $|Y_t| \leq c$ for all $t \in \{ 1, 2, \ldots, T\}$. Let $R_{0, \delta}(T) := \sqrt{2 T c^2 \ln(1/\delta)}$. Then for every $\delta > 0$,
			
			\[
					\textstyle \Pr \left[ \sum_{t \in [T]} Y_t > R_{0, \delta}(T) \right] \leq \delta.
			\]	
	\end{lemma}
			
	\begin{lemma}[Chernoff-Hoeffding bounds]
			\label{lem:Chernoff}
			Let $X_1, X_2, \ldots, X_T$ be a sequence of independent random variables such that $|X_t| \leq c$ for all $t \in \{ 1, 2, \ldots, T\}$. Let $Z _t:= \E[X_t]$. Then for every $\delta > 0$,
			
			\[
					\textstyle \Pr \left[ \left| \sum_{t \in [T]} X_t - Z_t \right| > 3 \sqrt{ \Paren{\sum_{t \in [T]}Z_t} c^2 \ln(1/\delta)} \right] \leq \delta.
			\]		
		\end{lemma}

\subsection{Lagrangians: proof of Lemma~\ref{lm:gameToLP}}
	\label{app:eq:LagrangeMinMaxNew}
	
Assume one of the resources is the dummy resource, and one of the arms is the null arm. Consider the linear program $\myLP{\vM}{B}{T}$, for some outcome matrix $\vM$. Let
    $\mL = \myLag{\vM}{B}{T}$
denote the Lagrange function.

\begin{lemma}[Lemma~\ref{lm:gameToLP}, restated]\label{lem:app-gameToLP-Nash}
Suppose $(\vec{X}^*, \vec{\lambda}^*)$ is a mixed Nash equilibrium for the Lagrangian game. Then $\vec{X}^*$ is an optimal solution for the linear program \eqref{lp:primalAbstract}. Moreover, the minimax value of the Lagrangian game equals the LP value:
    $\mL(\vec{X}^*, \vec{\lambda}^*) = \OPTLP$.
\end{lemma}

In what follows we prove Lemma~\ref{lem:app-gameToLP-Nash}. Writing out the definition of the mixed Nash equilibrium,
	\begin{equation}
		\label{eq:Saddle}
		 \mL(\vec{X}^*, \vec{\lambda}) \geq \mL(\vec{X}^*, \vec{\lambda}^*) \geq  \mL(\vec{X}, \vec{\lambda}^*) \qquad \forall \vec{X} \in \Delta_K, \vec{\lambda} \in \Delta_d.
	\end{equation}

For brevity, denote
    $r(\vec{X}^*) = \sum_{a\in[K]} \vec{X}^*(a)\; r(a)$
and
    $c_i(\vec{X}^*) = \sum_{a\in[K]} \vec{X}^*(a)\; c_i(a)$.

We first state and prove the complementary slackness condition for the Nash equilibrium.

\begin{claim}
	\label{cl:sadHelpLemma}
For every resource $i \in [d]$ we have,
	\begin{itemize}
		\item[(a)] $1-\tfrac{T}{B}\, c_i(\vec{X}^*) \geq 0$,
		\item[(b)] $\lambda_i^* > 0 \implies 1-\tfrac{T}{B}\,c_i(\vec{X}^*) = 0$.
	\end{itemize}
\end{claim}
\begin{proof}
{\bf Part (a).}
For contradiction, consider resource $i$ that minimizes the left-hand side in (a), and assume that the said left-hand side is strictly negative.  We have two cases: either $\lambda_{i}^* < 1$ or $\lambda_{i}^*=1$. When $\lambda_i^* < 1$, consider another distribution $\tilde{\vec{\lambda}} \in \Delta_d$  such that $\tilde{\lambda}_i = 1$ and $\tilde{\lambda}_{i'}=0$ for every $i' \neq i$. Note that we have, $\mL(\vec{X}^*, \tilde{\vec{\lambda}}) < \mL(\vec{X}^*, \vec{\lambda}^*)$. This contradicts the first inequality in \eqref{eq:Saddle}.
	
Consider the second case, when $\lambda^*_i = 1$. Then $\mL(\vec{X}^*, \vec{\lambda}^*)
    = r(\vec{X}^*) + 1 - \frac{T}{B}c_i(\vec{X}^*)$.
 Consider any arm $a \in [K]$ such that $X^*(a) \neq 0$. Let $\tilde{\vec{X}}\in \Delta_K$ be another distribution such that $\tilde{X}(a) := 0$ and $\tilde{X}(\nop) := X^*(\nop) + X^*(a)$ and $\tilde{X}(a') = X^*(a')$ for every $a' \not \in \{a, \nop\}$. Note that $\tilde{X}(\nop) \leq 1$. Since $(\vec{X}^*, \vec{\lambda}^*)$ is a Nash equilibrium, we have that $\mL(\tilde{\vec{X}}, \vec{\lambda}^*) \leq \mL(\vec{X}^*, \vec{\lambda}^*)$. This implies that $-X^*(a) r(a) + X^*(a) \frac{T}{B} c_i(a) \leq 0$. Re-arranging we obtain, $\frac{T}{B}c_i(a) \leq r(a) \leq 1$. Thus, we have $1-\frac{T}{B}c_i(a) \geq 0$.

Since this holds for every $a \in [K]$ with $X^*(a) \neq 0$,
we obtain a contradiction:
	\[
\textstyle
1-\tfrac{T}{B}\, c_i(\vec{X}^*)
= \sum_{a \in [K]} X^*(a) \left( 1-\frac{T}{B}c_i(a) \right) \geq 0.
	\]

{\bf Part (b).}
For contradiction, assume the statement is false for some resource $i$. Then, by part (a), $\lambda_i^* > 0$ and $1-\tfrac{T}{B} c_i(\vec{X}^*) > 0$, and consequently
    $\mL(\vec{X}^*, \vec{\lambda}^*) > r(\vec{X}^*)$.
Now, consider distribution $\tilde{\vec{\lambda}}$ which puts probability $1$ on the dummy resource. We then have $\mL(\vec{X}^*, \tilde{\vec{\lambda}}) = r(\vec{X}^*) < \mL(\vec{X}^*, \vec{\lambda}^*)$,
contradicting the first inequality in \refeq{eq:Saddle}.
\end{proof}

Let $\tilde{\vec{X}}$ be some feasible solution for the linear program \eqref{lp:primalAbstract}. Plugging the feasibility constraints into the definition of the Lagrangian function,
    $\mL(\tilde{\vec{X}}, \vec{\lambda}^*) \geq r(\tilde{\vec{X}})$.
Claim~\ref{cl:sadHelpLemma}(a) implies that $\vec{X}^*$ is a feasible solution to the linear program \eqref{lp:primalAbstract}. Claim~\ref{cl:sadHelpLemma}(b) implies that
    $\mL(\vec{X}^*, \vec{\lambda}^*) = r(\vec{X}^*)$.
Thus,
	\[
	 r(\vec{X}^*) = \mL(\vec{X}^*, \vec{\lambda}^*) \geq \mL(\tilde{\vec{X}}, \vec{\lambda}^*) \geq r(\tilde{\vec{X}}).
		\]
So, $\vec{X}^*$ is an optimal solution to the $\LP$. In particular,
    $\OPTLP = r(\vec{X}^*) = \mL(\vec{X}^*, \vec{\lambda}^*) $.

\subsection{The stopped LP for \AdvBwK: proof of \refeq{eq:adv-stoppedLP}}
	\label{sec:appxAdversarial}

The proof is similar to prior work~\cite{BwK-focs13,DevanurJSW-ec11}.
Denote $\mathcal{D}_{\tau}$ to be the set of all distributions over the arms such that for every $\vec{p} \in \mathcal{D}_{\tau}$ we have the following: for every $i \in [d]$ we have $\sum_{t \in [\tau]} \vec{p}\, \cdot\, \vec{c}_{t, i} \leq B$. In other words, $\mathcal{D}_\tau$ denotes the set of distributions whose expected stopping time is at least $\tau$. Thus it immediately follows that $\OPT_{\LP}(\tau) \geq \max_{\vec{p} \in \mathcal{D}_{\tau}} \sum_{t \in [\tau]} \vec{p}\, \cdot\, \vec{r}_t$ since for any given $\vec{p} \in \mathcal{D}_{\tau}$ it is feasible to $\LP(\tau)$. Thus $\OPT_{\LP}(\tau)$ is at least the value of any feasible solution $\vec{p} \in \mathcal{D}_{\tau}$. Note that for every fixed distribution $\vec{p} \in \Delta_K$, there exists a $\tau$ such that either $\vec{p} \in \mathcal{D}_\tau$ and $\vec{p} \not \in \mathcal{D}_{\tau+1}$ or $\vec{p} \in \mathcal{D}_T$. Moreover the total expected reward we can obtain using $\vec{p}$ is $\sum_{t \in [\tau]} \vec{p} \cdot \vec{r}_t$. Thus $\max_{1 \leq \tau \leq T} \OPT_{\LP}(\tau) \geq \OPTFD$.

\subsection{Regret minimization in games: proof of Lemma~\ref{lem:learningGames}}
\label{app:pf-games}

Let us revisit adversarial online learning, as per Figure~\ref{prob:adv}. Denote the benchmark in \refeq{eq:prelims-regret-weak} as
\[
\optAOL(T) := \textstyle \max_{a\in A} \sum_{t\in [T]} f_t(a).
\]
Recall that $[\rmin,\rmax]$ is the payoff range, and denote $\sigma = \rmax-\rmin$.
\begin{lemma}
\label{lem:standard-to-non-standard}
Suppose an algorithm for adversarial online learning satisfies \eqref{eq:prelims-regret-weak} for some $\delta>0$. Then
\begin{align}
\textstyle
\Pr\left[\; \forall \tau\in[T]\;
  		\optAOL(\tau)
  		 \;-\;  \sum_{t\in [\tau]} \vec{f}_t\cdot \vec{p}_t  \leq \sigma\cdot\,\left( R_{\delta/T}(T) + \sqrt{2T \log (T/\delta)} \right)
		\;\right] \geq 1-2\delta. \label{eq:appx:otherRegBound}
 \end{align}
\end{lemma}
\begin{proof}
Let us use the stronger regret bound \eqref{eq:prelims-regret-strong} implied by \eqref{eq:prelims-regret-weak}. Note that 	
    \[ \E[f_t(a_t)~|~a_1, a_2 \LDOTS a_{t-1}] = \vec{f}_t \cdot \vec{p}_t.\]
Applying the Azuma-Hoeffding inequality for each $\tau \in [T]$, and taking a union bound, we have
\begin{align}\label{eq:appxgenReg}
\Pr\left[\;
    \forall \tau \in [T] \quad \textstyle \sum_{t \in [\tau]} f_t(a_t) - \sum_{t \in [\tau]} \vec{f}_t \cdot \vec{p}_t \leq \sigma\cdot
    \sqrt{2T \log (T/\delta)}
\;\right] \geq 1-\delta.
\end{align}
Taking a union bound over \refeq{eq:appxgenReg} and \refeq{eq:prelims-regret-strong} and adding the equations we get \refeq{eq:appx:otherRegBound}.
\end{proof}

\begin{remark}
For Hedge algorithm, regret bound \refeq{eq:appx:otherRegBound} is already proved in \cite{FS97}.
\OMIT{More specifically,
	\refeq{eq:appx:otherRegBound} is proved in \cite{FS97}. Combining this with a direct application of the Azuma-Hoeffding inequality yields \refeq{eq:appxgenReg}.}
\end{remark}

Let
    $W= \sqrt{2 T \log (T/\delta)}$
denote the term from Lemma~\ref{lem:standard-to-non-standard} in what follows.

We now prove Lemma~\ref{lem:learningGames}, similar to the proof in \cite{freund1996game} for the deterministic game. Recall that we take averages up to some fixed round $\tau\in [T]$. We prove that the following two inequalities  hold, each with probability at least $1-\delta$.
\begin{align}
\frac{1}{\tau} \sum_{t \in [\tau]} \vec{p}_{t,1}^{\tran}\, \vec{G}_t\, \vec{p}_{t, 2} &  \geq v^* - \sigma\cdot \frac{R_{1,\, \delta/T}(T) +2W}{\tau}. & \label{eq:rewardGamesLB} \\
\frac{1}{\tau} \sum_{t \in [\tau]} \vec{p}_{t,1}^{\tran}\, \vec{G}_t\, \vec{p}_{t, 2} &  \leq \overline{\vec{p}}_{1}^{\tran}\, \vec{G}\, \vec{p}_2 + \sigma\cdot  \frac{R_{2,\, \delta/T}(T) + 2W}{\tau} &  \forall \vec{p}_2 \in \Delta_{A_2}. \label{eq:rewardGamesUB}
\end{align}
		
\noindent \refeq{eq:prelims-games-play} in Lemma~\ref{lem:learningGames} follows by adding \refeq{eq:rewardGamesLB} and \refeq{eq:rewardGamesUB}.

First we prove \refeq{eq:rewardGamesLB}. Following the set of inequalities in Section 2.5 of \cite{freund1996game} we have,
\begin{align*}
\frac{1}{\tau} \sum_{t \in [\tau]} \vec{p}_{t, 1}^{\tran} \vec{G}_t \vec{p}_{t, 2}
			& \geq_{whp}  \frac{1}{\tau} \sum_{t \in [\tau]} {\vec{p}_1^*}^{\tran}\, \vec{G}_t\, \vec{p}_{t, 2} - \sigma\cdot \frac{R_{1,\, \delta/T}(T)+W}{\tau} & \text{From Lemma~\ref{lem:standard-to-non-standard}} &\\
				& \geq_{whp} \frac{1}{\tau} \sum_{t \in [\tau]} {\vec{p}_1^*}^{\tran}\, \vec{G}\, \vec{p}_{t, 2} - \sigma\cdot \frac{R_{1,\, \delta/T}(T) + 2W}{\tau} & \text{From Lemma~\ref{lem:AzumaHoeffding}}& \\
				& = \max_{\vec{p}_1 \in \Delta_{A_1}} \frac{1}{\tau} \sum_{t \in [\tau]} {\vec{p}_1}^{\tran}\, \vec{G}\, \vec{p}_{t, 2} - \sigma\cdot \frac{R_{1,\, \delta/T}(T) + 2W}{\tau} & \text{From Definition of $\vec{p}_1^*$.} &\\
				& = \max_{\vec{p}_1 \in \Delta_{A_1}} {\vec{p}_1}^{\tran}\, \vec{G}\, \overline{\vec{p}}_2 - \sigma\cdot \frac{R_{1,\, \delta/T}(T) + 2W}{\tau}  & \text{From Definition of $\overline{\vec{p}}_2$.} &\\
				& \geq \min_{\vec{p}_2 \in \Delta_{A_2}}\;  \max_{\vec{p}_1 \in \Delta_{A_1}} \vec{p}_1^{\tran}\, \vec{G}\, \vec{p}_2 - \sigma\cdot \frac{R_{1,\, \delta/T}(T) + 2W}{\tau} &
\end{align*}
Here $\leq_{whp}$ denotes statements that hold with probability at least $1-\delta$.	

Now let us prove \eqref{eq:rewardGamesUB}. Fix distribution
    $\vec{p}_2 \in \Delta_{A_2}$.
Then:
\begin{align*}
\frac{1}{\tau} \sum_{t \in [\tau]} \vec{p}^{\tran}_{t, 1}\, \vec{G}_t\, \vec{p}_{t, 2}
			& \leq_{whp}\frac{1}{\tau} \sum_{t \in [\tau]} {\vec{p}_{t, 1}}^{\tran}\, \vec{G}_t\, \vec{p}_2 + \sigma\cdot \frac{R_{2,\, \delta/T}(T)+W}{\tau} & \text{From Lemma~\ref{lem:standard-to-non-standard}} \\
				& \leq_{whp} \frac{1}{\tau} \sum_{t \in [\tau]} {\vec{p}_{t, 1}}^{\tran}\, \vec{G}\, \vec{p}_2 + \sigma\cdot \frac{R_{2,\, \delta/T}(T) + 2W}{\tau} & \text{From Lemma~\ref{lem:AzumaHoeffding}}\\
				& = \overline{\vec{p}_1}^{\tran}\, \vec{G}\, \vec{p}_2 + \sigma\cdot \frac{R_{2,\, \delta/T}(T) + 2W}{\tau} & \text{From Definition of $\overline{\vec{p}_1}$.}
			\end{align*}
Taking a union bound over all the four high-probability inequalities, we get the lemma.

\newpage
\section{Table of notation}
\label{sec:notations}

For reference, let us summarize the important notation used across sections.

\begin{center}
\begin{tabular}{ |c|l| }
\hline
	Notation & Usage  \\
\hline
 $\OPTFD$ & Optimal value of the fixed distribution over arms in hindsight. \\ \hline
 $\OPTDP$ & Optimal dynamic policy in hindsight. \\ \hline
 $\REW$ & Total (random) reward obtained by the algorithm \\ \hline
$\vM$ & Outcome matrix; rewards and consumption for every arm; $\vec{o}$ used to represent a row. \\ \hline
$\bvM_{\tau}$ & Average of outcome matrices after $\tau$ time-steps. \\ \hline
$\bvM^{\ips}_{\tau}$ & Average of outcome matrices estimated using IPS estimates after $\tau$ time-steps. \\ \hline
$\vec{G}$ & Payoff matrix in the Lagrangian game \\ \hline
$R_{j, \delta}(\tau)$ & Regret of $\ALG_j$ with probability at least $1-\delta$ after $\tau$ rounds \\ \hline
$R_{0, \delta}(\tau)$ or $R_0(\tau)$ & Confidence term in the Azuma-Hoeffding inequality. \\ \hline
$U_j(T~|~T_0)$ & Regret of $\ALG_j$ after $T$ rounds given the parameter $T_0$ (this affects scaling of regret). \\ \hline
$\mL(.)$ & Lagrange function \\ \hline
$T_0$ & Parameter used in the Lagrangian. $T_0 = T$ in \StochasticBwK and $T_0 = \guess$ in \AdversarialBwK. \\ \hline
$B_0$ & Scaled budget. $B_0 = \tfrac{B}{\ratio}$ in \AdversarialBwK (high-probability). Otherwise $B_0 = B$. \\ \hline
$\guess$, $\Gmax$, $\Gmin$ & Guess, maximum and minimum range of this guess respectively in \AdversarialBwK. \\ \hline
$\kappa$ & Multiplicative factor with which guess is increased. \\ \hline
$\OPTLPfull[\tau]$ & Best objective of the $\tau$ stopped $\LP$s (\ie stopped at times $1, 2 \LDOTS \tau$). \\ \hline
$\LP_{\bvM_{\tau}, B, \tau}$ & Linear program corresponding to the average outcome matrix $\bvM_{\tau}$. \\ \hline
$\OPTLP(\bvM_{\tau}, B, \tau)$ & Optimal value of $\LP_{\bvM_{\tau}, B, \tau}$. \\ \hline
\end{tabular}
\end{center}

\end{document}